\documentclass[letterpaper,11pt]{article}

\usepackage{setspace}
\usepackage[utf8]{inputenc}
\usepackage[margin=1in]{geometry}
\usepackage{array}
\usepackage{ stmaryrd,amsxtra,amstext,amscd,amsfonts,fancyhdr}
\usepackage{ amssymb }
\usepackage{amsmath}
\usepackage{amsthm}
\usepackage{enumerate}

\usepackage{tikz}
\usetikzlibrary{calc}

\newcommand\del[1]{}

\allowdisplaybreaks

\makeatletter
\@namedef{subjclassname@2010}{%
  \textup{2010} Mathematics Subject Classification}
\makeatother

\newtheorem{theorem}{Theorem}[section]
\newtheorem{corollary}[theorem]{Corollary}
\newtheorem{proposition}[theorem]{Proposition}

\theoremstyle{definition}
\newtheorem{definition}[theorem]{Definition}
\newtheorem{lemma}[theorem]{Lemma}

\newtheorem{example}[theorem]{Example}
\newtheorem{remark}[theorem]{Remark}
\newtheorem{claim}[theorem]{Claim}

\numberwithin{equation}{section}

\newcommand\N{{\mathbb{N}}}
\newcommand\R{{\mathbb{R}}}

\newcommand\TT{{\mathbb{T}}}

\newcommand\pd{{\partial}}

\renewcommand\L{{\mathcal{L}}}
\newcommand{\hk}{\mathbin{\! \hbox{\vrule height0.3pt width5pt depth 0.2pt \vrule height5pt width0.4pt depth 0.2pt}}}

\newcommand\dt{\frac{d}{dt}}
\newcommand\pdx{\frac{\pd}{\pd x^i}}

\begin{document}


\title{Noether's Theorem}
\author{Nat Leung}
\begin{titlepage}
\begin{center}
\setlength{\parindent}{0pt}
\setlength{\parskip}{17pt}

\vspace*{10pt}

{\Large \bf{Noether's Theorem Under the Legendre Transform}\rm\\ 
\vspace{0.2cm}

\Large{by}
\vspace{-0.2cm}

\Large{Jonathan Herman}}

\vspace*{10pt}

{\small A research paper\\  
\vspace{0.1cm}
presented to the University of Waterloo\\
\vspace{0.1cm}in fulfilment of the\\
\vspace{0.1cm}research paper requirement for the degree of\\
\vspace{0.1cm}Master of Mathematics\\
\vspace{0.1cm}in\\
\vspace{0.1cm}Pure Mathematics

\vspace*{\fill}

Waterloo, Ontario, Canada, 2014

\copyright \  Jonathan Herman, 2014
}
\end{center}
\end{titlepage}

 \del{
\cleardoublepage

 \pagenumbering{roman}
 \setcounter{page}{2}
 
\vspace*{0.5in}

\noindent
{   \setlength{\parindent}{0pt}
   \setlength{\parskip}{24pt}
   \setlength{\textwidth}{7in}

   {\sffamily\bfseries \index{copyright!author's declaration}
   AUTHOR'S DECLARATION}

   I hereby declare that I am the sole author of this research paper.  This is a true
   copy of the research paper, including any required final revisions, as accepted by
   my examiners.

   I understand that my research paper may be made electronically available to the
   public. 
   
   \vspace*{1.0in}
   
Jonathan Herman
}

}

\newpage

\begin{abstract}

In this paper we demonstrate how the Legendre transform connects the statements of Noether's theorem in Hamiltonian and Lagrangian mechanics. We give precise definitions of symmetries and conserved quantities in both the Hamiltonian and Lagrangian frameworks and discuss why these notions in the Hamiltonian framework are somewhat less rigid. We explore conditions which, when put on these definitions, allow the Legendre transform to set up a one-to-one correspondence between them. We also discuss how to preserve this correspondence when the definitions of symmetries and conserved quantities are less restrictive.

\end{abstract}

\newpage


\onehalfspacing
\cleardoublepage
\del{
\section*{Acknowledgements}
First and foremost, a very sincere thank you goes to my supervisor Dr. Spiro Karigiannis. I am extremely grateful for his patience and the tremendous amount of time he spent teaching and helping me this summer. I would also like to thank Dr. Shengda Hu, my second reader, for the very useful comments and corrections that he provided. I need to acknowledge two of my good friends; Janis Lazovksis and Cameron Williams. Janis is a LaTeX machine, and I am very appreciative of the time he spent helping me this summer. Cam was a great support this year, always there to help me work through any problem. I also would like to thank my family for their constant love and interest, it means very much to me. Last, but not least (I'd say least goes to Cameron Williams), a thank you goes to my cousin Matt Rappoport, for without our discussions some of the contents in this paper would not exist.
}


\del{

\cleardoublepage
\vspace*{70pt}
\begin{center}
\itshape OPTIONAL DEDICATION CAN GO HERE.
\end{center}

}

\newpage

\tableofcontents

\newpage

\pagenumbering{arabic}
\setcounter{page}{1}

\section{Introduction}

This paper studies the theorem that Emmy Noether published in 1918, which provides a mathematical way to see connections between `symmetries' and `conserved quantities'.   As we shall see, Noether's theorem can be stated in both the Lagrangian and Hamiltonian frameworks. In Section 6 we demonstrate how the Legendre transform relates these statements and furthermore how, under specific requirements, it gives a one-to-one correspondence between the respective notions of symmetry and conserved quantity. 
\vspace{0.3cm}

Section 2 is dedicated to introducing the tools needed from symplectic geometry to formulate Hamiltonian mechanics. In Section 1.6 we will see how geodesic flow on a Riemannian manifold arises as a symplectomorphism generated by a specific diffeomorphism. In Section 5.3 we apply the Legendre transform to this setup and recover an equivalent way to define geodesic flow in the Hamiltonian framework.
\vspace{0.3cm}

Section 3 gives an introduction to Lagrangian mechanics. In particular, we derive the Euler-Lagrange equations using tools from the calculus of variations. A Lagrangian is just a smooth function on the tangent bundle and we will see that when this function is a `natural Lagrangian,' the Euler-Lagrange equations are equivalent to Newton's second law. We also demonstrate how the Euler-Lagrange equations are a generalization of Newton's second law; in particular, the Euler-Lagrange equations hold in non-inertial reference frames. We give many examples of Lagrangian systems and then translate these systems to the Hamiltonian framework in Section 5.
\vspace{0.3cm}

In Section 4 we use the tools introduced in Section 2 to study some basic notions in Hamiltonian mechanics. As mentioned above,  the main object of study in Lagrangian systems is the Lagrangian, which is just a smooth function on the tangent bundle $TM$. In Hamiltonian mechanics the main object is the Hamiltonian, which is just a smooth function on the cotangent bundle $T^\ast M$. In Section 2.4 we show how the cotangent bundle always has a canonical symplectic structure and so we see that an advantage of Hamiltonian mechanics is that it incorporates the use of tools from symplectic geometry.
\vspace{0.3cm}

After introducing Lagrangian and Hamiltonian mechanics, Section 5 demonstrates how the two formulations are equivalent under the Legendre transform. Given a Lagrangian $L\in C^\infty(TM)$ we get an induced map called the Legendre transform, which we denote by $\Phi_L$, from $TM$ to $T^\ast M$. Similarily, given a Hamiltonian $H\in C^\infty(T^\ast M)$ we get the induced Legendre transform $\Phi_H:T^\ast M\to TM$. Under certain conditions, which we discuss, the Legendre transform is a diffeomorphism. We use the Legendre transform to translate examples given in Section 3 and Sections 4 into the opposing frameworks. In particular, we will see how the Legendre transform takes motions in one framework to motions in the other. 
\vspace{0.3cm}

In Section 6 we study Noether's theorem in both the Lagrangian and Hamiltonian frameworks. We give physical examples in both settings to demonstrate the power of this theorem. We then show how the statements of Noether's theorem can be translated, under the Legendre transform, from one framework to the other. We give examples of how the Legendre transform takes symmetries to symmetries and conserved quantities to conserved quantities. With the definitions given, we use the Laplace-Runge-Lenz vector to show how the notions of symmetry and conserved quantity are not in one-to-one correspondence. However, we fix this problem by putting restrictions on the symmetries and conserved quantities. Lastly, we discuss the problem of how to make the correspondence one-to-one when the definitions are more general. 
\vspace{0.3cm}

Throughout this paper we will use the Einstein summation convention.

\del{

We will see how Hamiltonian mechanics

Once we give the statement and proof of Noether's theorem in both Lagrangian and Hamiltonian mechanics, we will see how these theorems can be obtained from each other, via the Legendre transform, and 
The goal of this paper is to demonstrate how the statements of Noether's theorem in Lagrangian and Hamiltonian mechanics can be translated 

, the notion of symmetry and conserved quantityMore precisely, we will give precise definitions of symmetries and conserved quantities in both the Lagrangian and Hamiltonian frameworks and show how the Legendre transform gives a one-to-one correspondence between these definitions.  how the notions of symmetries and conserved quantities in the Hamiltonian and Lagrangian frameworks are in

 \\

Before putting Noether's theorem under the Legendre transform, we first demonstrate other ways in which this function relates the two formulations of mechanics. In particular, we will see how the Legendre transform `translates' the concepts of geodesic flow and the defined motions in the two formulations. We start by introducing 

}
\newpage

\section{Symplectic Geometry}
We discuss here the concepts in symplectic geometry which will be needed to formulate Hamiltonian mechanics in Section 4. In Riemannian geometry, manifolds are equipped with a non-degenerate symmetric quadratic form, whereas  in symplectic geometry the non-degenerate quadratic form is required to be skew-symmetric. Although there are some similarities between symplectic and Riemannian geometry, as we shall see there are also some vast differences. 


\subsection{Symplectic Vector Spaces}

Let $V$ be an $m$ dimensional real vector space and $\Omega: V\times V\to\R$ a skew-symmetric bilinear map. Let $U=\left\{u\in V \ ; \ \Omega(u,v)=0\text{ for all } v\in V\right\}$. Suppose that $\dim U=k$ and that $\{u_1,\dots, u_k\}$ is a basis. Recall the standard form theorem for skew-symmetric bilinear maps:

\begin{theorem} \bf{(Standard Form for Skew-Symmetric Bilinear Maps)} \rm With $U, V$ and $\Omega$ as above, we can find $n\in \N$ and a basis $u_1,\dots,u_k,e_1,\dots e_n,f_1\dots,f_n$ of $V$ such that
\[
\begin{array}{l c l}
\Omega(u_i,v)=0&\ &\text{for all $i$ and for all $v\in V$}\\
\Omega(e_i,e_j)=0=\Omega(f_i,f_j)&\ &\text{for all $i,j$}\\
\Omega(e_i,f_j)=\delta_i^j& \ & \text{for all $i,j$}
\end{array}
\]
\end{theorem}
\begin{proof}
This is a fairly straightforward induction proof. See \cite{Da Silva}, page 3 for details.
\end{proof}
It follows that, with respect to this basis, the matrix representation of $\Omega$ is 
\[
\left[\begin{matrix}
0 & 0&0\\
0&0&\mathrm{Id}\\
0&-\mathrm{Id} & 0\\
\end{matrix}\right]
\]


\begin{definition}The bilinear map $\Omega$ is said to be \bf{symplectic (or non-degenerate)} \rm if $U=\{0\}$. If this is the case then the pair $(V,\Omega)$ is called a \bf{symplectic vector space} \rm and $\left\{e_1,\dots,e_n,f_1,\dots,f_n\right\}$ is called the corresponding \bf{symplectic basis}\rm.
\end{definition} 

It follows from the theorem that any symplectic vector space is necessarily even dimensional and the corresponding skew-symmetric bilinear map is of the form
\[
\left[\begin{matrix}
0&\text{Id}\\
-\text{Id} & 0\\
\end{matrix}\right]
\]
\begin{example}\bf(Symplectic Vector Space Prototype)\rm

The simplest example of a symplectic vector space is $(\R^{2n},\Omega_0)$ where $\Omega_0$ is defined such that $e_1=(1,0,\dots,0),\dots, e_n=(0,\dots, 0,1,0\dots, 0)$ together with $f_1=(0,\dots,0,1,0\dots,0),\dots , f_n=(0,\dots, 0,1)$ form a symplectic basis. The reason this symplectic vector space is referred to as a prototype is given by the Darboux theorem, which is stated in section $2.2$.
\end{example}


\begin{definition} Let $(V_1,\Omega_1)$ and $(V_2,\Omega_2)$ be symplectic vector spaces. A linear isomorphism $\varphi:V_1\to V_2$ is called a \bf{symplectomorphism} \rm if $\varphi^\ast\Omega_2=\Omega_1$. Here $\varphi^\ast$ is the pullback of $\varphi$ meaning that $(\varphi^\ast\Omega_2)(u,v)=\Omega_2(\varphi(u),\varphi(v))$.
\end{definition}

In the same way that a Riemannian metric  induces the musical isomorphism between $TM$ and $T^\ast M$, where $M$ is some Riemannian manifold, so the skew-symmetric bilinear form $\Omega$ induces a natural isomorphism between $V$ and $V^\ast$.

\begin{proposition}
Given a symplectic vector space $(V,\Omega)$, the non-degenerate bilinear form $\Omega$ induces an isomorphism between $V$ and $V^\ast$ through the map \[V\to V^\ast \ \ \ \ \ \ \ v\mapsto \Omega(v,\cdot)\]
\end{proposition}

\begin{proof}
The non-degeneracy of $\Omega$ shows this map is injective, while we know that $\dim V=
\dim V^\ast$. Hence this is indeed an isomorphism.
\end{proof}

\begin{definition}
Let $(V,\Omega)$ be a finite dimensional symplectic vector space and $Y\subset V$ a subspace. The \bf{symplectic complement} \rm of $Y$ is defined to be the subspace \[Y^\Omega:=\left\{v\in V\ ; \ \Omega(v,u)=0\text{ for all } u\in Y\right\}.\]
\end{definition} 

For a subspace $Y\subset V$, consider the map \[\Phi:V\to Y^\ast \ \ \ \ \ \ \ v\mapsto \Omega(v,\cdot)|_Y\]It's clear that $\ker\Phi= Y^\Omega$. The surjectivity of $\Phi$ follows by combining Proposition  2.5 together with the fact that any element of $\alpha\in Y^\ast$ can be extended to an element of $\widetilde\alpha\in V^\ast$ such that $\left.\widetilde\alpha\right|_Y=\alpha$. 
It follows, by the first isomorphism theorem, $V/Y^\Omega\cong Y^\ast$. Since $\dim Y=\dim Y^\ast$, we have that  $\dim V=\dim Y+\dim Y^\Omega$. Moreover, by definition, $\Omega|_{Y\times Y}$ is non-degenerate if and only if $Y\cap Y^\Omega=0$. That is $\Omega|_{Y\times Y}$ is non-degenerate if and only if $V=Y\oplus Y^\Omega$. This leads to the following definition.

\begin{definition}
 If $\Omega|_{Y\times Y}\equiv 0$ then $Y$ is called an \bf{isotropic} \rm subspace of $V$. If $Y$ is isotropic and $\dim Y=\frac{1}{2}\dim V$ then $Y$ is called a \bf{Lagrangian subspace} \rm of $V$.
 \end{definition} 

The above remarks give us 

\begin{proposition} A subspace $Y\subset V$ is Lagrangian if and only if $Y=Y^\Omega$.
\end{proposition}

\subsection{Symplectic Manifolds}
Let $M$ be a manifold and let $\omega\in\Omega^2(M)$ be a $2$-form. By definition, for each $p\in M$ we have that $\omega(p):=\omega_p$ is a skew-symmetric bilinear map $\omega_p:T_pM\times T_pM\to\R$.
\begin{definition} A $2$-form $\omega\in\Omega^2(M)$ is said to be \bf {symplectic} \rm if it is closed and if $\omega_p$ is symplectic (non-degenerate) for each $p\in M$. In such a case, the pair $(M,\omega)$ is called a \bf{symplectic manifold.} \rm By the standard form theorem, a symplectic manifold is necessarily even dimensional.
\end{definition}
\vspace{0.1cm}

\begin{definition}
Given a $1$-form $\mu\in T^\ast M$, the unique vector field $V_\mu$ in $T^\ast M$ such that $\omega(V_\mu,\cdot)=\mu$ is called the \bf{symplectic dual} \rm of $\mu$. That is, for each $p\in M$ we set $V_\mu(p)$ to be the pre-image of $\mu(p)$ under the map defined in Proposition 2.5. In other words, $V_\mu$ is the unique vector field satisfying \[V_\mu \hk  \omega=\mu\]
\end{definition}

\begin{example}\bf{(Prototype of a Symplectic Manifold)}\rm
\vspace{0.1cm}

Let $M=\R^{2n}$ with standard coordinates $x^1,\dots,x^n,y_1,\dots, y_n$. The form $\omega_0:=dx^i\wedge dy_i$ is symplectic, and $T_pM\cong\R^2$ has symplectic basis 
$\left\{\left. \frac{\pd}{\pd x^1}\right|_p,\dots,\left. \frac{\pd}{\pd x^n}\right|_p,\left. \frac{\pd}{\pd y_1}\right|_p,\dots, \left. \frac{\pd}{\pd y_n}\right|_p\right\}$ so that $(\R^{2n},\omega_0)$ is a symplectic manifold.
\end{example}
The Darboux theorem shows why the  manifold above can be thought of as the prototype of symplectic manifolds. The theorem locally classifies symplectic manifolds up to symplecteomorphism. That is, locally every symplectic manifold is symplectomorphic to $(\R^{2n},\omega_0)$.
\begin{theorem} \bf{(Darboux)} \rm Let $(M,\omega)$ be a symplectic manifold. For any $p\in M$ there exists a coordinate chart $(U,x^1,\dots x^n,y_1,\dots, y_n)$ centred at $p$ such that \[\omega|_U=\sum_{i=1}^ndx^i\wedge dy_i\]
\end{theorem}
The coordinates giving this local expression of $\omega$ are called \bf{Darboux coordinates}\rm.
\begin{proof}
The proof is just an application of the Frobenius theorem together with a characterization of Darboux coordinates. See \cite{Lee}, page 349 for the details.
\end{proof}

In the same way we defined Lagrangian subspaces of a vector space, we can define Lagrangian submanifolds.

\begin{definition}Given a symplectic manifold $(M,\omega)$, a submanifold $(N,\iota)$ of $M$ is called a \bf{Lagrangian submanifold} \rm if at each $p\in N$, \ $T_pN$ is a Lagrangian subspace of $T_pM$. That is, $N$ is Lagrangian if and only if $\iota^\ast\omega=\omega|_{T_pN\times T_pN}=0$ and $\dim N=\frac{1}{2}\dim M$.
\end{definition}

We finish this subsection with a simple yet important proposition.

\begin{proposition} Let $(M_1,\omega_1)$ and $(M_2,\omega_2)$ be symplectic manifolds. If $(L,\iota)$ is a Lagrangian submanifold of $(M_1,\omega_1)$ and $f:(M,\omega_1)\to (M_2,\omega_2)$ is a symplectomorphism, then $(f(L),f\circ\iota)$ is a Lagrangian submanifold of $(M_2,\omega_2)$. 
\end{proposition}
\begin{proof}
By definition, we have that $f^\ast\omega_2=\omega_1$. Hence \[(f\circ\iota)^\ast\omega_2=\iota^\ast f^\ast\omega_2=\iota^\ast\omega_1=0\]since $(L,\iota)$ is a Lagrangian submanifold of $(X_1,\omega_1)$.
\end{proof}

Using the results from this section we can answer the question of when a diffeomorphism between two symplectic manifolds is a symplectomorphism. 

\subsection{When is a Diffeomorphism a Symplectomorphism?}
Let $\varphi:(M_1,\omega_1)\to (M_2,\omega_2)$ be a diffeomorphism of two symplectic manifolds. We will see that the answer to the posed question of this subsection is ``if and only  the graph of $\varphi$ is a Lagrangian submanifold of the `twisted' symplectic manifold $(M_1\times M_2,\widetilde\omega)$.'' We first formalize the definitions in this statement.

\vspace{0.2cm}

Given the two symplectic manifolds $(M_1,\omega_1)$ and $(M_2,\omega_2)$  as above, consider their Cartesian product $M_1\times M_2$. Let $\pi_1$ and $\pi_2$ denote the projection maps onto the first and second factors respectively. For any $a,b\in\R\backslash\{0\}$, consider the $2$-form\[\omega:=a(\pi_1^\ast\omega_1)+b(\pi_2^\ast\omega_2)\]Since the exterior derivative commutes with the pull-back, it follows $\omega$ is closed. Moreover, to see that $\omega$ is symplectic, let $(p,q)\in M_1\times M_2$  be arbitrary and consider non-zero $(V_p,W_q)\in T_pM_1\times T_qM_2$. Without loss of generality, suppose that $V_p$ is nonzero. By the non-degeneracy of $\omega_1$ there exists $X_p\in T_pM$ such that $(\omega_1)_p(V_p,X_p)\not=0$ so that $\omega((p,V_p),(q,0_q))=a\cdot\omega_{1,p}(V_p,X_p)\not=0$.

\begin{definition}
In particular, taking $a=1$ and $b=-1$ we obtain the \bf{twisted product} \rm symplectic form $\widetilde\omega\in\Omega^2(M_1\times M_2)$: \[\widetilde\omega=\pi_1^\ast\omega_1-\pi_2^\ast\omega_2\]
\end{definition}

Let $\Gamma_\varphi=\left\{(p,\varphi(p)) ; \ p\in M\right\}$ denote the graph of $\varphi$. It's clear the the function \[f:M_1\to \Gamma_\varphi \ \ \ \ \ p\mapsto (p,\varphi(p))\] is an embedding. Since $\Gamma_\varphi$ is the image of $M_1$ under $f$, it follows that $\Gamma_\varphi$ is a submanifold of $M_1\times M_2$ of dimension $4n-2n=2n$. 
\vspace{0.2cm}

Hence $\Gamma_\varphi$ always satisfies `half' of the requirements of being Lagrangian. We can now prove the statement posed at the beginning of this section.

\begin{proposition}The diffeomorphism $\varphi$ is a symplectomorphism $\iff \Gamma_\varphi$ is a Lagrangian submanifold of $(M_1\times M_2,\widetilde \omega)$.
\end{proposition}

\begin{proof} We already know that $(\Gamma_\varphi,\iota)$ is a submanifold of $(M_1\times M_2,\widetilde\omega)$, where $\iota:\Gamma_\varphi:M_1\times M_2$ is the inclusion map. Let $f$ be as above. We have that 
\begin{align*}
\Gamma_\varphi\text{ is Lagrangian }&\iff \iota^\ast\widetilde\omega=0\\
&\iff f^\ast\iota^\ast\widetilde\omega=0 &\text{since $f$ is a diffeomorphism}\\
&\iff(\iota\circ f)^\ast\widetilde\omega=0\\
\end{align*}But \[(\iota\circ f)^\ast\widetilde\omega:=((\iota\circ f)^\ast\circ\pi_1^\ast)\omega_1-((\iota\circ f)^\ast\circ\pi_2^\ast)\omega_2=(\pi_1\circ \iota\circ f)^\ast\omega_1-(\pi_2\circ \iota\circ f)^\ast\omega_2=\omega_1-\varphi^\ast\omega_2\]Hence\[\Gamma_\varphi\text{ is Lagrangian } \iff \varphi^\ast\omega_2=\omega_1\]  

\end{proof}

\begin{remark}
It is crucial in the above proof that the $2$-form on $M_1\times M_2$ is the twisted product form, otherwise this would not work. 
\end{remark}

\subsection{Canonical Symplectic Structure of Cotangent Bundles}

Given an arbitrary manifold $M$, the total space of the cotangent bundle $T^\ast M$ can always be turned into a symplectic manifold. This subsection describes how.

\vspace{0.2cm}

Let $M$ be an arbitrary $n$-dimensional manifold and $T^\ast M$ the cotangent bundle. To turn $T^\ast M$ into a symplectic manifold we need to find a closed symplectic $2$-form $\omega\in\Omega^2(T^\ast M)$. Consider first the $1$-form $\alpha\in\Omega^1(T^\ast M)$ defined by \[\alpha_{(p,\xi_p)}(V_{(p,\xi_p)}):=\xi_p\left(\pi_\ast (V_{(p,\xi_p)})\right)\]where $(p,\xi_p)\in T_p^\ast M$ and $V_{(p,\xi_p)}\in T_{(p,\xi_p)}(T^\ast M)$ are arbitrary and $\pi_\ast$ is the differential of the projection map $\pi:T^\ast M\to M$. Define $\omega:=-d\alpha$. By definition, both $\alpha$ and $\omega$ are global forms on $T^\ast M$. After the computation of $\alpha$ and $\omega$ in local coordinates, shown below,  it is straightforward to verify that $\omega$ is symplectic. It is clear that $\omega$ is closed, since it is exact. Hence $(T^\ast M,\omega)$ is a symplectic manifold.

\begin{definition}
The $1$-form $\alpha\in\Omega^1(T^\ast M)$ is called the \bf{tautological $1$-form} \rm and the $2$-form  $\omega\in\Omega^2(T^\ast M)$ is called the \bf{canonical symplectic $2$-form}. 
\end{definition}

For future use we compute here $\alpha$ and $\omega$ in local coordinates. Let $(U,x^1,\dots,x^n)$ be an arbitrary coordinate chart in $M$ and $(T^\ast U, x^1,\dots, x^n,\xi_1,\dots,\xi_n)$ the induced chart on $T^\ast M$. The first thing to show is how $\pi_\ast:T(T^\ast M)\to TM$ works. For arbitrary $(p,\sigma_p)\in T_p^\ast M$ and $W_{(p,\sigma_p)}=W^i\left.\pdx\right|_{(p,\sigma_p)}+\widetilde W^i\left.\frac{\pd}{\pd\xi_i}\right|_{(p,\sigma_p)}\in T_{(p,\sigma_p)}(T^\ast M)$ there exists $a^i\in \R$ such that $\pi_\ast(W_{(p,\sigma_p)})=a^i\left.\pdx\right|_p$. It follows 
\begin{align*}
a^i&=(\pi_\ast(W_{(p,\sigma_p)}))(\left.dx^i\right|_p)\\
&=W_{(p,\sigma_p)}(x^i\circ\pi)\\
&=W^j\left.\frac{\pd}{\pd x^j}\right|_{(p,\sigma_p)}(x^i\circ\pi)+\widetilde W^j\left.\frac{\pd}{\pd\xi_j}\right|_{(p,\sigma_p)}(x^i\circ\pi)\\
&=W^i
\end{align*}That is, \[\pi_\ast\left(W_{(p,\sigma_p)}\right)=\left(p,W^i\left.\pdx\right|_p\right).\]  \\

Since $\alpha$ is an element of $\Gamma(T^\ast(T^\ast M))$ we have functions $a_1,\dots, a_n, b^1,\dots,b^n\in C^\infty(\pi^{-1}(U))$ such that $\alpha=a_idx^i+b^id\xi_i$. By definition, for arbitrary $(p,\sigma_p)\in T_p^\ast M$ \[a_i(p,\sigma_p)=\alpha_{(p,\sigma_p)}\left(\left.\frac{\pd}{\pd x^i}\right|_{(p,\sigma_p)}\right)=\sigma_p\left(\pi_\ast\left(\left.\frac{\pd}{\pd x^i}\right|_{(p,\sigma_p)}\right)\right)=\sigma_p\left(\left.\pdx\right|_p \right)=\sigma_i(p)=\xi_i(p,\sigma_p)\]and\[b^i(p,\sigma_p)=\alpha_{(p,\sigma_p)}\left(\left.\frac{\pd}{\pd \xi_i}\right|_{(p,\sigma_p)}\right)=\sigma_p\left(\pi_\ast\left(\left.\frac{\pd}{\pd \xi_i}\right|_{(p,\sigma_p)}\right)\right)=\sigma_p(p,0_p)=0=0(p,0_p).\]\\
We have shown that in local coordinates\[\alpha=\xi_i dx^i\] and it follows\[\omega:=-d\alpha=dx^i\wedge d\xi_i\] 

Given another $1$-form $\mu\in\Omega^1(M)$ we now show that the graph of $\mu$, considered as a function $M\to T^\ast M$, is a Lagrangian submanifold of $T^\ast M$ if and only if $\mu$ is closed. 
To avoid confusion, let $s_\mu$ denote the map $s_\mu:M\to T^\ast M$ given by $p\mapsto (p,\mu_p)$ and let $\Gamma_{s_\mu}$ denote the image of $\mu$ in $T^\ast M$ \[s_\mu(M):=\Gamma_{s_\mu}=\{(p,\mu_p) \ ; \ p\in M\}\]That is, the image of $\mu$ as a map is the same thing as the graph of $s_\mu$. Let $\pi:T^\ast M\to M$ denote the projection mapping. It's clear that $\pi\circ s_\mu=\mathrm{id}$.

\begin{proposition} Let $\alpha$ be the tautological $1$-form on $T^\ast M$. Then $s_\mu^\ast \alpha=\mu$.
\end{proposition}
\begin{proof}Fix arbitrary $(p,\mu_p)\in \Gamma_{s_\mu}$. By definition, $\alpha_{(p,\mu_p)}(V)=\mu_p(\pi_\ast V)$. Hence for arbitrary $V\in T_pM$, \[s_\mu^\ast\alpha(V)=\alpha((s_\mu)_\ast V)=\mu_p(\pi_\ast(s_\mu)_\ast V)=\mu_p((\pi\circ s_\mu)_\ast V)=\mu_p(\mathrm{id}(V))=\mu_p(V).\]
\end{proof}Using this we get
\begin{proposition} 
$\Gamma_{s_\mu}$ is a Lagrangian submanifold of $T^\ast M$ $\iff \mu$ is closed.
\end{proposition}

\begin{proof}
\rm Let $\tau:M\to \Gamma_{s_\mu}$ be the same map as $s_\mu$ but with range restricted to $\Gamma_{s_\mu}$. It follows that $\tau$ is a diffeomorphism and $s_\mu=\iota\circ\tau$. Hence
\begin{align*}
\Gamma_{s_\mu} \text{ is Lagrangian }&\iff \iota^\ast\omega\equiv 0\\
&\iff \iota^\ast d\alpha\equiv 0\\
&\iff \tau^\ast\iota^\ast d\alpha\equiv 0&\text{because $\tau$ is a diffeomorphism}\\
&\iff (\iota\circ\tau)^\ast d\alpha\equiv 0\\
&\iff (s_\mu)^\ast d\alpha\equiv 0\\
&\iff d(s_\mu)^\ast \alpha\equiv 0\\
&\iff d\mu\equiv 0\\
&\iff \mu \text{ is closed }
\end{align*}
\end{proof}





\subsection{Lifting a Diffeomorphism}

\begin{definition} Given a diffeomorphism $f:M_1\to M_2$ between two manifolds $M_1$ and $M_2$, there is an induced symplectomorphism $f_\sharp:T^\ast M_1\to T^\ast M_2$ called the \bf{lift} \rm of $f$ which is constructed as follows.
\vspace{0.1cm}

Since $f$ is a diffeomorphism we have that $f^\ast:T^\ast M\to T^\ast M$ is an isomorphism. For arbitrary $(p_1,\xi_{p_1})\in T^\ast_{p_1} M_1$ we define $f_\sharp$ by \[ f_\sharp(p_1,\xi_{p_1}):=(f(p_1),(f^\ast)^{-1}(\xi_{p_1})).\]Since $f$ is a diffeomorphism we get that both $f_\sharp$ and $f_\sharp^{-1}$ are bijective and smooth. Moreover, we have the following commutative diagram. 
\[
\begin{array}{l r}
\begin{tikzpicture}
\node (x1) at (0,0) {$M_1$};
\node (m1) at (0,2)  {$T^\ast M_1$};
\node (x2) at (2,0) {$M_2$};
\node (m2) at (2,2) {$T^\ast M_2$};
\draw[->] (m1) to node[left] {$\pi_1$}  (x1);
\draw[->] (m1) to node[above] {$f_\sharp$}  (m2);
\draw[->] (x1) to node[below] {$f$}  (x2);
\draw[->] (m2) to node[right] {$\pi_2$}  (x2);
\end{tikzpicture}
&(\dagger)
\end{array}
\]
\end{definition}
\begin{proposition}
Let $\alpha_1$ and $\alpha_2$ denote the tautological forms on $T^\ast M_1$ and $T^\ast M_2$ respectively. Then\[f_\sharp^\ast(\alpha_2)=\alpha_1.\]
\end{proposition}

\begin{proof}
Let $(p_1,\xi_{p_1})\in T^\ast_{p_1} M_1$ and $(p_2,\xi_{p_2})\in T^\ast_{p_2}M_2$ be such that $p_2=f(p_1)$ and $\xi_{p_1}=f^\ast\xi_{p_2}$ 

\noindent It needs to be shown that $(f_\sharp)^\ast(\alpha_2)_{(p_2,\xi_{p_2})}=(\alpha_1)_{(p_1,\xi_{p_1})}$ By definition, $(f_\sharp)^\ast(\alpha_2)_{(p_2,\xi_{p_2})}\in T^\ast_{(p_1,\xi_{p_1})}(T^\ast M_1)$ so let $\eta\in T_{(p_1,\xi_{p_1})}(T^\ast M_1$) be arbitrary. Then
\begin{align*}
f_\sharp^\ast\left((\alpha_2)_{(p_2,\xi_{p_2})}\right)(\eta)&:=(\alpha_2)_{(p_2,\xi_{p_2})}\left(((f^\sharp)_\ast)\eta\right)\\
&:=\xi_{p_2}\circ(\pi_2)_{\ast}\left(f^\sharp_\ast\eta\right)\\
&=\xi_{p_2}\left((\pi_2\circ f_\sharp)_\ast\eta\right)\\
&=\xi_{p_2}\left((f\circ\pi_1)_\ast\eta\right)&\text{ by $\dagger$}\\
&=\xi_{p_2}\left(f_\ast(\pi_{1\ast}\eta)\right)\\
&=f^\ast\xi_{p_2}\left(\pi_{1\ast}\eta\right)\\
&=\xi_{p_1}\left(\pi_{1\ast}\eta\right)\\
&=(\alpha_1)_{(p_1,\xi_{p_1})}\eta
\end{align*}
\end{proof} 
\begin{corollary}
In the setup of Proposition 2.22, if we take $M_1=M_2=M$ and let $f:M\to M$ be a diffeomorphism, then the lift of $f$ preserves $\omega$. That is, $f^\ast_\sharp(\omega)=\omega$. 
\end{corollary}

\begin{proof}
This follows immediately from the fact that the pull back commutes with the exterior derivative.
\end{proof}
The following Lemma and Theorem will be needed in section 6 to study Noether's theorem.
\begin{lemma}
Let $M$ be a manifold. Fix $X\in\Gamma(TM)$ and let $\theta_t$ denote its flow. There exists a unique vector field $X_\sharp$ on the cotangent bundle (i.e. $X_\sharp\in\Gamma(T(T^\ast M))$) such that the flow of $X_\sharp$, say $\Psi_t$, is the lift of $\theta_t$. That is, $\Psi_t=\theta_{t,\sharp}$. Note that by Corollary 2.23, each $\Psi_t$ is a symplectomorphism.
\end{lemma}

\begin{proof}
Let $\theta_t$ denote the flow of $X\in\Gamma(TM)$. We have that $\theta_t$ is a diffeomorphism $\theta_t:M\to M$ so its lift $\theta_{t,\sharp}$ is a symplectomorphism $\theta_{t,\sharp}:T^\ast M\to T^\ast M$. Proposition 2.22 shows that $\theta_{t,\sharp}$ preserves $\alpha$. Just let $X_\sharp$ be the infinitesimal generator of $\theta_{t,\sharp}$. Here the integral curves are of the form $\theta^{(p,\xi)}_\sharp:\R\to T^\ast M \ , \ t\mapsto \theta_{t,\sharp}(p,\xi)$, and so $\theta_{t,\sharp}$ is a local flow of $X_\sharp$.
\end{proof}

\begin{theorem} \bf{(Lifting to the Cotangent Bundle)}\rm

Let $M$ be a manifold. Let $\alpha\in\Gamma(T^\ast(T^\ast M))$ denote the tautological $1$-form on $T^\ast M$ and consider the symplectic manifold $(T^\ast M,\omega=-d\alpha)$. If $g:T^\ast M\to T^\ast M$ is a symplectomorphism preserving $\alpha$ (i.e. $g^\ast\alpha=\alpha$) then there exists a diffeomorphism $f:M\to M$ such that $g=f_\sharp$. 
\end{theorem}

\begin{proof}

The proof of this theorem is done by combining the following claims. For what is below we let $V$ denote the symplectic dual of $\alpha$. That is, $\omega(V,\cdot)=V\hk\omega=\alpha$.

\begin{claim}
 If $g^\ast\alpha=\alpha$ then $g$ commutes with the flow of $V$, or equivalently $g_\ast V=V$.
\end{claim}

\begin{proof}
Let $\theta_t$ denote the flow of $V$. It needs to be shown that $g\circ\theta_t=\theta_t\circ g$, or equivalently, that $g\circ\theta_t\circ g^{-1}=\theta_t$. By definition, for each $p\in M$, we have that $\theta^{(p)}$ is the unique curve satisfying $\theta^{(p)}(0)=p$ and $\left.\dt\right|_{t=0}\theta^{(p)} =V_p$. Since $\theta_0$ is the identity, we have that $g\circ\theta_0\circ g^{-1}(p)=p$. Hence, by uniqueness, it suffices to show that $\left.\dt\right|_{t=0} g\circ\theta_t\circ g^{-1}(p)=V_p$ which happens, by the non-degeneracy of $\omega$, if and only if \begin{align}\left(\left.\dt\right|_{t=0} g\circ\theta_t\circ g^{-1}(p)\right)\hk\omega_p=V_p\hk\omega_p=\alpha_p\end{align} By the chain rule
\begin{align*}
\left.\dt\right|_{t=0}g\circ\theta_t\circ g^{-1}(p)&=g_{\ast,g^{-1}(p)}\left(\dt(\theta_t(g^{-1}(p)))\right)\\
&=g_{\ast,g^{-1}(p)}(V_{g^{-1}(p)})
\end{align*}
Fix an arbitrary $Y_p\in\Gamma(T_{p}M)$ and plug it into both sides of $(2.1)$. The right hand side is \[V_p\hk\omega_p(Y_p)=\alpha_{p}(Y_p)\] while the left hand side becomes
\begin{align*}
\left(\left.\dt\right|_{t=0} g\circ\theta_t\circ g^{-1}(p)\hk\omega_p\right)(Y_p)&=\omega_p\left(g_{\ast,g^{-1}(p)}(V_{g^{-1}(p)}),Y_p\right)\\
&=\omega_p\left(g_{\ast,g^{-1}(p)}(V_{g^{-1}(p)}),g_{\ast,g^{-1}(p)}\circ g^{-1}_{\ast,p}(Y_p)\right)\\
&=(g^\ast\omega)_{g^{-1}(p)}\left(V_{g^{-1}(p)},g^{-1}_{\ast,p}(Y_p)\right)\\
&=\omega_{g^{-1}(p)}\left(V_{g^{-1}(p)},g^{-1}_{\ast,p}(Y_p)\right)\\
&=\alpha_{g^{-1}(p)}(g^{-1}_{\ast,p}(Y_p))\\
&=(g^\ast\alpha)_p((g_{\ast,p})^{-1}(Y_p))\\
&=\alpha_p(Y_p)
\end{align*}

\end{proof}
Notice that this claim had nothing to do with the fact the symplectic manifold was a cotangent bundle. This result holds for any symplectic manifold $(X,\omega)$ for which $\omega=-d\alpha$ for some $1$-form $\alpha$ and $g^\ast\alpha=\alpha$.

\begin{claim}
The integral curves, $\gamma:\R\to T^\ast M$, of $V$ are of the form \[\gamma^{(p,\sigma)}(t)=(p,\sigma e^{-t})\]where $(p,\sigma)\in T^\ast M$ is arbitrary. 
\end{claim}

\begin{proof} In local coordinates we know that $\alpha=\xi_i dx^i$ and $\omega=dx^i\wedge d\xi_i$. Let $V=a^i\pdx +b_i\frac{\pd}{\pd\xi_i}$ where $a^i,b_i\in C^\infty(T^\ast U)$. By definition \[\alpha=V\hk\omega=a^id\xi_i-b_idx^i\]and so $a^i=0$ and $b_i=-\xi_i$. Let $(p,\sigma)\in T^\ast M$ be arbitrary and suppose $\gamma:\R\to T^\ast M$ is an integral curve of $V$ starting at $(p,\sigma)$. We can write $\gamma(t)=(q(t),r(t))$ and it follows \[\gamma^\prime(t)=V_{\gamma(t)}=(q^i(t))^\prime\left.\frac{\pd}{\pd x^i}\right|_{\gamma(t)}+(r_i(t))^\prime\left.\frac{\pd}{\pd \xi_i}\right|_{\gamma(t)}\]It follows that for all $t\in\R$, $(q_i(t))^\prime=0$ and $(r_i(t))^\prime=-\xi_i(\gamma(t))=-r_i(t)$. That is, $q^i(t)$ is a constant function while $r_i(t)=r_i(0)e^{-t}$. By assumption $\gamma(0)=(p,\sigma)$ and so it follows $\gamma(t)=(p,\sigma e^{-t})$



\end{proof}

It immediately follows that $\theta_t$ is fibre preserving. That is, $\theta_t(T_x^\ast M)=T_x^\ast M$. Also, it implies that if $g(p,\xi)=(q,\eta)$ then for all $\lambda>0 \ , \ g(p,\lambda\xi)=g(q,\lambda\eta)$. This is because the flow of $V$ is complete and $e^{-t}$ is surjective onto $(0,\infty)$. Also, by the continuity of $g$ and $\theta_t$ we have that \[g(p,0_p)=g(p,\lim_{t\to \infty}\theta_t(\xi))=\lim_{t\to \infty}g(\theta_t(p,\xi))=\lim_{t\to \infty}\theta_t(g(p,\xi))=g(q,0_q)
\] Hence \[g(p,\xi)=(q,\eta)\implies g(p,\lambda\xi)=(q,\lambda\eta)\text{ for all $\lambda\geq0$ }\]Supposing that $g(p,\xi)=(q,\eta)$ consider another arbitrary element $(p,\sigma_p)$ of $T_p^\ast M$. Suppose that $g(p,\sigma_p)=(\widetilde q,\mu_{\widetilde q})$. Then by applying the above it follows $g(p,\lambda\sigma_p)=(\widetilde q,\lambda\mu_p)$ for $\lambda=0$. That is, $g(p,0)=(\widetilde q,0)$. But $g(p,0)=(q,0)$ and so it must be that $q=\widetilde q$. Hence $g$ maps fibres to fibres. \\

We are now ready to construct $f:M\to M$ such that $f_\sharp=g$. Indeed, define $f$ by \[f:M\to M \ \ \ \ \ \ \ p\mapsto \pi\circ g(p,0_p).\]Claim 2.27 shows that $f$ is well defined, while it readily follows that \begin{align}f\circ\pi=\pi\circ g.\end{align} To prove $f_\sharp=g$, we will show $H:=g\circ f_\sharp^{-1}$ is the identity map. By definition, $H$ is a map from $T^\ast M$ to $T^\ast M$. Let $(p,\sigma_p)\in T^\ast M$ be arbitrary and suppose that $H(p,\sigma_p)=(q,\eta_q)$. Let $V_{(p,\sigma_p)}\in\Gamma(T_{(p,\sigma_p)}(T^\ast M))$ be arbitrary. By hypothesis, $g$ preserves $\alpha$ while Proposition 2.22 shows that $f_\sharp$ also preserves $\alpha$. Hence $H$ preserves $\alpha$. That is, 
\begin{align*}
(H^\ast\alpha_{(q,\eta_q)})(V_{(p,\sigma_p)})&=(H^\ast\alpha)_{(p,\sigma_p)}(V_{(p,\sigma_p)})\\
&=\alpha_{(p,\sigma_p)}(V_{(p,\sigma_p)})\\
&:=\sigma_p(\pi_\ast(V_{(p,\sigma_p)}))
\end{align*}
On the other hand, 
\begin{align*}
(H^\ast\alpha_{(q,\eta_q)})(V_{(p,\sigma_p)})&=\alpha_{(q,\eta_q)}(H_\ast(V_{(p,\sigma_p)}))\\
&=\eta_q(\pi_\ast(H_\ast(V_{(p,\sigma_p)})))\\
&=\eta_q((\pi\circ g\circ f_\sharp^{-1})_\ast(V_{(p,\sigma_p)}))\\
&=\eta_q((f\circ\pi\circ f_\sharp^{-1})_\ast(V_{(p,\sigma_p)}))&\text{ by (2.2)}\\
&=\eta_q(\pi_\ast(V_{(p,\sigma_p)}))&\text{by $\dagger$}
\end{align*}

But if $\sigma_p(\pi_\ast(V_{(p,\sigma_p)}))=\eta_q(\pi_\ast(V_{(p,\sigma_p)}))$ for all $V_{(p,\sigma_p)}\in T_{(p,\sigma_p)}(T^\ast M)$ then it must be that $(p,\sigma_p)=(q,\eta_q)=H(p,\sigma_p)$. That is, $H$ is the identity map.
\end{proof}

\subsection{Constructing Symplectomorphisms}

Given two manifolds $M_1$ and $M_2$, we demonstrated in Section 2.4 that their cotangent bundles have a canonical symplectic structure. Let $\alpha_1$ and $\alpha_2$ denote the tautological $1$-forms on $T^\ast M_1$ and $T^\ast M_2$ respectively. Let $\omega_1=-d\alpha$ and $\omega_2=-d\alpha_2$ denote the canonical $2$-forms on $T^\ast M_1$ and $T^\ast M_2$ respectively.  A straightforward calculation shows that the tautological $1$-form on $T^\ast M_1\times T^\ast M_2$ is \[\alpha=\pi_1^\ast\alpha_1+\pi_2^\ast\alpha_2.\] implying that the canonical symplectic form on $(T^\ast M_1\times T^\ast M_2)\cong T^\ast(M_1\times M_2)$ is \[\omega=\pi_1^\ast\omega_1+\pi_2^\ast\omega_2.\]By Proposition 2.16 we know that if the graph of a diffeomorphism $\varphi:(T^\ast M_1,\omega_1)\to (T^\ast M_2,\omega_2)$ is a Lagrangian submanifold of the `twisted product' $(T^\ast M_1\times T^\ast M_2,\widetilde\omega)$, then $\varphi$ is a symplectomorphism. While by Proposition 2.20 we know that the graph of a  $1$-form $\mu\in\Omega^1(M_1\times M_2)$ is a Lagrangian submanifold of $T^\ast(M_1\times M_2)$ if the $1$-form is closed. Hence, in particular, we have that $\Gamma_{df}$ is a Lagrangian submanifold of $(T^\ast M_1,\omega_1)\times (T^\ast M_2,\omega_2)\cong (T^\ast (M_1\times M_2),\omega)$ for any smooth function $f\in C^\infty(M_1\times M_2)$. Also, by Proposition 2.14 we know that symplectomorphisms take Lagrangian submanifolds to Lagrangian submanifolds. With this in mind, we first find a symplectomorphism from $(T^\ast(M_1\times M_2),\omega)$ to $(T^\ast(M_1\times M_2),\widetilde\omega)$ so that we can find Lagrangian submanifolds of the later. After doing this, we `construct' a symplectomorphism  between $T^\ast M_1$ and $T^\ast M_2$ by finding a diffeomorphism $\tau:T^\ast M_1\to T^\ast M_2$ whose graph, which is a subset of $T^\ast(M_1\times M_2)$, equals the graph of $df$. We will see that the existence of such a diffeomorphism is governed by the implicit function theorem.\\

Consider the functions \[\sigma_2:T^\ast M_2\to T^\ast M_2 \ , \ (p,\xi)\mapsto(p,-\xi)\]and\[\sigma:=\mathrm{id}\times \sigma_2: T^\ast M_1\times T^\ast M_2\to T^\ast M_1\times T^\ast M_2\]
\begin{proposition} The map $\sigma:(T^\ast M_1\times T^\ast M_2,\omega)\to (T^\ast M_1\times T^\ast M_2,\widetilde\omega)$ is a symplectomorphism. That is, $\sigma^\ast\widetilde\omega=\omega$
\end{proposition}
\begin{proof}
First note that $\sigma$ is involutive and so bijective.  Moreover in local coordinates $x^1,\dots,x^n,\xi_1,\dots,\xi_n$ on $T^\ast M_2$ we have \[\sigma_2^\ast\alpha_2=\sigma_2^\ast(\xi_idx^i)=(\xi_i\circ\sigma_2)(d(x^i\circ\sigma_2))=-\xi_idx^i=-\alpha_2\]and so \[\sigma^\ast\omega=\sigma^\ast(\pi_1^\ast\omega_1)+\sigma^\ast(\pi_2^\ast\omega_2)=(\pi_1\circ\sigma)^\ast\omega_1+(\pi_2\circ\sigma)^\ast\omega_2=\pi_1^\ast\omega_1-\pi_2^\ast\omega_2=\widetilde\omega\]
\end{proof}

\begin{definition}
If $Y$ is a Lagrangian submanifold of $(T^\ast M_1\times T^\ast M_2,\omega)$ we define the \bf{twist} \rm of $Y$, denoted $Y^\sigma$ to be the image of $Y$ under $\sigma$. That is, $Y^\sigma:=\sigma(Y)$.
\end{definition}

\begin{proposition} If $Y$ is a Lagrangian submanifold of $(T^\ast M_1\times T^\ast M_2,\omega)$ then the twist of $Y$ is a Lagrangian submanifold of $(T^\ast M_1\times T^\ast M_2,\widetilde\omega)$
\end{proposition}
\begin{proof} \rm Since $\sigma:(T^\ast M_1\times T^\ast M_2,\omega)\to (T^\ast M_1\times T^\ast M_2,\widetilde\omega)$ is a symplectomorphism, this result is a corollary of Proposition 2.14.
\end{proof}

As mentioned at the beginning of this section, we now want to find a diffeomorphism whose graph equals the graph of the closed $1$-form $df$, where $f\in C^\infty(M_1\times M_2)$. We call the graph of $df$ the `Lagrangian submanifold generated by $f$'. Before stating this formally, we introduce some notation so that we can write this submanifold in a way that will allow us to find conditions on when it is the graph of a diffeomorphism $\varphi$. By definition, we have that $(df)_{(x,y)}=\pi_1^\ast((d_1f)_x)+\pi_2^\ast((d_2f)_y)$ where $x^1,\dots,x^n$ and $y^1,\dots, y^n$ are local coordinates on $M_1$ and $M_2$ respectively, $\pi_1$ and $\pi_2$ are the natural projections on $T^\ast M_1\times T^\ast M_2$ and $d_1f=\frac{\pd f}{\pd x^i}dx^i$ and $d_2f=\frac{\pd f}{\pd y^i}dy^i$.


\begin{definition}
The \bf{Lagrangian submanifold generated by $f$} \rm is the Lagrangian submanifold of $M_1\times M_2$ defined by\[Y_f:=\left\{((x,y),(df)_{(x,y)}) \ ; \ (x,y)\in M_1\times M_2\right\}=\left\{((x,y),((d_1f)_x,(d_2f)_y)) \ ; \ (x,y)\in M_1\times M_2\right\}\]
\end{definition}

\begin{definition}
If there exists a diffeomorphism $\varphi:T^\ast M_1\to T^\ast M_2$ such that $Y_f^\sigma=\Gamma_\varphi$ then, by Proposition 2.16, $\varphi$ is a symplectomorphism. If such a symplectomorphism exists we call it the \bf{symplectomorphism generated by $f$}\rm. Recall that here $Y^\sigma_f$ is the twist of $Y_f$.
\end{definition}

We are trying to find a diffeomorphism $\varphi: T^\ast M_1\to T^\ast M_2$ such that $\Gamma_\varphi=Y^\sigma_f$. But notice that 
\begin{align*}
\Gamma_\varphi=Y_f^\sigma&\iff \left\{((x,\xi),(y,\eta)) : \varphi(x,\xi)=(y,\eta)\right\}=\left\{((x,(d_1f)_x),(y,-(d_2f)_y))\right\}\\
&\iff \xi=(d_1f)_x\text{ and } \eta=-(d_2f)_y\\
&\iff \xi_i=\frac{\pd f}{\pd x_i}(x,y) \ \ (\star) \ \ \ \text{ and } \ \ \ \eta_i=-\frac{\pd f}{\pd y_i}(x,y) \ \ (\star\star)
\end{align*}

That is, writing $\varphi(x,\xi)=(\varphi_1(x,\xi),\varphi_2(x,\xi))$, for $\Gamma_\varphi$ to equal $Y_f^\sigma$ all we need is that $\varphi_2(x,\xi)=-(d_2f)_y$, for then it automatically follows $\varphi_1(x,\xi)=y$.  Given any $(x,\xi)$ the implicit function theorem says (locally) when a solution to $(\star)$ exists. That is, it tells us when one can write $y$ as a function of both $x$ and $\xi$. If $(x^1,\dots, x^n,y^1,\dots, y^n)$ are local coordinates on $M_1\times M_2$ the implicit function theorem says that we can write $y$ as a function of $x$ and $\xi$ locally if and only if \[\det \left[\frac{\pd }{\pd y^j}\left(\frac{\pd f}{\pd x^i}\right)\right]\not= 0.\]

Also, note that if we have a solution to $(\star)$ say, $y=\varphi_1(x,\xi)$ then we can plug this solution into $(\star\star)$ and thus completely determine the map $\varphi$ satisfying $\Gamma_\varphi=Y_f^\sigma$. We give some examples of this process  in the next section.

\subsection{Applications to Geodesic Flow}

Recall the definition of geodesic flow. 

\begin{definition}
Let $(M,g)$ be a Riemannian manifold. The \bf{geodesic flow} \rm of $M$ is the local $\R$-action on $TM$ defined by \[\Theta:\R\times TM\to TM \ , \ (t_0,V_p)\mapsto \left.\dt\right|_{t=t_0}\gamma_{V_p}(t)\] where $\gamma_{V_p}$ is the unique geodesic starting at $p\in M$ with initial velocity $V_p$. 
\end{definition}

Recall that a Riemannian manifold is called \bf{geodesically complete} \rm if every geodesic is defined for all $t\in\R$ and is called \bf{geodesically convex} \rm if for any two points in the manifold there exists a minimizing geodesic connecting them. For the rest of this section, unless stated otherwise, all manifolds are assumed to be compact. Since compact metric spaces are complete, it follows from the Hopf-Rinow theorem that all of our Riemannian manifolds are geodesically complete and geodesically convex.

\begin{example}\bf{(Free Translational Motion)}\rm
\vspace{0.1cm}

\noindent Let $M_1=\R^n=M_2$ with coordinate charts $(\R^n,x^1,\dots, x^n)$ and $(\R^n,y^1,\dots, y^n)$ respectively. Endow $M_1$ and $M_2$ with the standard metric. We have the respective induced coordinate charts $(T^\ast\R^n,x^1,\dots,x^n,\xi_1,\dots,\xi_n)$ and $(T^\ast\R^n,y^1,\dots, y^n, \eta_1,\dots,\eta_n)$. Let $f\in C^\infty (M_1\times M_2)$ be given by \[f(x,y)=-\frac{1}{2}d(x,y)^2=-\frac{1}{2}\sum_{i=1}^n(x^i-y^i)^2.\] Since the metric is assumed to be the standard one it follows that $d(x,y)$ is the usual Euclidean distance. By definition,\[Y_f^\sigma=\left\{\left(a,b,\frac{\pd f}{\pd x^i}(a,b)dx^i,-\frac{\pd f}{\pd y^i}(a,b)dy^i\right) \ | \ a,b\in\R^n\right\}.\]We would like to find the symplectomorphism genereated by $f$. That is, we would like to find a map $\varphi:\R^n\to\R^n$ such that  $Y_f^\sigma$ equals $\Gamma_\varphi$. In this case $(\star)$ is $\xi_i=\frac{\pd f}{\pd x^i}(a,b)=b^i-a^i$ and $(\star\star)$ is $\eta_i=-\frac{\pd f}{\pd y^i}=b^i-a^i$. Since $\frac{\pd f}{\pd x^i}=y^i-x^i$ we have that \[\left[\frac{\pd}{\pd y^j}\left(\frac{ \pd f}{\pd x^i}\right)\right]_{ij}=\delta_{ij}\] so that the implicit function theorem guarantees a solution to $(\star)$. We have shown that \[Y_f^\sigma=\left\{\left(a,b,b-a,b-a\right) \ | \ a,b\in\R^n\right\}.\] For fixed $a\in\R^n$, the implicit function theorem has shown the existence of a function $\varphi_1:\R^n\to\R^n$ where for every $b\in\R^n$ there exists $\xi\in\R^n$ such that $b=\varphi_1(a,\xi)$. In this case, it is obvious that every element $\xi\in T^\ast_a\R^n=\R^n$ is of the form $b-a$ for some $b\in\R^n$. We set \[\varphi_1(a,\xi)=\varphi_1(a,b-a)=b\] and \[\eta=\varphi_2(a,\xi):=-\frac{\pd f}{\pd y^i}(a,b)=\xi\] so that $\varphi(a,\xi)=(b,\eta)=(\xi+a,\xi)$ and $\Gamma_\varphi=Y_f^\sigma$. That is, $\varphi$ is the symplectomorphism generated by $f$. Identifying $T_p^\ast \R^n$ with $T_p\R^n=\left\{ \text{vectors emanating from $p$} \right\}$ we see that $\varphi$ is free translational motion. 

\end{example}

The above example is a special case of the following. 

\begin{example}\bf{(Geodesic Flow)} \rm 
\vspace{0.1cm}

Let $(M,g)$ be an arbitrary compact Riemannian manifold. As in the previous example, let $d:M\times M\to\R$ denote the Riemann distance function and $f\in C^\infty (M\times M)$ be given by $f=-\frac{1}{2}d(a,b)^2$. To find the sympelctomorphism generated by $f$ we need to solve \[(\star) \ \ \  \xi_i=d_af \ \ \text{ and } \ \ (\star\star) \ \ \ \eta_i =-d_bf.\] Using the musical isomorphism we can identify $T^\ast M$ with $TM$. That is, for $a\in M$, we have \[\flat:T_aM\to T^\ast_aM \ \ \ \ \ \ \ V\mapsto V^\flat:=g(V,\cdot).\]Let $V$ and $W$ in $TM$ be the unique vector fields such that $V^\flat=\xi$ and $W^\flat =\eta$. Then the above equations become \[(\star) \ \ \ g(V,\cdot)=d_af(\cdot) \ \ \text{ and }\ \   (\star\star) \ \ \ g(W,\cdot)=-d_bf(\cdot).\]We now show that under this identification the symplectomorphism generated by $f$ is the geodesic flow. First we need the following Lemma. 

\begin{lemma}Let $U,V\in T_aM$. Then \[\left.\frac{d}{dt}\right|_{t=0}-\frac{1}{2}d((\exp)_a(U),(\exp)_a(tV))^2=<V,U>.\]
\end{lemma}
\begin{proof}First notice that for $s\in\R$ small enough so that $(\exp)_a(sU)$ is contained in a geodesic ball centred at $(\exp)_a(U)$, we have that $d((\exp)_a(U),(\exp)_a(sU))=|1-s||U|$. Indeed, let $\gamma:\R\to M$ be the unique geodesic starting at $a$ with initial velocity $U$. Since geodesics have constant speed it follows that the length of $\gamma$ over $[s,1]$, denoted $L\left(\gamma|_{[s,1]}\right)$, is 
\begin{align*}\
L\left(\gamma|_{[s,1]}\right)&=\int_s^1\left|\gamma^\prime(t)\right|dt\\
&=|U|\int_s^1dt\\
&=|U||1-s|
\end{align*} Since the radial geodesic from $(\exp)_a(U)$ to $(\exp)_a(sU)$ is the unique minimizing curve from $(\exp)_a(U)$ to $(\exp)_a(sU)$ (for a proof of this see \cite{Lee2}, Proposition 6.10) it follows that \[d((\exp)_a(U),(\exp)_a(sU))=|1-s||U|.\] Next, fix $(\exp)_a(U)\in M$ and consider the function, also denoted $f$, defined by \[f:M\to\R \ \ \ \ \ \ \ b\mapsto f((\exp)_a(U),b).\] Notice that given $V^\perp\in T_aM$ with $V^\perp\perp U$ the Gauss Lemma shows that there exists a curve $\beta:\R\to M$ starting at $a$ with initial velocity $V^\perp$ such that $d(\beta(t),b)$ is constant. That is, $f(\beta(t))=c$ implying that $f_\ast(\beta^\prime(0))=0$. Observe that we can write $V=V^\perp+\lambda U$ for some $\lambda\in\R$. Letting $G$ denote the function $G:\R\to M$ given by $t\mapsto (\exp)_a(tV)$ it follows that 
\begin{align*}
\left.\frac{d}{dt}\right|_{t=0}-\frac{1}{2}d((\exp)_a(U),(\exp)_a(tV))^2&=\left.\frac{d}{dt}\right|_{t=0}f(G(t))\\
&=(f_\ast)_{G(0)}G^\prime(0)\\
&=(f_\ast)_a(V)\\
&=\lambda f_\ast(U)\\
&=\lambda\left.\frac{d}{ds}\right|_{s=0}-\frac{1}{2}(d((\exp)_a(sU),(\exp)_aU))^2\\
&=\lambda\left.\frac{d}{ds}\right|_{s=0}-\frac{1}{2}|1-s|^2|U|^2\\
&=\lambda|U|^2\\
&=<V,U>
\end{align*}
\end{proof}
With this lemma we now reconsider $(\star)$. Evaluating the left hand side at $V$ we get $|V|^2$. By geodesic convexity, there exists some $U\in TM$ such that $(\exp)_a(U)=b$. Using Lemma 2.36, the right hand side is

\begin{align*}
d_af(V)&=(f_\ast)_a(V)\\
&=\left.\frac{d}{dt}\right|_{t=0}d((\exp)_a(tV),y)\\
&=\left.\frac{d}{dt}\right|_{t=0}d((\exp)_a(tV),(\exp)_a(U))\\
&=<V,U>
\end{align*}
Now take any vector $V^\prime\in TM$ such that $V^\prime\perp V$. Plugging $V^\prime$ into $(\star)$ the left hand side becomes $0$ and the right hand side is $<V^\prime,U>$. Hence, $<V,U>=<V,V>$ and $<V^\prime,U>=0=<V,V^\prime>$ for any $V^\prime\perp V$. It follows that $U=V$ and so $b=(\exp)_a(V)$.
\vspace{0.1cm}
 
We now need to solve $(\star\star)$. We will see the solution is given by $W=\gamma_V^\prime(1)=\left.\frac{d}{dt}\right|_{t=1}(\exp)_a(tV)$.  Indeed, let $\widetilde W=\left.\frac{d}{dt}\right|_{t=1}(\exp)_a(tV)$ and fix any $W^\prime\perp \widetilde W$. Again, by the Gauss Lemma, we have that $d_bf(W^\prime)=0$ and so $W=k\widetilde W$ for some $k\in\R$. But since geodesics have constant speed it follows that $|V|^2=|\widetilde W|^2$. Therefore the left hand side of $(\star\star)$ is

\begin{align*}
<W,\widetilde W>&=k<\widetilde W,\widetilde W>\\
&=k|\widetilde W|^2\\
&=k|V|^2
\end{align*}while the right hand side is

\begin{align*}
-(d_bf)(\widetilde W)&=-(f_\ast)_b(\widetilde W)\\
&=\left.\frac{d}{ds}\right|_{s=0}\frac{1}{2}d(a,(\exp)_a((1+s)(V)))^2\\
&=\left.\frac{d}{ds}\right|_{s=0}\frac{1}{2}d((\exp)_a(0V),(\exp)_a((1+s)V))^2\\
&=\left.\frac{d}{ds}\right|_{s=0}\frac{1}{2}(1+s^2)|V|^2\\
&=|V|^2
\end{align*}
Hence $k=1$ showing that $W=\left.\frac{d}{dt}\right|_{t=1}(\exp)_a(V)$.\\

In summary, for $\Gamma_\varphi$ to equal $Y_f^\sigma$ it needs to be that $V$ is the unique vector field such that $b=(\exp)_a(V)$  (so that $b$ is a function of $a$ and $V$) and further that $W=\left.\frac{d}{dt}\right|_{t=1}(\exp)_a(tV)$. That is, the map $\varphi$ is given by \[(a,V)\ \ \mapsto \ \ \left((\exp)_a(V),\left.\frac{d}{dt}\right|_{t=1}(\exp)_a(tV)\right)=\left(\gamma_V(1),\gamma_V^\prime(1)\right)\]\\

In the section on Hamiltonian mechanics, we will return to the concept of geodesic flow and give some insight as to why it arose as the symplectomorphism generated by the Riemann distance function.

\end{example}

\newpage

\section{Lagrangian Mechanics}

Recall that Newton's second law states that in an inertial reference frame the motion of a particle, with position $\gamma(t)$, is given by the solution to the ODE \[F(\gamma(t))=p^\prime(t)\] where $F$ is the net force acting on the particle and $p(t)=m(t)v(t)$ is the particle's momentum. Lagrangian mechanics is a reformulation of Newtonian mechanics in which motions are given by solutions to the Euler-Lagrange equations. In any situation where Newton's second law can be applied, so can the Euler-Lagrange equations. We will see that, in any such system, the Euler-Lagrange equations are equivalent to Newton's second law. However, the equations also hold in settings in which Newton's second law does not hold. For example, Newton's second law only holds in an inertial reference frame, while the Euler-Lagrange equations are valid in any coordinate system. Another advantage in using the Euler-Lagrange equations comes with the way in which they allow constraint forces on a mechanical system to be ignored. For example, if studying the motion of a bead on a wire, in the Lagrangian setting we do not need to worry about the forces keeping the bead constrained to the wire.
\vspace{0.3cm}

For the rest of this paper all forces are assumed to be conservative, and so we first recall this definition.

\subsection{Conservative and Central Forces}

Consider $\R^n$ equipped with standard coordinates $x^1,\dots,x^n$. Let $g$ be a metric on $\R^n$ and let $x^1,\dots,x^n,v^1\dots,v^n$ denote the induced coordinates on $T\R^n=\R^{2n}$. Recall the kinetic energy is defined by \[K:\R^{2n}\to\R \ \ \ \ \ \ \ (x,v)\mapsto\frac{1}{2}mg(v,v)=\frac{1}{2}mg_{ij}v^iv^j=\frac{1}{2}m|v|^2\] where $m$ is some positive constant, called the mass of the particle. Fix two points $r_1,r_2\in\R^n$ and consider a curve $\gamma:[a,b]\to\R^n$ such that $\gamma(a)=r_1$ and $\gamma(b)=r_2$. Then the work done by a force $F$ on a particle moving along $\gamma$  is defined to be the integral of $F$ over the curve $\gamma$ \[W(r_1\to r_2,\gamma):=\int_\gamma F\cdot ds=\int_a^bF(\gamma(t))\cdot\gamma^\prime(t)dt\]

\begin{theorem}\bf{(Work-Kinetic Energy Theorem)}\rm 
\vspace{0.1cm}

Given a system of $k$ particles of masses $m_i$ with position $r_i$, the change in kinetic energy of the system is equal to the sum of the work done on each particle.
\end{theorem}

\begin{proof}
Let $F_i$ denote the force acting on the $i$-th particle and $W_i$ the work done by $F_i$ on the $i$-th particle. By Newton's second law we have that \[\left.\dt\right|_{r(t)} K=\sum_{i=1}^n(m_i r_i^{\prime\prime}(t))\cdot r_i^\prime(t)=\sum_{i=1}^nF_i(r_i(t))\cdot r^\prime(t).\]It follows that \[K(t_2)-K(t_1)=\int_{t_1}^{t_2}\frac{d K}{dt}dt=\sum_{i=1}^n\int_{t_1}^{t_2}F_i\cdot r_i^\prime dt=\sum_{i=1}^nW_i\]
\end{proof}

As the next example shows, given two paths $\alpha,\beta:[a,b]\to\R^n$ with $\alpha(a)=r_1=\beta(a)$ and $\alpha(b)=r_2=\beta(b)$, it may be that $W(r_1\to r_2,\alpha)\not =W(r_1\to r_2,\beta)$.  
\vspace{0.1cm}

\begin{example}\bf{(A Non-Conservative Force)}\rm
\vspace{0.1cm}

Consider a particle moving around the unit circle under the following force field. 
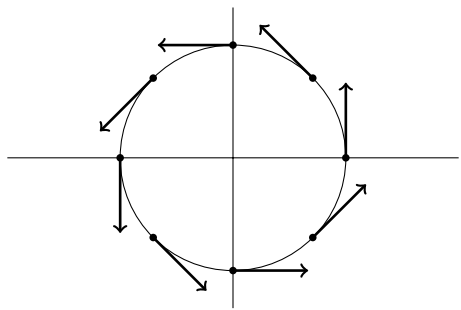
\[
\newcommand\numarrs{8}
\begin{tikzpicture}
\draw (-3,0)--(3,0) (0,-2)--(0,2);
\foreach \r in {1,...,\numarrs}{
  \draw[->,line width=1pt] ($(0,0)+(360/\numarrs*\r:1.5)$)--+(360/\numarrs*\r+90:1);
  \fill[fill=black] ($(0,0)+(360/\numarrs*\r:1.5)$) circle (.05);
}
\draw (0,0) circle (1.5);
\end{tikzpicture}
\]

This vector field (force) has integral curve $\gamma(t)=(\cos t,\sin t)$. It follows that $F(x,y)=(-y,x)$. Fix the points $(1,0)$ and $(-1,0)$. Consider the curve $\alpha:[0,\pi]\to \R^2$ given by $t\mapsto (\cos t,\sin t)$ and $\beta:[0,\pi]\to \R^2$ given by $t\mapsto (\cos t,-\sin t)$. Then $\alpha$ and $\beta$ are both curves starting at $(1,0)$ and ending at $(-1,0)$; however, \[W((1,0)\to (-1,0),\alpha)=\int_\alpha F\cdot ds=\int_0^\pi (-\sin t,\cos t)\cdot (-\sin t,\cos t)dt=\pi\] while \[W((1,0)\to(-1,0),\beta)=\int_\beta F\cdot ds=\int_0^\pi (\sin t, \cos t)\cdot (-\sin t,-\cos t)=-\pi\] 
\end{example}

Contrary to this example, for certain forces (such as the gravitational and electrostatic forces) it is the case that $W$ does not depend on the path traversed by the particle. This leads to the following definition.

\begin{definition}
A force is called \bf{conservative} \rm if the work done is path independent. Letting $r_1,r_2\in\R^n$ be arbitrary, this means that  for any two curves $\alpha,\beta:[a,b]\to\R^n$ with $\alpha(a)=r_1=\beta(a)$ and $\alpha(b)=r_2=\beta(b)$, \[W(r_1\to r_2,\alpha)=\int_\alpha F\cdot ds=\int_\beta F\cdot ds=W(r_1\to r_2,\beta)\]In this case the work done by $F$ is justifiably denoted $W(r_1\to r_2)$. A mechanical system is called \bf{conservative} \rm if the net force is conservative. 
\end{definition}
Recall that the gradient of a function $U:\R^n\to\R^n$ is defined to be \[\nabla U:=(dU)^\sharp=\left(g^{ji}\frac{\pd U}{\pd x^i}\right)\frac{\pd}{\pd x^j}\]where $\sharp:T^\ast \R^n\to T\R^n$ is the musical isomorphism and $g^{ji}=[g^{-1}]_{ji}$. With the standard metric, the above definition reduces to the standard notion of the gradient.
\begin{theorem}
A force $F$ is conservative if and only if there exists a continuously differentiable function $U:\R^n\to \R$ such that $F=-\nabla U$. 
\end{theorem}

\begin{proof} 
First suppose that the work done by $F$ is conservative. Fix a point $r_0\in\R^n$ and define \[U:\R^n\to \R \ \ \ \ \ \ \ r\mapsto -\int_\gamma F\cdot ds\] where $\gamma:[a,b]\to\R^n$ is an arbitrary curve with $\gamma(a)=r_0$ and $\gamma(b)=r$. By hypothesis this function is well defined. By the fundamental theorem of line integrals it follows \[\int_\gamma\nabla U\cdot ds=U(r)-U(r_0)=-\int_\gamma F\cdot ds-U(r_0).\]In particular, taking $\gamma$ to be the curve $\gamma:[0,t]\to\R^n$ given by $t\mapsto (1-t)r_0+tr$, it follows that \[\int_0^t\nabla(\gamma(t))(r-r_0)dt=-\int_0^tF(\gamma(t))(r-r_0)dt-U(r_0).\]But since $U(r_0)$ is a constant, it follows from the fundamental theorem of calculus that $F=-\nabla U$. 

Conversely, suppose that there exists $U:\R^n\to\R$ such that $F=-\nabla U$. Then for arbitrary $\gamma$ as above, \[-\int_\gamma F\cdot ds=\int_\gamma \nabla U\cdot ds=U(r)-U(r_0)\] That is, $W(r_0\to r,\gamma)=U(r)-U(r_0)$ and so only depends on the end points $r_0$ and $r$. Hence, the work done is path independent showing $F$ is conservative.
\end{proof}

\begin{definition}
For a conservative force $F$, the scalar function $U:\R^n\to\R$ such that $F=-\nabla U$ is called the \bf{potential energy} \rm.
\end{definition}
The reason that a force with this property is called conservative comes from Theorem 3.7 below. 

\begin{definition}
The \bf{total energy} \rm of a conservative system is defined to be \[E=T+U,\] the kinetic energy plus the potential energy.
\end{definition}

\begin{theorem}\bf{(Conservation of Total Energy)} \rm
In a conservative mechanical system, the total energy is conserved. That is, $\dt E=0$.
\end{theorem}

\begin{proof}
Suppose that $F=-\nabla U$. For arbitrary $t_1,t_2\in\R$, by the Work-Kinetic Energy Theorem\[K(\gamma(t_2))-K(\gamma(t_1))=W(\gamma(t_1)\to\gamma(t_2))=\int_\gamma F\cdot ds=\int_\gamma\nabla U\cdot ds=U(\gamma(t_1))-U(\gamma(t_2))\]Hence $(K+U)(\gamma(t_2))=(K+U)(\gamma(t_1))$ and so $E:=K+U$ is independent of $t$.
\end{proof}




In section 6 we will do some interesting computations with the Laplace-Runge-Lenze vector and so we recall here the two-body central force problem.

\begin{definition}
A \bf{central force} \rm on a particle with position vector $\vec r$ is a conservative force for which the corresponding potential energy is only a function of $\|\vec r\|$. 
\end{definition}

The classical example of a central force is one given by a potential of the form \[U(\vec r)=-\frac{k}{\|\vec r\|}\] where $k$ is some constant. Given a system of two particles, (say $\vec r_1$ and $\vec r_2$) in a closed system, by fixing one of the particles and considering the relative position vector $\vec r=\vec r_1-\vec r_2$ we can give explicit formulas for the potential energy corresponding to the gravitational and Coulomb force.  The gravitational potential is given by \[U(\vec r)=-\frac{Gm_1m_2}{\|\vec r\|}\] and the Coulomb by \[U(\vec r)=\frac{1}{4\pi\varepsilon_0}\frac{q_1q_2}{\|\vec r\|}\] where $q_1$ and $q_2$ are the charges of the two particles. Here $G$ and $\varepsilon_0$ are two constants whose explicit values depend on the units being used.

\begin{proposition}
Consider a closed system of two particles moving in $\R^3$ with the standard metric. If the particles are subject to a central force field then the force is always parallel to the relative position of the two particles.
\end{proposition}
\begin{proof}
By definition, $U$ is a function of only $\|\vec r\|$. That is $U=U(\|\vec r\|)$.  In spherical coordinates we have that $\nabla U=\frac{\pd U}{\pd r}\frac{\pd}{\pd r}$ where $\frac{\pd}{\pd r}$ is associated with the vector $\hat r=\frac{\vec r}{\|\vec r\|}$.

\end{proof}

\begin{definition}
If $\vec r$ is the position vector of a particle, then its \bf{angular momentum} \rm is defined to be \[\vec L:=\vec r\times p=m(\vec r\times \dot {\vec r})\]and the \bf{torque} is \[\tau:=\vec r\times F\]
\end{definition}
Using Newton's second law, it immediately follows that torque is the time derivative of angular momentum.
\begin{proposition}
If two particles are subject to a central force field then the angular momentum of their relative position vector is constant.
\end{proposition}
\begin{proof}
By proposition 3.9 the force is parallel to $\vec r$. Hence $\dt \vec L=m(\dot{\vec r}\times \dot {\vec r})+(\vec r\times F)=0$.
\end{proof}

\begin{remark}
By definition, the angular momentum is always orthogonal to the momentum and position vector. In the two-body central force problem, Proposition 3.11 showed that the angular momentum vector is constant. Hence it must be that the plane determined by the momentum and position of the relative position vector is constant. That is, the movement of the two particles is always restricted to a plane. Under translation and rotation, it is no loss of generality to assume that the particles motion is restricted to the $xy$-plane.
\end{remark}

\subsection{The Calculus of Variations and the Euler-Lagrange Equations}
\noindent The calculus of variations studies functionals on a given space $X$. In this section $X$ will be the set of smooth curves $\gamma:[a,b]\to M$, where $M$ is some manifold.  The Euler-Lagrange equations will arise as the extreme points of a specific function, called the action. After deriving the Euler-Lagrange equations we will see how they are a generalization of Newton's second law.

\vspace{0.1cm}


\begin{definition}
Let $M$ be an $n$-dimensional manifold with tangent bundle $TM$. Given $L\in C^\infty(TM)$ the pair $(M,L)$ is called a \bf{Lagrangian system} \rm and $L$ is called the \bf{Lagrangian}.
\end{definition}

\begin{definition}
A Riemannian manifold $(M,g)$ with Lagrangian $L=K-U$ is called a \bf{natural} \rm Lagrangian system. 
\end{definition}

At first, it may seem that the Lagrangian defined in a natural system is random; however, we will see that in these systems the Euler-Lagrange equations are equivalent to Newton's second law. The Euler-Lagrange equations have proven to be more effective than Newton's second law in many different natural Lagrangian systems. 


\begin{definition} 
Given a smooth curve $\gamma:[a,b]\to M$ we get an induced curve $\widetilde\gamma$, called the \bf{lift} \rm of $\gamma$, into the tangent bundle defined by\[\widetilde\gamma:[a,b]\to TM \ \ \ \ \ \ \ t\mapsto (\gamma(t),\gamma^\prime(t))\]
\end{definition}
\begin{definition}
Let $\mathcal C =\left\{\gamma:[a,b]\to M \ ; \ \gamma \text{ is smooth}\right\}$. Consider the function $\mathcal A:\mathcal C\to\R$ defined by \[\mathcal A(\gamma):=\mathcal A_{\gamma}:=\int_a^b\widetilde\gamma^\ast L(t)dt=\int_a^bL(\gamma(t),\gamma^\prime(t))dt\]

\noindent The function $\mathcal A:\mathcal C\to\R$ is called the \bf{action} \rm of the Lagrangian system.
\end{definition}

The goal of this section is to find curves $\gamma:[a,b]\to M$ which are critical points for $\mathcal A$. Just as minimum and maximum points are extreme points in elementary calculus, we seek to find curves for which the ``derivative" of the action vanishes. Intuitively, for a curve $\gamma:[a,b]\to M$ to be a minimum we need  the value of $\mathcal A$ to be no greater on $\gamma$ than on curves `close' to $\gamma$, say `within $\varepsilon$', as pictured below.
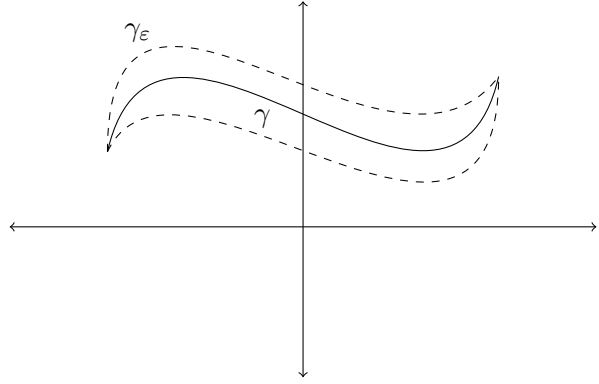
\vspace{-1cm}
\[
\begin{tikzpicture}[xscale=1.3]
\draw (0,0) .. controls +(80:3) and +(260:3) .. node[below=1pt,pos=.42] {$\gamma$} (4,1);
\draw[dashed] (0,0) .. controls +(90:3.5) and +(240:2.2) .. node[above=2pt,pos=.2] {$\gamma_{\varepsilon}$} (4,1);
\draw[dashed] (0,0) .. controls +(70:2) and +(270:3.5) .. node[below=2pt,pos=.8]{$\ $} (4,1);
\draw[shift={(2,-1)},<->] (-3,0)--(3,0);
\draw[shift={(2,-1)},<->] (0,-2)--(0,3);
\end{tikzpicture}
\]
More precisely, fix a coordinate chart $(U,x^1,\dots, x^n)$ in $M$ and consider $(TU,x^1,\dots,x^n,v^1,\dots,v^n)$, the induced chart  on $TM$.  Let $\gamma:[a,b]\to U$ be a curve. Given $\varepsilon\in\R$, pick arbitrary $c^1,\dots,c^n\in C^\infty([a,b])$ with $c^i(a)=0=c^i(b)$. Define $\gamma_\varepsilon(t):=(\gamma^1(t)+\varepsilon c^1(t),\dots,\gamma^n(t)+\varepsilon c^n(t))$. Note that we can choose $\varepsilon$ small enough so that $\gamma_\varepsilon$ is contained in $U$. We have that $\gamma_0=\gamma$ and so if $\gamma$ is a minimum of $\mathcal A$ then \[\left.\frac{d}{d\varepsilon}\right|_{\varepsilon=0} \mathcal A_{\gamma_\varepsilon}=0\]On the other hand, the chain rule, product rule and fundamental theorem of calculus give that
\begin{align*}
\left.\frac{d}{d\varepsilon}\right|_{\varepsilon=0}\mathcal A_{\gamma_\varepsilon}&=\frac{d}{d\varepsilon}\int_a^bL(\gamma_\varepsilon(t),\gamma_\varepsilon^\prime(t))dt\\
&=\int_a^b\left(\left(\left.\frac{\pd L}{\pd x^i}\right|_{\widetilde\gamma(t)}\right)c^i(t)+\left(\left.\frac{\pd L}{\pd v^i}\right|_{\widetilde\gamma(t)}\right)\frac{d}{dt}c^i(t)\right)dt&\hspace{-2.5cm}\\
&=\int_a^b\left(\left(\left.\frac{\pd L}{\pd x^i}\right|_{\widetilde\gamma(t)}\right)c^i(t)+\frac{d}{dt}\left(\left(\left.\frac{\pd L}{\pd v^i}\right|_{\widetilde\gamma(t)}\right)c^i(t)\right)-\frac{d}{dt}\left(\left.\frac{\pd L}{\pd v^i}\right|_{\widetilde\gamma(t)}\right)c^i(t)\right)dt\\
&=\int_a^b\left(\left(\left.\frac{\pd L}{\pd x^i}\right|_{\widetilde\gamma(t)}\right)c^i(t)-\dt\left(\left.\frac{\pd L}{\pd v^i}\right|_{\widetilde\gamma(t)}\right)c^i(t)\right)dt+\left[\left(\left.\frac{\pd L}{\pd v^i}\right|_{\widetilde\gamma(t)}\right)c^i(t)\right]_a^b\\
&=\int_a^b\left(\left(\left.\frac{\pd L}{\pd x^i}\right|_{\widetilde\gamma(t)}\right)c^i(t)-\dt\left(\left.\frac{\pd L}{\pd v^i}\right|_{\widetilde\gamma(t)}\right)c^i(t)\right)dt\\
&=\int_a^b\left[\left(\left.\frac{\pd L}{\pd x^i}\right|_{\widetilde\gamma(t)}\right)-\dt\left(\left.\frac{\pd L}{\pd v^i}\right|_{\widetilde\gamma(t)}\right)\right]c^i(t)
\end{align*}

But since this expression is equal to zero for all $c^1(t),\dots,c^n(t)$ with $c^i(a)=c^i(b)=0$, the \emph{Fundamental Lemma of Calculus of Variations} (see \cite{Arnold} page 57) implies \[\left.\frac{\pd L}{\pd x^i}\right|_{\widetilde\gamma(t)}-\dt\left(\left.\frac{\pd L}{\pd v^i}\right|_{\widetilde\gamma(t)}\right)=0 \ \ \ \text{ for all } i=1,\cdots, n\]These $n$ second order ODE's are called the \bf{Euler-Lagrange} \rm equations. Hence we have shown that a necessary condition for a curve to minimize the action is that it needs to satisfy the Euler-Lagrange equations. 

\begin{remark}If the Lagrangian is strictly convex, meaning for fixed $x\in M$, arbitrary $v,w\in T_x M$ and $0<t<1$ we have $L(x,tw+(1-t)v)<tL(x,w)+(1-t)L(x,v)$, then the converse is locally true. That is, if a curve $\gamma:[a,b]\to M$ satisfies the Euler-Lagrange equations, then there exists a subinterval $[a_1,b_1]\subset[a,b]$ such that $\gamma|_{[a_1,b_1]}$ minimizes the action. See \cite{Da Silva} page 117 for a proof of this. 
\end{remark}

\begin{remark} In the above derivation of the Euler-Lagrange equations, it was assumed that the Lagrangian was time independent. However, there are many situations in which the Lagrangian does depend on time, some of which we will see in subsequent sections. But notice that even if the Lagrangian were of the form $L(x^1,\dots,x^n,v^1,\dots,v^n,t)$ the above calculation would be exactly the same and the Euler-Lagrange equations derived above would not change.
\end{remark}

 



\subsection{Examples of Lagrangian Systems}
\vspace{0.2cm}

\begin{example}\bf{The Action as the Length Functional} \rm
\vspace{0.2cm}

In the case that we are working in a natural system for which the net force is zero, our Lagrangian reduces to \[L=K:TM\to\R \ \ \ \ \ \ \ (p,V_p)\mapsto \frac{1}{2}mg_p(V_p,V_p).\]In this case $\mathcal A(\gamma)$ is just a constant times the length of $\gamma$. That is, \[\mathcal A(\gamma)=\frac{1}{2}m\int_\gamma g_{\gamma(t)}(\gamma^\prime(t),\gamma^\prime(t))dt.\]A standard result from Riemannian geometry is that the critical points of the length functional are geodesics. Hence any geodesic is a critical point of $\mathcal A$. In particular, if the Riemannian manifold is $\R^n$ with the standard metric, and the net force is zero, then we get that the solutions of the Euler-Lagrange equations are straight lines. That is, the shortest path between two points is a straight line. Using this mechanical system we will give, in the section on Hamiltonian mechanics, another interpretation of geodesic flow.
\end{example}

\begin{example}\bf{(Natural System with Standard Coordinates)}\rm
\vspace{0.2cm}

Consider a natural Lagrangian system in $\R^2$ with the standard metric. Let $x^1,x^2$ denote the standard coordinates and let $x^1,x^2,v^1,v^2$ be the induced coordinates on $T\R^2=\R^4$.  In this coordinate system we have that $g=(dx^1)^2+(dx^2)^2$. Suppose that a particle of mass $m$ is moving in $\R^2$ under a conservative force field $F=-\nabla U$. By definition the Lagrangian $L:\R^4\to\R$ is defined by  \begin{align*}L(x^1,x^2,v^1,v^2)&:= K(x^1,x^2,v^1,v^2)-U(x^1,x^2,v^1,v^2)\\
&=\frac{1}{2}m(v^1)^2+\frac{1}{2}m(v^2)^2-U(x^1,x^2).\end{align*}The elements of the Euler-Lagrange equations are \[\frac{\pd L}{\pd x^1}(\widetilde\gamma(t))=-\frac{\pd U}{\pd x^1}(\gamma(t)) \ \ ,\ \ \frac{\pd L}{\pd x^2}(\widetilde\gamma(t))=-\frac{\pd U}{\pd x^2}(\gamma(t))\]and \[\dt\left(\frac{\pd L}{\pd v^1}(\widetilde\gamma(t))\right)=\dt\left(mv^1(\widetilde\gamma(t))\right)=m\ddot\gamma^1(t) \ \ ,\ \ \dt\left(\frac{\pd L}{\pd v^2}(\widetilde\gamma(t))\right)=\dt\left(mv^2(\widetilde\gamma(t))\right)=m\ddot\gamma^2(t).\]

Hence, in this setting, the Euler-Lagrange equations are equivalent to Newton's second law. This equivalence easily extends to conservative force fields on $\R^n$. The same result also holds for $k$ particles moving in $\R^n$. To see this,  just take the manifold to be $\R^{kn}$ so that the motion of the $k$-particles can be described by the motion of one particle. 
\vspace{0.1cm}

\begin{example}\bf{(Natural System with Polar Coordinates)}\rm
\vspace{0.1cm}

Consider the setup of the previous example, but with polar coordinates $(r,\theta)$. Let $r,\theta,\widetilde r,\widetilde \theta$ denote the induced coordinates on $T\R^2=\R^4$. By the chain rule we have $\frac{\pd}{\pd r}=\cos\theta\frac{\pd}{\pd x^1}+\sin\theta\frac{\pd}{\pd x^2}$ and $\frac{\pd}{\pd\theta}=-r\sin\theta\frac{\pd}{\pd x^1}+r\cos\theta\frac{\pd}{\pd x^2}$. It follows that in polar coordinates \[g=\left[\begin{array}{c c} 1&0\\ 0 &r^2\end{array}\right] \ \ \text{ and } \ \ g^{-1}=\left[\begin{array}{c c} 1&0\\ 0 &\frac{1}{r^2}\end{array}\right]\]In these coordinates, the velocity of $\gamma(t)=(r(t),\theta(t))$ is given by $\dot r\frac{\pd}{\pd r}+\dot\theta\frac{\pd}{\pd\theta}$. Therefore the kinetic energy of the particle at time $t$ is 
\begin{align*}
K(\widetilde\gamma(t))&=K(r,\theta,\dot r,\dot\theta)\\
&=\frac{1}{2}mg\left(\dot r\frac{\pd}{\pd r},\dot\theta\frac{\pd}{\pd \theta}\right)\\
&=\frac{1}{2}m(\dot r)^2+\frac{1}{2}mr^2(\dot\theta)^2
\end{align*}

Since $L=K-U$ it follows that \[\frac{\pd L}{\pd r}=mr\dot\theta^2-\frac{\pd U}{\pd r} \ \ \text{ and } \ \ \frac{\pd L}{\pd\theta}=-\frac{\pd U}{\pd\theta}\] while 

\[\dt\left(\left.\frac{\pd L}{\pd\widetilde r}\right|_{\widetilde\gamma(t)}\right)=m\ddot r \ \ \text{ and } \ \ \dt\left(\left.\frac{\pd L}{\pd\widetilde\theta}\right|_{\widetilde\gamma(t)}\right)=\dt(mr^2\dot\theta)=mr^2\ddot\theta+2mr\dot r\dot\theta\]

\noindent We have that $F=-\nabla U=-\frac{\pd U}{\pd r}\frac{\pd}{\pd r}-\frac{1}{r^2}\frac{\pd U}{\pd\theta}\frac{\pd}{\pd\theta}$ and by definition, the forces in the $r$ and $\theta$ directions are 

\[F_r=\frac{F\cdot \frac{\pd}{\pd r}}{\left|\frac{\pd}{\pd r}\right|}=-\frac{\pd U}{\pd r} \ \ \text{ and } \ \ F_\theta=\frac{F\cdot\frac{\pd}{\pd\theta}}{\left|\frac{\pd}{\pd\theta}\right|}=-\frac{1}{r}\frac{\pd U}{\pd\theta}\]

\noindent Combining these equalities the first Euler Lagrange equation gives that \[F_r=m(\ddot r-r\dot\theta^2)\] which is the $r$-component of Newtons second law, while the second equation says \[rF_\theta=-\frac{\pd U}{\pd\theta}=\dt\left(mr^2\dot\theta\right)=mr^2\ddot\theta+2mr\dot r\dot\theta\]which is precisely the statement that the torque, $rF_\theta$, is the derivative of angular momentum, $mr^2\dot\theta$.
\end{example}
\vspace{0.05cm}

\begin{remark}
In the above examples we showed that a motion is determined by solving the Euler-Lagrange equations. However, we do not know if this solution is a maximum or a minimum. To prove that a solution is a maximium or a minimum usually requires some extra work. However, the above examples do demonstrate \bf{Hamilton's Principle} \rm which is that the path a particle follows is a critical point of the action $\mathcal A$.
\end{remark}

\end{example}
\noindent These examples lead to the following definition.

\begin{definition}
In a Lagrangian system $(M,L)$, a curve $\gamma: [a,b]\to M$ is called a \bf{motion} \rm if $\widetilde\gamma(t)$ satisfies the Euler-Lagrange equations.
\end{definition}






Since the Euler-Lagrange equations were derived from a statement about curves, independent of the coordinate system chosen, this means that if the equation holds in one coordinate system, they hold in any other. However, we can also prove this rigorously. 

\begin{proposition}
Let $(M,L)$ be a Lagrangian system. If the Euler-Lagrange equations hold in one coordinate chart then they hold in all coordinate charts.
\end{proposition}

\begin{proof}
Let $(U,x^1,\dots,x^n)$ and $(V,\widetilde x^1,\dots,\widetilde x^n)$ denote two arbitrary coordinate charts on $M$ with non-trivial intersection and let $x^1,\dots,x^n,v^1,\dots,v^n$ and $\widetilde x^1,\dots,\widetilde x^n,\widetilde v^1,\dots,\widetilde v^n$ denote the induced coordinates on $TU$ and $TV$ respectively. Suppose that $\gamma$ is a motion in $(M,L)$. That is \[\left.\frac{\pd L}{\pd x^i}\right|_{\widetilde\gamma(t)}=\dt\left(\left.\frac{\pd L}{\pd v^i}\right|_{\widetilde\gamma(t)}\right).\]The chain rule shows that $\frac{\pd}{\pd x^i}=\frac{\pd \widetilde x^j}{\pd x^i}\frac{\pd}{\pd \widetilde x^j}$ implying that $\widetilde v^i=\frac{\pd\widetilde x^i}{\pd x^j}v^j$. Using these expressions we get \begin{align}\frac{\pd\widetilde v^i}{\pd v^l}=\frac{\pd\widetilde x^i}{\pd x^j}\delta_{l}^j=\frac{\pd\widetilde x^i}{\pd x^l}\end{align} and \begin{align}\frac{\pd\widetilde v^i}{\pd x^l}=\frac{\pd ^2\widetilde x^i}{\pd x^j\pd x^l}v^j.\end{align}
It follows that \begin{align*}
\left.\frac{\pd L}{\pd x^i}\right|_{\widetilde\gamma(t)}&=\left.\left(\frac{\pd L}{\pd\widetilde x^j}\frac{\pd\widetilde x^j}{\pd x^i}+\frac{\pd L}{\pd\widetilde v^j}\frac{\pd\widetilde v^j}{\pd x^i}\right)\right|_{\widetilde\gamma(t)}\\
&=\left.\left(\frac{\pd L}{\pd\widetilde x^j}\frac{\pd\widetilde v^j}{\pd v^i}+\frac{\pd L}{\pd\widetilde v^j}\frac{\pd ^2\widetilde x^j}{\pd x^l\pd x^i}v^l\right)\right|_{\widetilde\gamma(t)}&\text{by $(3.1)$ and $(3.2)$}\\
&=\left.\left(\frac{\pd L}{\pd\widetilde x^j}\frac{\pd\widetilde v^j}{\pd v^i}+\frac{\pd L}{\pd\widetilde v^j}\frac{\pd ^2\widetilde x^j}{\pd x^l\pd x^i}\frac{d x^l}{dt}\right)\right|_{\widetilde\gamma(t)}&\text{since we are substituting $\widetilde\gamma(t).$}
\end{align*}We also have that 
\begin{align*}
\dt\left(\left.\frac{\pd L}{\pd v^i}\right|_{\widetilde\gamma(t)}\right)&=\dt\left(\left.\frac{\pd L}{\pd\widetilde v^j}\frac{\pd\widetilde v^j}{\pd v^i}\right|_{\widetilde\gamma(t)}\right)\\
&=\dt\left(\left(\left.\frac{\pd L}{\pd\widetilde v^j}\frac{\pd\widetilde x^j}{\pd x^i}\right)\right|_{\widetilde\gamma(t)}\right)&\text{by $(3.1)$}\\
&=\dt\left(\left.\frac{\pd L}{\pd \widetilde v^j}\right|_{\widetilde\gamma(t)}\right)\frac{\pd \widetilde x^j}{\pd x^i}+\left.\frac{\pd L}{\pd\widetilde v^j}\right|_{\widetilde\gamma(t)}\left(\frac{\pd^2 \widetilde x^j}{\pd x^l\pd x^i}\frac{d x^l}{dt}\right)&\text{by the product rule}\\
&=\dt\left(\left.\frac{\pd L}{\pd \widetilde v^j}\right|_{\widetilde\gamma(t)}\right)\frac{\pd \widetilde v^j}{\pd v^i}+\left.\frac{\pd L}{\pd\widetilde v^j}\right|_{\widetilde\gamma(t)}\left(\frac{\pd^2 \widetilde x^j}{\pd x^l\pd x^i}\frac{d x^l}{dt}\right)&\text{by $(3.1).$}
\end{align*}
By combining these two calculations it follows that \[\left.\frac{\pd L}{\pd x^i}\right|_{\widetilde\gamma(t)}-\dt\left(\left.\frac{\pd L}{\pd v^i}\right|_{\widetilde\gamma(t)}\right)=\left(\left.\frac{\pd L}{\pd \widetilde x^j}\right|_{\widetilde\gamma(t)}-\dt\left(\left.\frac{\pd L}{\pd\widetilde v^j}\right|_{\widetilde\gamma(t)}\right)\right)\frac{\pd \widetilde v^j}{\pd v^i}\]

\end{proof}

To make a comparison between Lagrangian and Hamiltonian mechanics, we compute here the Euler-Lagrange equations for the simple pendulum and revisit the calculation in the next section.

\begin{example}\bf{The Simple Pendulum} \rm
\vspace{0.1cm}

The simple pendulum is illustrated below.

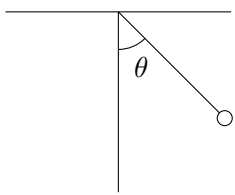
\[
\begin{tikzpicture}
\draw (-1.5,0)--(1.5,0) (0,0)--(0,-2.4);
\draw (0,0)--+(315:2);
\draw[fill=white] ($(0,0)+(315:2)$) circle (.1);
\draw (0,-.5) arc (270:315:.5);
\node at ($(0,0)+(292.5:.8)$) {$\theta$};
\end{tikzpicture}
\]

 We have a mass $m$ attached to a weightless rod of length $l$. That is, we are working in the $1$-dimensional submanifold $(S^1,\theta)$ of $(\R^2,r,\theta)$ which we endow with the standard metric.  By Newton's second law we know that the equation of motion for the mass is \[\ddot\theta=-\frac{g}{l}\sin\theta\] Let $\theta,\widetilde\theta$ denote the induced coordinates on $T^\ast S^1\cong S^1\times\R$. As in example 3.21 we have that $K=\frac{1}{2}ml^2(\widetilde\theta)^2$. A standard calculation gives that the net force $F=-\frac{mg}{l}\sin\theta$ is conservative with potential energy $U=mgl(1-\cos\theta)$.  With the natural Lagrangian $L=K-U$, let $\alpha(t)$ be a motion in the Lagrangian system $(S^1,L)$. That is, suppose $\widetilde\alpha(t) =(\alpha(t),\alpha^\prime(t))$ satisfies the Euler-Lagrange equations, i.e.
\[\left.\frac{\pd L}{\pd\theta}\right|_{\widetilde\alpha(t)}=\dt\left(\left.\frac{\pd L}{\pd\widetilde\theta}\right|_{\widetilde\alpha(t)}\right) \ \ \Longrightarrow \ \ -mgl\sin(\alpha(t))=\dt\left(ml^2\alpha^\prime(t)\right)=ml^2\alpha^{\prime\prime}(t)\]
This is precisely the statement that the torque exerted by gravity on the pendulum, $-mgl\sin\alpha(t)$, is the product of the moment of inertia, $ml^2$, and angular acceleration, $\alpha^{\prime\prime}(t)$.
\end{example}

\begin{remark}In the above example we found the equations of motion without computing the tension of the rope, the constraint force. Although this example is simple enough to solve using Newton's second law, it gives a glimpse into how the Euler-Lagrange equations can be used to simplify other complicated systems with constraint forces.
\end{remark}
\newpage
\section{Hamiltonian Mechanics}

The next level of formality in studying mechanics is done using Hamilton's formulation. This approach is done through the  framework of symplectic geometry, which was introduced in section 2. In section 2.4 we saw how the cotangent bundle has a canonical symplectic structure. Where Lagrangian mechanics studies curves living in $TM$, Hamilton's approach instead studies curves in the symplectic manifold $T^\ast M$. We will see that the outcomes predicted in Hamiltonian mechanics agree with those from Lagrangian mechanics in situations where both can be applied.
\subsection{Hamiltonian Vector Fields}

\begin{definition}
A triple $(X,\omega, H)$ where $(X,\omega)$ is a symplectic manifold and $H\in C^\infty(X)$ is called a \bf{Hamiltonian system} \rm and $H$ is called the associated \bf{Hamiltonian function}\rm.
\end{definition}

\noindent Fix a Hamiltonian system $(X,\omega, H)$. We have that $dH\in \Gamma(T^\ast X)$ and by Proposition 2.5 there exists a corresponding vector field $V_H\in\Gamma(TX)$, the symplectic dual of $dH$. Notice that any $G\in C^\infty(X)$ induces a vector field $V_G$ in this way.


\begin{definition}
Given $G\in C^\infty(X)$, the vector field $V_G$ is called the \bf{Hamiltonian vector field} \rm associated to $G$.
\end{definition}

We will need the following propositions in section 6.

\begin{proposition}The flow $\theta_t$ of $V_H$ preserves $\omega$ (i.e. $\theta_t^\ast\omega=\omega$).
\end{proposition}
\begin{proof} By Cartan's magic formula together with the closedness of $\omega$ it follows that \[\mathcal L_{V_H}\omega=d(V_H\hk\omega)+V_H\hk d\omega=d(V_H\hk\omega)=d(dH)=0\]
\end{proof}
\vspace{-0.5cm}
\begin{proposition}The Hamiltonian function of a Hamiltonian vector field is constant on its flows. That is, $H\circ\theta_t =H$ on the domain of $\theta_t$. 
\end{proposition}

\begin{proof}
By the antisymmetry of the interior product\[\L_{V_H}H=V_HH=V_H\hk dH=V_H\hk(V_H\hk\omega)=0\]
\end{proof}
\vspace{-0.5cm}
The following is an illustration of this result.
\begin{example} \bf{(Height Function on $S^2$)}\rm
\vspace{0.2cm}

The sphere $S^2$ is a symplectic manifold when equipped with the local $2$-form $d\theta\wedge dh$, where $\theta$ is a local coordinate for $S^1$ and $h$ is the $x^3$-coordinate of $\R^3$. Let $H(\theta,h)=h$. It's clear that $dH=dh$ and since $\omega(V_H,\cdot)=(d\theta\wedge dh)(V_H,\cdot)=dh$ it must be that $V_H=\frac{\pd}{\pd\theta}$. The vector field $V_H=\frac{\pd}{\pd\theta}$ has flow $\rho_t:S^2\to S^2$ given by $(\theta,h)\mapsto (\theta+t,h)$, which is clearly constant under the height function. That is, $H\circ\rho_t=H$.
\end{example}

By the closedness of $\omega$ we have $\mathcal L_W\omega=d(W\hk\omega)$ for any $W\in\Gamma(TX)$. We see that the flow of an arbitrary $W\in\Gamma(TX)$ preserves $\omega$ if $W\hk\omega$ is closed. This motivates the following definition.

\begin{definition}
A vector field $W\in\Gamma(TX)$ is said to be \bf{Hamiltonian} \rm if $W\hk \ \omega$ is exact and is said to be \bf{symplectic} \rm if $W\hk \ \omega$ is closed.
\end{definition}
By definition all Hamiltonian vector fields are symplectic. However, the converse is not true.

\begin{example}\bf{(Symplectic Vector Field that is not Hamiltonian)}\rm
\vspace{0.1cm}

The  $2$-torus $\mathbb{T}^2$ is a symplectic manifold when equipped with local $2$-form $d\theta\wedge d\varphi$, where $\theta$ and $\varphi$ are two different local coordinates for $S^1$. The vector field $\frac{\pd}{\pd\theta}$ on $\mathbb{T}^2$ is symplectic but not Hamiltonian since $\frac{\pd}{\pd\theta}\hk (d\theta\wedge d\varphi)=d\varphi$ is closed but not exact. The $1$-form $d\varphi$ is locally exact, but not globally exact since $\varphi$ is only defined on a proper open subset of $S^1$. The same can be said about the vector field $\frac{\pd}{\pd\varphi}$.
\end{example}

\begin{remark}The Lie bracket $[\cdot,\cdot]$ turns the subspace of Hamiltonian vector fields into a Lie algebra. In fact a stronger result holds; the Lie bracket of any two symplectic vector fields is Hamiltonian.  Indeed, if $X,Y$ are symplectic then, using the identity\[[\L_X,\imath_Y]=\imath_{[X,Y]}=[\imath_X,\L_Y]\]it follows that \begin{align*}[X,Y]\hk\omega&=\imath_{[X,Y]}\omega\\&=\L_X\imath_Y\omega-\imath_Y\L_X\omega\\&=d(X\hk Y\hk\omega)+X\hk d(Y\hk\omega)-Y\hk d(X\hk\omega)\\&=d(\omega(Y,X))\end{align*}
\end{remark}

Just as the Euler-Lagrange equations determine the motions a Lagrangian system must satisfy, curves satisfying the Hamilton equations give the motions in a Hamiltonian system. We will see why this is below.

\subsection{Hamilton's Equations}

By Darboux's theorem we can find local coordinates $(q^1,\dots,q^n,p_1,\dots, p_n)$ in $X$ such that $\omega=dq^j\wedge dp_j$. We have that $V_H=A^i\frac{\pd}{\pd q^i}+B_i\frac{\pd}{\pd p_i}$ for some $A^i, B_i\in C^\infty(X)$. It follows that $V_H\hk\omega=A^idp_i-B_idq^i$. Since $dH=\frac{\pd H}{\pd q^i}dq^i+\frac{\pd H}{\pd p_i}dp_i$, it must be that $A^i=\frac{\pd H}{\pd p_i}$ and $B_i=-\frac{\pd H}{\pd q^i}$. Hence any integral curve $\gamma(t)=(\alpha(t),\beta(t))$ of $V_H$ must satisfy
\[
\begin{array}{l}
\dt \alpha^i(t)=\left(\frac{\pd}{\pd p_i}H\right)(\gamma(t))\\
\ \ \\
\dt \beta_i(t)=-\left(\frac{\pd}{\pd q^i}H\right)(\gamma(t))
\end{array}
\]
These are the \bf{Hamilton equations}\rm, a system of $2n$ first order ODE's. That is, a curve $\gamma:\R\to X$ is an integral curve for $V_H$ if and only if $\gamma(t)$ satisfies Hamilton equations. We have shown that \[V_H=\frac{\pd H}{\pd p_i}\frac{\pd}{\pd q^i}-\frac{\pd H}{\pd q^i}\frac{\pd}{\pd p_i}\]

\begin{definition}
An integral curve $\gamma:\R\to X$ of $V_H$ is called a \bf{motion} \rm of the Hamiltonian system $(X,\omega,H)$. That is, $\gamma$ is a motion if and only if $\gamma$ satisfies Hamilton's equations. 
\end{definition}

A rather basic example of a Hamiltonian system is the simple pendulum. This was discussed in the Lagrangian setting as Example 3.25

\begin{example}\bf{(The Simple Pendulum)} \rm

Consider the  $1$-dimensional manifold $S^1$. Let $(V,\theta)$ be a chart in $S^1$. We know that the cotangent bundle $(T^\ast S^1,\theta,\xi)$ is a symplectic manifold with symplectic $2$-form $\omega=d\theta\wedge d\xi$. Consider the function $K:\R\to\R$ given by $\xi\mapsto \frac{\xi^2}{2ml^2}$ and $V:U\to\R$ given by $\theta\mapsto gml(1-\cos\theta)$. Define the Hamiltonian to be  \[H:T^\ast S^1\to\R \ \ \ \ \ \ \ (\theta,\xi)\mapsto K(\xi)+U(\theta).\] At first this definition of the Hamiltonian may seem ad hoc; however, after introducing the Legendre transform we will see where it comes from. In fact, the choice for naming the functions $K$ and $U$ above is to indicate that the Hamiltonian is to be thought of, in this setting, as the total energy. It follows that the Hamiltonian vector field is \[V_H=\frac{\pd H}{\pd\xi}\frac{\pd}{\pd\theta}-\frac{\pd H}{\pd \theta}\frac{\pd}{\pd\xi}=mgl\sin\theta\frac{\pd}{\pd\theta}-\frac{\xi}{ml^2}\frac{\pd}{\pd\xi}\]Hence if $\gamma(t)=(\alpha(t),\beta(t))$, where $\alpha(t)=\theta(\gamma(t))$ and $\beta(t)=\xi(\alpha(t))$, is a motion in this Hamiltonian system then $\gamma(t)$ satisfies the Hamilton equations \[ \dt(\beta(t))=\left.\frac{\pd H}{\pd\xi}\right|_{\gamma(t)} \ \ \Longrightarrow \ \ \beta^\prime(t)=mgl\sin(\alpha(t))\]and \[\dt(\alpha(t))=-\left.\frac{\pd H}{\pd\theta}\right|_{\gamma(t)} \ \ \Longrightarrow \ \ \alpha^\prime(t)=-\frac{\beta(t)}{ml^2}\]Combining these two equations it follows that \[\alpha^{\prime\prime}(t)=-\frac{g}{l}\sin\theta\]which is precisely the motion of the pendulum as prescribed by Newton's second law. That is, the integral curves of the Hamiltonian vector field give the motions of this mechanical system.
\end{example} 
\subsection{The Poisson Bracket}

Given a symplectic manifold $(X,\omega)$, the Poisson bracket turns $C^\infty(X)$ into a Lie algebra. Suppose that the symplectic manifold also has a Hamiltonian $H\in C^\infty(X)$. In the same way the Lie bracket measures commutativity of vector fields, the Poisson bracket measures the commutativity of functions with the Hamiltonian vector field. 

\begin{definition}
Given $f,g\in C^\infty(X)$ their \bf{Poisson bracket} \rm is defined to be \[ \{f,g\}:=\omega(V_f,V_g).\]
\end{definition}

\noindent By definition \[\{f,g\}:=\omega(V_f,V_g)=(V_f\hk \ \omega)(V_g)=df(V_g)=V_gf.\] Using this calculation, is not hard to verify that $\{\cdot,\cdot\}$ does indeed turn $C^\infty(X)$ into a Lie algebra. Moreover, there is a Leibniz rule; \[\{f,gh\}=g\{f,h\}+\{f,g\}h.\]A straightforward computation shows that $C^\infty(X)\owns H\mapsto V_H\in\Gamma(TX)$ is a Lie algebra anti-homomorphism.

\begin{proposition}For $f\in C^\infty(X)$ we have $\{f, H\}=0 \iff f\circ\theta_t=f$ on the domain of $\theta_t$, where $\theta_t$ is the flow of $V_H$. That is, $\{f,H\}=0$ if and only if $f$ is constant along the integral curves of $V_H$.
\end{proposition}

\begin{proof} 
\begin{align*}
f\circ\theta_t=f&\iff \theta_t^\ast f=f\\
&\iff \mathcal L_{V_H}f=0\\
&\iff V_Hf=0\\
&\iff (df)(V_H)=0\\
&\iff (V_f\hk\omega)(V_H)=0\\
&\iff \omega(V_H,V_f)=0\\
&\iff\{H,f\}=0=\{f,H\}
\end{align*}
\end{proof}

\begin{definition}
A function $f$ such that $\{f,H\}=0$ is called an \bf{integral of motion}.\rm 
\end{definition}

By definition, if two functions $f_1,f_2$ commute 
(with respect to $\{\cdot,\cdot\}$) then $\omega(X_{f_1},X_{f_2})=0$, showing that given a collection of commuting integrals of motion, they generate an isotropic subspace of $T_pM$.  But in the paragraph following Definition 2.6, we showed that given a subspace $Y$ of a vector space $V$ we have $\dim V=\dim Y+\dim Y^\Omega$. If $Y$ is isotropic then $Y\subset Y^\Omega$ so that an isotropic subspace has dimension at most half the dimension of $M$. A Hamiltonian system is called \bf{(completely) integrable} \rm if there exists $n$ Poisson commuting independent integrals of motion $f_1=H,f_2,\dots, f_n$. 

\begin{theorem}\bf{(Arnold-Lioville Theorem, \cite{Da Silva})}\rm
\vspace{0.1cm}

Let $(X,\omega,H)$ be a completely integrable Hamiltonian system of dimension $2n$ with integrals of motion $f_1=H,f_2,\dots,f_n$. Consider the function $f:TX\to\R^n$ defined by $f:=(f_1,\dots,f_n)$. Let $c\in \R^n$ be a regular value of $f$. That is, $c$ is a point in $\R^n$ such that for every point in $f^{-1}(c)$ the differential of $f$ is surjective. Then $f^{-1}(c)$ is a Lagrangian submanifold of $X$. Moreover we have
\begin{enumerate}[(a)]

\item If the flows of the Hamiltonian vector fields $X_{f_1},\dots, X_{f_n}$ starting at a point $p\in f^{-1}(c)$ are complete, then the connected component of $f^{-1}(c)$ containing $p$ is a homogeneuous space for $\R^n$. This connected component has local coordinates $\varphi_1,\dots,\varphi_n$, called \bf{angle coordinates}\rm, for which the flows of $X_{f_1},\dots, X_{f_n}$ are linear.

\item There exists coordinates $\psi_1,\dots,\psi_n$, called \bf{action coordinates}\rm, such that each $\psi_i$ is an integral of motion and also such that $\varphi_1,\dots,\varphi_n,\psi_1,\dots,\psi_n$ form Darboux coordinates. 
\end{enumerate}

\end{theorem}
\begin{proof}
This is Theorem 18.12 in \cite{Da Silva}. A proof of part (a) can be found in \cite{Da Silva}, page 110.  For a proof of part (b) see \cite{Arnold}, pages 271-274 and 279-281. 
\end{proof}

There is a rather simple expression of the Poisson bracket in Darboux coordinates.
\begin{proposition}
Let $(M,\omega)$ be a symplectic manifold and let $(U,q^1,\dots, q^n,p_1,\dots,p_n)$ be a Darboux chart. Then for $f,g\in C^\infty (M)$ we have $\{f,g\}=\frac{\pd g}{\pd p_i}\frac{\pd f}{\pd q^i}-\frac{\pd f}{\pd p_i}\frac{\pd g}{\pd q^i}$.
\end{proposition}
\begin{proof}In the above Darboux coordinates we have that $\omega=dq^i\wedge dp_i$. Hence
\begin{align*}
\{f,g\}&:=\omega(V_f,V_g)\\
&=(dq^i\wedge dp_i)\left(\frac{\pd f}{\pd p_i}\frac{\pd}{\pd q^i}-\frac{\pd f}{\pd q^i}\frac{\pd}{\pd p_i} \ , \ \frac{\pd g}{\pd p_i}\frac{\pd}{\pd q^i}-\frac{\pd g}{\pd q^i}\frac{\pd}{\pd p_i}\right)\\
&=-\frac{\pd f}{\pd p_i}\frac{\pd g}{\pd q^i}-(-\frac{\pd g}{\pd p_i}\frac{\pd f}{\pd q^i})\\
&=\frac{\pd g}{\pd p_i}\frac{\pd f}{\pd q^i}-\frac{\pd f}{\pd p_i}\frac{\pd g}{\pd q^i}
\end{align*}
\end{proof}

\newpage

\section{The Legendre Transform}

We have seen the definition of motions in a Lagrangian system $(M,L)$. These are just curves whose derivatives are solutions to the Euler-Lagrange equations. We have also seen how motions are defined in a Hamiltonian system $(X,\omega, H)$. These are just curves  that satisfy the Hamilton equations. But given a Lagrangian system $(M,L)$ we can always consider the cotangent bundle to get a symplectic manifold $(T^\ast M,\omega)$. The Legendre transform provides the link between Lagrangian mechanics in $(M,L)$ and Hamiltonian mechanics in $(T^\ast M,\omega=-d\alpha, H)$, where $H\in C^\infty (T^\ast M)$ will be defined below. Conversely, given a Hamiltonian system on a manifold which is a cotangent bundle, $(T^\ast M,\omega=-d\alpha, H)$ we can consider, under the Legendre transform, motions in an induced Lagrangian system $(M,L)$. In section 6 we will see how the Legendre transform relates the two statements of Noether's theorem.




\subsection{The Legendre Transform on a Vector Space}
Let $V$ denote an $n$-dimensional vector space with ordered basis $\{e_1,\dots, e_n\}$ and let $\{v^1,\dots,v^n\}$ denote the coordinate functions. Fix $L\in C^\infty(V)$.
\begin{definition}
The \bf{Legendre transform associated to $L$} \rm is the map \[\Phi_L:V\to V^\ast  \ \ \ \ \ \ \ p\mapsto \frac{\pd L}{\pd v}(p)\]where $\frac{\pd L}{\pd v}(p)$ is the co-vector $(\frac{\pd L}{\pd v^1}(p),\dots,\frac{\pd L}{\pd v^n}(p))\in T_p^\ast V\cong V^\ast$
\end{definition}

In other words, $\Phi_L(p)$ is the Jacobian of $L$ evaluated at $p$. 

\begin{definition}
The \bf{dual function associated to $L$ }\rm is the function $L^\ast: V^\ast\to\R$ defined by \[L^\ast: V^\ast\to\R \ \ \ \ \ \ \ \alpha\mapsto \sup\{\alpha\cdot p-L(p) \ , \ \ p\in V\}\]
\end{definition}
Notice that if $V^\ast$ has coordinates $\{\xi_1,\dots,\xi_n\}$ with respect to the dual basis $\{e^1,\dots,e^n\}$ then we can take the Legendre transform of $L^\ast$ which is just $\Phi_{L^\ast}(p)=\left.\text{Jac}(L^\ast)\right|_p=\left[\frac{\pd L^\ast}{\pd \xi_i}(p)\right]$.\\

Suppose now that $V=\R^n$. For the rest of this section, let $(x^1,\dots,x^n)$ denote the standard coordinates on $\R^n$ and let $(x^1,\dots,x^n,v^1,\dots,v^n)$ and $(x^1,\dots,x^n,\xi_1,\dots,\xi_n)$ be the induced coordinates on $T\R^n$ and $T^\ast\R^n$ respectively.. Fix a smooth function $L\in C^\infty(T\R^n)=C^\infty(\R^{2n})$. For what is to follow, let $x=(x^1,\dots,x^n)$, $v=(v^1,\dots,v^n)$,  $a=(a_1,\dots,a_n)$ and $\xi=(\xi_1,\dots,\xi_n)\in \R^n$ be arbitrary. For each $x\in\R^n$ the map $L$ gives an induced map \[L_x:\R^n\to\R \ , \ v\mapsto L(x,v)\]

Recall that the \bf{Hessian} \rm of $L_x$ is the map Hess$(L_x):\R^n\to L(\R^n,\R^n)$ defined by \[\text{Hess}(L_x):=\left[\frac{\pd^2 L_x}{\pd v^i\pd v^j}\right]\] 

\begin{remark}
For fixed $x\in\R^n$ the Legendre transform is a function from $\R^n$ to $\R^n$\[\Phi_{L_x}:\R^n\to\R^n \ \ , \ \ p\mapsto \frac{\pd L_x}{\pd v}(p)\]and so by definition of the Hessian it follows that \[\text{Hess}(L_x)=\left[\frac{\pd^2 L_x}{\pd v^iv^j}\right]=\text{Jac}(\Phi_{L_x})\]
\end{remark}
\begin{definition}
A function $L\in C^\infty (T\R^n)$ is called \bf{strongly convex} \rm if for each $p\in\R^n$ the symmetric matrix Hess$(L_x(p))$ satisfies $u^T\text{Hess}(L_x(p))u>0$ for all non-zero $u\in \R^n$ (i.e. Hess$(L_x(p))$ is a positive definite matrix).
\end{definition}





\begin{proposition}
If $L\in C^\infty (T\R^n)$ is strongly convex then $L$ is strictly convex. 
\end{proposition}

\begin{proof}
Let $x\in\R^n$ be arbitrary. We need to show that $L_x:\R^n\to \R$ is strictly convex. For arbitrary, $p,q\in\R^n$ with $q\not=0$ let $(L_x)_{p,q}$ denote the function \[(L_x)_{p,q}:\R\to\R \ , \ t\mapsto L_x(p+tq).\]Notice that $L_x$ is strictly convex if and only if $(L_x)_{p,q}$ is strictly convex for all $p,q\in\R^n$. But a standard calculation shows that $\left(q^T\cdot\text{Hess}(L_x)_{p,q}\cdot q\right)=(L_x)^{\prime\prime}_{p,q}$. By assumption, it follows that $(L_x)^{\prime\prime}_{p,q}(t)>0$ for all $t\in \R$ and $p,q\in\R^n$. Thus, from basic calculus, it follows that $(L_x)_{p,q}$ is strictly convex for all $p,q\in\R^n$.
\end{proof}

\begin{proposition}
 Fix $x\in\R^n$ and suppose that $L_x:\R^n\to\R$ is strongly convex. Then the following are equivalent
\begin{enumerate}
\item $L_x$ has a critical point $($i.e. there exists $v_0\in\R^n$ such that $\frac{\pd L_x}{\pd v^i}(v_0)=0$ for all $i=1,\dots,n)$.
\item $L_x$ has a local minimum 
\item $L_x$ has a unique global minimum
\end{enumerate}

\begin{proof}

$(1)\Longrightarrow (2)$ Suppose that $v_0$ is a critical point $L_x$. By hypothesis Hess$\left.(L_x)\right|_{v_0}$ is positive definite and so has only positive eigenvalues. Thus by the second derivative test (see \cite{Linear Algebra}, Theorem 6.37) $L$ has a local minimum at $p$. 

$(2)\Longrightarrow (3)$ Suppose that $v_0$ is a local minimum of $L$. Then by definition there exists a neighbourhood $U\subset\R^n$ such that $v_0\in U$ and $L(v_0)\leq L(u)$ for all $u\in U$. Suppose that $v_0$ is not a global minimum. Then there exists $w\in\R^n$ such that $L(w)<L(v_0)$. But then for arbitrary $\theta\in(0,1)$, by Proposition 5.5, we have that $L((1-\theta)v_0+\theta w)<L(v_0)-\theta L(v_0)+\theta L(w)<L(v_0)$. But this is a contradiction since we can choose $\theta$ sufficiently small so that $(1-\theta)v_0+\theta w\in U$.

$(3)\Longrightarrow (1)$ This is known from basic calculus. 
\end{proof}
\end{proposition}
\begin{proposition}
Fix an arbitrary $x\in\R^n$. If $L_x\in C^\infty (T_x\R^n)=C^\infty(\R^{n})$ is strongly convex then $\Phi_{L_x}:T_x\R^n\to\Phi_L(T_x\R^n)$ is a diffeomorphism.
\end{proposition}

\begin{proof}
 By definition, for arbitrary $v_0\in T_x\R^n\cong\R^{n}$ we have $\Phi_{L_x}(v_0)=\frac{\pd L_x}{\pd v}(v_0)\in T_x^\ast \R^n=\R^{n}$. By assumption, the Jacobian of $\Phi_{L_x}$ is positive definite. That is, Hess$\left.(L_x)\right|_{v_0}$ is positive definite so in particular  \[\det\left[\frac{\pd}{\pd v^j}\Phi_{L_x}(v_0)\right]=\det\left[\frac{\pd}{\pd v^j}\frac{\pd}{\pd v^i}L_x(v_0)\right]>0\]Thus, by the inverse function theorem, $\Phi_{L_x}$ is a local diffeomorphism. Since a bijective local diffeomorphism is a diffeomorphism, it suffices to show that $\Phi_{L_x}:T_x\R^n=\R^n\to\Phi_{L_x}(\R^n)$ is injective. So suppose that $p,q\in\R^n=T_x\R^n$ are such that $p\not=q$. Let $w=q-p$ so that $w\not=0$. Since $L_x$ is smooth, we have that $L_x$ is smooth on the line segment $\{p+tw \ ; \ 0\leq t\leq 1\}$. By the chain rule \[\dt \Phi_{L_x}(p+tw)=\left(\text{Jac}(\Phi_{L_x})(p+tw)\right)w=\text{Hess}(L_x(p+tw))w\]We also have that $\Phi_{L_x}(q)-\Phi_{L_x}(p)=\int_0^1\dt \Phi_{L_x}(p+tw)dt$. Putting this together yields
\begin{align*}
w^T\left(\Phi_{L_x}(q)-\Phi_{L_x}(p)\right)&=w^T\left(\int_0^1\dt \Phi_{L_x}(p+tw)dt\right)\\
&=w^T\left(\int_0^1\text{Hess}(L_x(p+tw))wdt\right)\\
&=\int_0^1w^T\text{Hess}(L_x(p+tw))w \ dt
\end{align*}However, Hess$(L_x)$ is positive definite on $\R^n$ and so this last expression is positive. Thus it can't be that $\Phi_{L_x}(p)=\Phi_{L_x}(q)$. Hence $\Phi_{L_x}=\frac{\pd}{\pd v}L_x$ is injective.
\end{proof}

\begin{proposition}
If $L_x$ is strongly convex then for all $\xi\in \Phi_{L_x}(\R^n)$ we have \[L_x^\ast(\xi)=\xi\cdot \Phi_{L_x}^{-1}(\xi)-L_x(\Phi_{L_x}^{-1}(\xi))\]
\end{proposition}

\begin{proof}
Fix $\xi\in\Phi_{L_x}(\R^n)\subset T_x^\ast\R^n=\R^n$. Consider the function $g:\R^n\to\R$ defined by $g(v)=\xi\cdot v-L_x(v)$.  It's clear that $g$ is smooth.  We have that \[\frac{\pd}{\pd v}g(v)=\xi-\frac{\pd}{\pd v}L_x(v)=\xi-\Phi_{L_x}(v)\] and so \[\left.\text{Hess}(g)\right|_v=-\text{Hess}\left.(L_x)\right|_v\] Hence, by hypothesis, it follows Hess$(g(v))$ is negative definite. Also, by Proposition 5.7 there exists a unique $q\in \R^n$ such that $\Phi_{L_x}(q)=\frac{\pd L_x}{\pd v}(q)=\xi$. That is, $\frac{\pd }{\pd v}g(q)=0$. But since Hess$(g)$ is negative definite,  by Proposition 5.6, it follows $q$ is a global maximum for $g$. Hence $g(a)\leq g(q)$ for all $a\in \R^n$. But then 
\begin{align*}L_x^\ast(\xi)&=\sup\{\xi a-f(a)\ , \ a\in \R^n\}\\&=\sup\{g(a) \ , \ a\in \R^n\}\\
&=g(q)\\
&=\xi(\Phi_{L_x}^{-1}(\xi))-L_x(\Phi_{L_x}^{-1}(\xi))\end{align*}
\end{proof}

\del
{
One of the most important properties of the Legendre transform is that it takes any set of coordinates on the tangent bundle $T\R^n$ to a new set of coordinates on the cotangent bundle $T^\ast \R^n$. 
\begin{lemma}
If $(x^1,\dots, x^n)$ are standard coordinates on $\R^n$ and $(x^1,\dots, x^n,v^1,\dots,v^n)$ are the induced coordinates on $T\R^n$, then setting $\xi_i=(\Phi_{L_x}(v))_i=\frac{\pd L_x}{\pd v^i}$ we have that $(x^1,\dots,x^n,\xi_1,\dots,\xi_n)$ is a local coordinate system on $T^\ast \R^n$. That is, the Legendre transform gives an induced coordinate system on $T^\ast \R^n$.

\begin{proof}
This result follows from the inverse function theorem. By definition, \[\left[\frac{\pd}{\pd v^i}\left(\frac{\pd L_x}{\pd v^j}\right)\right]=\text{Hess}(L_x(v))\] is positive definite and so has positive determinant. Hence $\frac{\pd L_x}{\pd v^i}=(\Phi_{L_x}(v))_i$ forms a local coordinate chart.
\end{proof}
\end{lemma}

\begin{remark}
Notice that this argument also works in the reverse direction. That is, if $H\in C^\infty(T^\ast\R^n)$ is strongly convex and  $x^1,\dots,x^n$ are local coordinates on $\R^n$ and $x^1,\dots,x^n,\xi_1,\dots,\xi_n$ are the induced coordinates on $T^\ast \R^n$, then we get an induced coordinate system $x^1,\dots,x^n,v^1,\dots,v^n$ on $T\R^n$ by setting $v^i=\frac{\pd H}{\pd\xi_i}$.
\end{remark}

}

\begin{theorem}
If $L_x$ is strongly convex then $\Phi_{L_{x}^\ast}=\Phi_{L_x}^{-1}$. 
\end{theorem}

\begin{proof}
By definition, the Legendre transform of $L^\ast$ is\[\Phi_{L^\ast_x}: T_x^\ast\R^n\to T_x\R^n \ , \ \xi\mapsto \frac{\pd L_x^\ast}{\pd \xi}\]

Proposition 5.8 showed that for all $\xi\in \Phi_{L_x}(U)$ \[ L_x^\ast(\xi)=\xi\cdot\Phi_{L_x}^{-1}(\xi)-L_x(\Phi^{-1}_{L_x})(\xi)\]

Thus 
\begin{align*}
\Phi_{L^\ast_x}(\xi)&=\frac{\pd L_x^\ast}{\pd \xi}(\xi)\\
&=\Phi_{L_x}^{-1}(\xi)+\xi\left(\frac{\pd}{\pd \xi}\Phi_{L_{x}}^{-1}(\xi)\right)-\left(\frac{\pd L_x}{\pd v}(\Phi^{-1}_{L_x}(\xi))\cdot \left(\frac{\pd \Phi_{L_x}{-1}}{\pd \xi}(\xi)\right)\right)&\text{by the chain rule}\\
&=\Phi_{L_x}^{-1}(\xi)+\xi\left(\frac{\pd}{\pd \xi}\Phi_{L_{x}}^{-1}(\xi)\right)-\left(\Phi_{L_x}(\Phi_{L_x}^{-1}(\xi))\cdot\left(\frac{\pd \Phi_{L_x}^{-1}}{\pd \xi}(\xi)\right)\right)&\text{by definition}\\
&=\Phi_{L_x}^{-1}(\xi)+\xi\left(\frac{\pd}{\pd \xi}\Phi_{L_{x}}^{-1}(\xi)\right)-\xi\cdot\left(\frac{\pd \Phi_{L_x}^{-1}}{\pd \xi}(\xi)\right)\\
&=\Phi_{L_x}^{-1}(\xi)
\end{align*}
\end{proof}

\begin{remark}Let $L\in C^\infty(\R^n)$ be a strongly convex function. We say that $L$ has \bf{quadratic growth at infinity} \rm if there exists a positive definite quadratic form $Q$ on $\R^n$ and a constant $K$ such that $L(p)\geq Q(p)-K$ for all $p\in \R^n$. If $L$ has quadratic growth at infinity, then $\Phi_L(\R^n)=(\R^n)^\ast$. That is, if $L$ is strongly convex and has quadratic growth at infinity, then $\Phi_L:T\R^n\to T^\ast\R^n$ is a diffeomorphism. This is exercise 54 in \cite{Da Silva}. We do not give a proof of this result as we do not need it for what is to follow.
\end{remark}
Using the above theorems we can show that the dual function of a strongly convex Lagrangian is strongly convex. 

\begin{proposition}
If $L\in C^\infty(T\R^n)$ is strongly convex, then $H=L^\ast\in C^\infty(T^\ast\R^n)$ is strongly convex.
\end{proposition}

\begin{proof}Suppose that $L\in C^\infty(\R^n)$ is strongly convex. We need to show that $H=L^\ast\in C^\infty((\R^n)^\ast)$ is strongly convex. That is, for fixed $x\in\R^n$, we need to show that for all $\xi\in\R^n$ we have Hess$(H_x(\xi))$ is positive definite. We showed in the proof of Theorem 5.9 that $\frac{\pd L_x^\ast}{\pd\xi}(\xi)=\Phi_{L_x}^{-1}(\xi)$. We also have that \[\Phi_{L_x}(\Phi_{L_x}^{-1}(\xi))=\Phi_{L_x}(\Phi_{H_x}(\xi))=\xi.\] Differentiating this equality with respect to $\xi$, the chain rule gives that \[\left(\frac{\pd \Phi_{L_x}}{\pd v}(\Phi_{H_x}(\xi))\right)\left(\frac{\pd}{\pd\xi}\Phi_{H_x}(\xi)\right)=1.\]Since $\Phi_{L_x}=\frac{\pd L_x}{\pd v}$ and $\Phi_{H_x}=\frac{\pd L_x^\ast}{\pd\xi}=\frac{\pd H_x}{\pd\xi}$ this equation is the same as\[\left(\frac{\pd^2L_x}{(\pd v)^2}(\Phi_{H_x}(\xi))\right)\left(\frac{\pd^2H_x}{(\pd\xi)^2}(\xi)\right)=1.\]That is we have shown that \[\text{Hess}(H_x(\xi))=\text{Hess}(L_x\left(\Phi_{H_x}(\xi)\right))^{-1}.\]Since the inverse of a positive definite matrix is positive definite, the result follows.
\end{proof}

\begin{corollary}\bf{(Involutivity of the Legendre Transform and Dual Function )}\rm \\
Let $L\in C^\infty(\R^n)$ be strongly convex. If $H=L^\ast$, then $H^\ast=L$. In particular, this means that $\Phi_{H^\ast}=\Phi_{L^{\ast\ast}}=\Phi_L$.
\end{corollary}

\begin{proof}
Let $L\in C^\infty(\R^n)$ be strongly convex.  From Proposition 5.11 we have that $H=L^\ast$ is strongly convex. Hence
\begin{align*}
H^\ast(x,v)&=v\cdot\Phi_{H_x}^{-1}(v)-H_x(\Phi_{H_x}^{-1}(v))&\text{by Proposition 5.8 and 5.11}\\
&=v\cdot \Phi_{H_x}^{-1}(v)-\left(\left(\Phi_{H_x}^{-1}(v)\right)\cdot\left(\Phi_{L_x}^{-1}(\Phi_{H_x}^{-1}(v))\right)-L_x\left(\Phi_{L_x}^{-1}(\Phi_{H_x}^{-1}(v))\right)\right)&\text{by Proposition 5.8}\\
&=v\cdot \Phi_{H_x}^{-1}(v)-\left(\left(\Phi_{H_x}^{-1}(v)\right)\cdot v-L_x(v)\right)&\text{since $\Phi_{L_x}^{-1}=\Phi_{H_x}$}\\
&=L(x,v)
\end{align*}
\end{proof}

\subsection{The Legendre Transform on Manifolds}

The Legendre transform can be extended naturally to act on manifolds since, at every point, the tangent and cotangent spaces are vector spaces.
\vspace{0.2cm} 

Let $M$ be a $n$-dimensional manifold and $(U,x^1,\dots,x^n)$ an arbitrary coordinate chart. By definition, $U$ is diffeomorphic to an open subset of $\R^n$. We have the induced coordinate charts $(TU,x^1,\dots,x^n,v^1,\dots,v^n)$ and $(T^\ast U,x^1,\dots,x^n,\xi_1,\dots,\xi_n)$ on $TM$ and $T^\ast M$ respectively, and we know that $TU\cong U\times\R^n\cong T^\ast U$. Suppose that $L\in C^\infty(TU)$ is strongly convex. We define the Legendre transform associated to $L$ to be the map \[\Phi_L:TU\to \Phi_L(TU) \ \ \ \ \ \ \ (x,v)\mapsto \frac{\pd L}{\pd v}(x,v).\] For a fixed $x\in M$, we have that $L_x\in C^\infty(T_xU)$ is strongly convex with respect to $v^1,\dots,v^n$. The Legendre transform induces the map \[\Phi_{L_x}:T_xU\to T_x^\ast U \ \ \ \ \ \ \ W_x\mapsto \Phi_{L_x}(W_x)=\frac{\pd L_x}{\pd v}(W_x)\]That is, $\frac{\pd L_x}{\pd v}(W_x)$ is the $n$-tuple $\left(\frac{\pd L_x}{\pd v^1}(W_x),\cdots, \frac{\pd L_x}{\pd v^n}(W_x)\right)$. For each $x\in M$ the dual function associated to $L$ is again defined to be the map \[L_x^\ast:T_x^\ast U\to\R \ \ \ \ \ \ \ \xi_x\mapsto \sup\{\xi_x\cdot W_x-L_x(W_x) \ ; \ W_x\in T_x U\}\]All of the results from the previous section still hold. That is, for each $x\in M$, we have 

\begin{itemize}
\item $L_x$ has a critical point $\iff L_x$ has a local minimum $\iff L_x$ has a unique global minimum.
\item $\Phi_{L_x}:T_xU\to \Phi_{L_x}\left(T_xU\right)$ is a diffeomorphism.
\item For all $\xi_x\in\Phi_{L_x}\left(T_xU\right)$ we have $L_x^\ast(\xi_x)=\xi_x\cdot\Phi_{L_x}^{-1}(\xi_x)-L_x(\Phi_{L_x}^{-1}(\xi_x)).$
\item $\Phi_{L^\ast_x}=\Phi_{L_x}^{-1}.$
\item If $L\in C^\infty(TU)$ is strongly convex then $H=L^\ast\in C^\infty(\Phi_L(TU))$ is strongly convex.
\item If $L\in C^\infty(TU)$ is strongly convex and $H=L^\ast\in C^\infty(\Phi_L(TU))$ then $H^\ast=L$. That is, the dual function and Legendre transform are involutive.
\del
{
\item The Legendre transform provides induced local coordinates on $T^\ast M$ given by $\xi_i=\frac{\pd L}{\pd v^i}$.
}
\end{itemize}






\subsection{The Legendre Transform Relates Lagrangian and Hamiltonian Mechanics}

Let $(M,L)$ be an arbitrary Lagrangian system, where $L\in C^\infty(TM)$ is strongly convex. In the previous section we defined $L^\ast\in C^\infty(T^\ast M)$, the dual function of $L$. Moreover, we know that $T^\ast M$ is a symplectic manifold when equipped with the canonical $2$-form $\omega$. Hence, we see that the Hamiltonian system $(T^\ast M,\omega, H=L^\ast)$ arises naturally from the Lagrangian system $(M,L)$. Similarly, given a Hamiltonian system of the form $(T^\ast M,\omega,H)$ for some strongly convex $H\in C^\infty(T^\ast M)$, we can define the dual function $H^\ast\in C^\infty(TM)$. This gives the Lagrangian system $(M,L=H^\ast)$. Since the Legendre transform and the dual function are involutive, we see that these induced systems are well defined and `inverse' to each other. This motivates the following definition.

\begin{definition}
Given a Lagrangian system $(M,L)$ the \bf{induced Hamiltonian system} \rm is the triple $(T^\ast M,\omega, H:= L^\ast)$ where $\omega$ is the canonical $2$-form and $T^\ast M$. Similarly, given a Hamiltonian system of the form $(T^\ast M,\omega, H)$, the \bf{induced Lagrangian system} \rm is the pair $(M,L:=H^\ast)$.
\end{definition} 
\vspace{0.1cm}

\begin{remark}
Let $(U,x^1,\dots,x^n)$ be a coordinate chart in a manifold $M$. We have the induced coordinate charts $(TU,x^1,\dots,x^n,v^1,\dots,v^n)$ and $(T^\ast U,x^1,\dots,x^n,\xi_1,\dots,\xi_n)$ on $TM$ and $T^\ast M$ respectively. Proposition 5.7 showed that $\Phi_{L}:TU\to \Phi_L(TU)$ is a diffeomorphism, while Theorem 5.9 showed that $\Phi_L^{-1}=\Phi_{L^\ast}$. Hence, in the induced Hamiltonian system $(T^\ast M,\omega, H=L^\ast)$ we have the coordinate chart $(\Phi_L(TU),x^1,\dots,x^n,\xi_1,\dots,\xi_n)$ where each $\xi_i$ satisfies $\xi_i=\frac{\pd L}{\pd v^i}$. We also have that $\Phi_L^{-1}=\Phi_H$. Similarly, if we are given a Hamiltonian system of the form $(T^\ast M,\omega, H)$, then the Legendre transform $\Phi_H:T^\ast U\to \Phi_H(T^\ast U)$ gives an induced coordinate chart $(\Phi_H(T^\ast U),x^1,\dots,x^n,v^1,\dots,v^n)$ where each $v^i=\frac{\pd H}{\pd\xi_i}$. Also, we have that $\Phi_H^{-1}=\Phi_L$.
\end{remark}
Given a natural Lagrangian system, the Hamiltonian function in the induced Hamiltonian system is always the total energy.
\begin{proposition}
Let $(M,g)$ be a Riemannian manifold with Lagrangian $L=K-U\in C^\infty(T^\ast M)$. Then $H:=L^\ast=E=K+U$.
\end{proposition} 

\begin{proof}
By definition, $L=\frac{1}{2}mg_{ij}v^iv^j-U$. For $(x,v)\in TM$, let $(x,\xi)=\Phi_L(x,v)\in\Gamma(T^\ast M)$. That is, 
\begin{align*}
\xi_k&=\frac{\pd L_x}{\pd v^k}\\
&=\frac{\pd}{\pd v^k}\left(\frac{1}{2}mg_{ij}v^iv^j\right)&\text{since $U$ is independent of $v^1,\dots,v^n$}\\
&=\frac{1}{2}mv^jg_{ij}\delta_k^i+\frac{1}{2}mv^ig_{ij}\delta^j_k&\text{by the product rule}\\
&=mg_{ik}v^i
\end{align*}

Letting $g^{ij}=[g^{-1}]_{ij}$ it follows that \[v^i=\frac{1}{m}g^{ik}\xi_k.\]
Since $\Phi_L^{-1}(x,\xi)=\Phi_L^{-1}(\Phi_L(x,v))=(x,v)$, by Proposition 5.8 it follows that
\begin{align*}
H(x,\xi)&=L^\ast(x,\xi)\\
&=\xi_iv^i-L(x,v)\\
&=mg_{ij}v^jv^i-\frac{1}{2}mg_{ij}v^iv^j+U(x)\\
&=\frac{1}{2}mg_{ij}v^iv^j+U(x)\\
&=K+U
\end{align*}

\end{proof}
\begin{remark}
Notice that if the metric is the standard one and the manifold is $\R^n$, then in $TM=\R^{2n}$ motions are described by their position and velocity coordinates. However, in the induced Hamiltonian system we have $\xi_i=mv^i$ and so motions in here are described at each time by specifying position and momentum coordinates.
\end{remark}
A simple calculation, which we show now, demonstrates that if the natural Lagrangian is time independent then the total energy is conserved. In section 6 we will see that time independence of the natural Lagrangian can be thought of as a `symmetry' and so conservation of energy can also be seen as a consequence of Noether's theorem. 
\begin{proposition}
In a natural Lagrangian system, if the Lagrangian is independent of time then energy is conserved. That is, if $L(x^1,\dots,x^n,v^1,\dots,v^n,t)$ is such that $\frac{\pd L}{\pd t}=0$ then $\dt H=0$.
\end{proposition}
\begin{proof}
Let $\gamma(t)$ be a motion in $(M,L)$. That is, suppose $\widetilde\gamma(t)=(\gamma(t),\gamma^\prime(t))$ satisfies the Euler Lagrange equations. We have that
\begin{align*}
\dt L(\gamma(t),\gamma^\prime(t),t)&=\left.\frac{\pd L}{\pd x^i}\right|_{(\widetilde\gamma(t),t)}\gamma^\prime(t)+\left.\frac{\pd L}{\pd v^i}\right|_{(\widetilde\gamma(t),t)}\gamma^{\prime\prime}(t)+\left.\frac{\pd L}{\pd t}\right|_{(\widetilde\gamma(t),t)}\\
&=\dt\left.\frac{\pd L}{\pd v^i}\right|_{(\widetilde\gamma(t),t)}\gamma^\prime(t)+\left.\frac{\pd L}{\pd v^i}\right|_{(\widetilde\gamma(t),t)}\gamma^{\prime\prime}(t)+\left.\frac{\pd L}{\pd t}\right|_{(\widetilde\gamma(t),t)}&\text{by the Euler-Lagrange equations}\\
&=\dt\left(\left.\frac{\pd L}{\pd v^i}\right|_{(\widetilde\gamma(t),t)}\gamma^\prime(t)\right)+\left.\frac{\pd L}{\pd t}\right|_{(\widetilde\gamma(t),t)}&\text{by the product rule}\\
&=\dt\left(\xi_i(\widetilde\gamma(t))v^i(\widetilde\gamma(t))\right)+\left.\frac{\pd L}{\pd t}\right|_{(\widetilde\gamma(t),t)}&\text{by definition}\\
\end{align*}

That is \[\dt\left(L(x,v,t)-\xi_iv^i\right)=\left.\frac{\pd L}{\pd t}\right|_{(\widetilde\gamma(t),t)}.\]However, by proposition 5.8 we have that $L(x,v,t)-\xi_iv^i=-L^\ast:=-H$. Hence if $\frac{\pd L}{\pd t}=0$ then $\dt H=0$. By Proposition 5.15 this means that the total energy is conserved.
\end{proof}

The first example we give that demonstrates how the Legendre transform relates Hamiltonian mechanics and Lagrangian mechanics is by showing how it translates motions in one formulation to motions in the other.  To see this we first need the following lemma.
\begin{lemma} Assume that $L\in C^\infty(TM)$ is strongly convex. As above, if $(TU,x^1,\dots,x^n,v^1,\dots,v^n)$ is a coordinate chart on $TM$ we get the induced chart $(\Phi_{L}(TU),x^1,\dots,x^n,\xi_1,\dots,\xi_n)$ on $T^\ast U$, where by definition $(x,\xi)=\Phi_{L}(x,v)$. Let $H=L^\ast\in C^\infty(T^\ast U)$. The claim is that $\frac{\pd L}{\pd x}(x,v)=-\frac{\pd H}{\pd x}(x,\xi)$
\end{lemma}
\begin{proof} By Proposition 5.8, $H_x(\xi)=L^\ast_x(\xi)=\xi (\Phi_{L_x}^{-1}(\xi))-L(\Phi_{L_x}^{-1}(\xi))$. However, $\Phi_{L_x}^{-1}(\xi)=\Phi_{L_x}^{-1}(\xi)=\Phi_{L_x}^{-1}(\Phi_{L_x}(v))=v$ and so \begin{align}H(x,\xi)=\xi\cdot v-L(x,v)=\xi_iv^i-L(x,v)\end{align} We know that $\frac{\pd v^j}{\pd x^i}=0$ for all $1\leq i,j\leq n$; however, by definition $\xi$ is dependent on $x$ and $v$. Hence, taking the total derivative of $H(x,\xi)$ with respect to $x^i$, the left hand side of $(5.1)$ is \[\frac{\pd H}{\pd x^i}(x,\xi)+\frac{\pd H}{\pd \xi_j}(x,\xi)\frac{\pd \xi_j}{\pd x^i}(x,v)\] while the total derivative of the right hand side is
\[\frac{\pd\xi_j}{\pd x^i}(x,v)v^j-\left(\frac{\pd L}{\pd x^i}(x,v)\right)\]

However, by hypothesis we have that $\frac{\pd H}{\pd \xi_i}(x,\xi)=\left(\Phi_{L}^{-1}(\Phi_L(x,v))\right)^i=v^i$ and so combining these equalities finishes the proof.
\end{proof}

This result gives us the following two theorems.

\begin{theorem} If a curve $\gamma :\R\to U$ satisfies the Euler-Lagrange equations on some chart $U\subset M$, then $\Phi_L\circ\widetilde\gamma:[a,b]\to T^\ast M$ is an integral curve of the Hamiltonian vector field $V_H$.

\end{theorem}

\begin{proof}
Let $(U,x_1,\dots,x_n)\subset M$ be an arbitrary chart. We have the induced charts $(TU,x^1,\dots,x^n,v^1,\dots,v^n)$ and $(\Phi_L(TU),x^1,\dots,x^n,\xi_1,\dots,\xi_n)$ on $TM$ and $T^\ast M$ respectively. By hypothesis, $\widetilde\gamma(t)=(\gamma(t),\gamma^\prime(t))$ satisfies the Euler-Lagrange equations. That is \[\frac{\pd L}{\pd x^i}(\widetilde\gamma(t))=\dt\left(\frac{\pd L}{\pd v^i}(\widetilde\gamma(t))\right)\]Let $\Psi(t)=\Phi_L(\widetilde\gamma(t))=\Phi_L(\gamma(t),\gamma^\prime(t))=(\gamma(t),\Phi_{L_{\gamma(t)}}(\gamma^\prime(t)))$. It needs to be shown that $\Psi(t)$ satisfies Hamilton's equations. That is, it needs to be shown that
\[
\begin{array}{r c l}
\dt \gamma^i(t)&=&\frac{\pd H}{\pd \xi_i}(\Psi(t))\\
\\
\dt (\Phi_{L_{\gamma(t)}}(\gamma^\prime(t)))_i(t)&=&-\frac{\pd H}{\pd x^i}(\Psi(t))
\end{array}
\]But since $\Psi(t)=\Phi_L(\widetilde\gamma(t))$ it follows that $\Phi_H(\Psi(t))=\widetilde\gamma(t)$. Hence \[\gamma^\prime(t)=\frac{\pd H}{\pd \xi}(\Psi(t))\]This is precisely the first line of Hamilton's equations. The second line also holds since
\begin{align*}
\dt(\Phi_{L_{\gamma(t)}}(\gamma^\prime(t)))_i&=\dt\frac{\pd L}{\pd v^i}(\widetilde\gamma(t)) &\text{by definition}\\
&=\frac{\pd L}{\pd x^i}(\widetilde\gamma(t)) &\text{by the Euler-Lagrange equations}\\
&=-\frac{\pd H}{\pd x^i}(\gamma(t),\Phi_{L_{\gamma(t)}}(\gamma^\prime(t)))&\text{by the Lemma}\\
&=-\frac{\pd H}{\pd x^i}(\Psi(t))&\text{by definition} 
\end{align*}

\end{proof}

\noindent A stronger version of the converse is also true:

\begin{theorem} Given a Lagrangian system $(M,L)$, where $L\in C^\infty(TM)$ is strongly convex, let $(T^\ast M,\omega, H=L^\ast)$ be the induced Hamiltonian system. If $\Psi:\R\to T^\ast M$ is an integral curve for $V_H$ then $\Psi=\Phi_L\circ\gamma$ for some motion $\gamma$ in $(M,L)$.
\end{theorem}

\begin{proof}
Let $\Psi(t)=(\alpha(t),\beta(t))\in T^\ast M$ be an integral curve for $V_H$. Then $\Psi(t)$ satisfies the Hamilton equations
\[
\begin{array}{r c l}
\dt \alpha^i(t)&=&\frac{\pd H}{\pd \xi_i}(\Psi(t))\\
\\
\dt \beta_i(t)&=&-\frac{\pd H}{\pd x^i}(\Psi(t))
\end{array}
\]
Here $\alpha$ is a curve $\alpha:\R\to M$. It follows that \[\widetilde\alpha(t)=(\alpha(t),\alpha^\prime(t))=(\alpha(t),\frac{\pd H}{\pd \xi}(\Psi(t)))=(\alpha(t),\Phi_{H_{\alpha(t)}}(\beta(t)))\]so that \[\Phi_L(\widetilde\alpha(t))=(\alpha(t),\Phi_{L_{\alpha(t)}}\Phi_{H_{\alpha(t)}}(\beta(t))=(\alpha(t),\beta(t)))=\Psi(t)\]It suffices to show that $\widetilde\alpha(t)$ satisfies the Euler-Lagrange equations. Indeed
\begin{align*}
\frac{\pd L}{\pd x}(\widetilde\alpha(t))&=\frac{\pd L_{\alpha(t)}}{\pd x}(\dt\alpha(t))&\text{by definition}\\
&=-\frac{\pd H_{\alpha(t)}}{\pd x}(\beta(t)) &\text{by the Lemma 5.17}\\
&=-\frac{\pd H}{\pd x}(\Psi(t))&\text{by definition}\\
&=\dt(\beta(t)) &\text{by Hamilton's equations}\\
&=\dt\Phi_{L_{\alpha(t)}}\left(\Phi_{H_{\alpha(t)}}(\beta(t))\right)&\text{since $\Phi_{L^\ast}=\Phi_L^{-1}$}\\
&=\dt\Phi_{L_{\alpha(t)}}(\frac{\pd H}{\pd \xi}(\Psi(t)))&\text{by definition}\\
&=\dt \Phi_{L_{\alpha(t)}}(\dt(\alpha(t)))&\text{by Hamliton's equation}\\
&=\dt\frac{\pd L}{\pd v}(\widetilde\alpha(t))&\text{by definition}
\end{align*}
\end{proof}

\begin{example}\bf{(Geodesic Flow in Hamiltonian Mechanics)} \rm
\vspace{0.1cm}

We can now show how the Legendre transform relates the concept of geodesic flow in Lagrangian and Hamiltonian mechanics. Let $(M,g,L)$ be a natural Lagrangian system, where $L\in C^\infty (TM)$ is strongly convex. Recall that in section 2.7 we derived the geodesic flow as the symplectomorphism generated by the Riemann distance function.  In other words, we set our Lagrangian to be \[L:TM\to\R \ \ \ \ \ \ \ (x,V_x)\mapsto \frac{1}{2}g_x(V_x,V_x).\]
As demonstrated in Example 3.19 if the net force on the mechanical system is 0, then the solutions to the Euler-Lagrange equations are geodesics. Consider what happens if we translate this system into the Hamiltonian setting. Using the argument in the proof of Proposition 5.15 we have that \[L^\ast(x,\xi)=\frac{1}{2}g^{ij}\xi_i\xi_j.\]By definition our Hamiltonian vector field is \[V_H:=\frac{\pd H}{\pd \xi_i}\frac{\pd}{\pd x^i}-\frac{\pd H}{\pd x^i}\frac{\pd}{\pd \xi_i}.\]The integral curves, $\gamma(t)=(x(t),\xi(t))$, of $V_H$ must satisfy Hamilton's equations:
\[
\begin{array}{ l c r }
\dt x^k(t)&=&\frac{\pd H}{\pd \xi_k}\\
\!&\!&\! \\
\dt \xi_k(t)&=&-\frac{\pd H}{\pd x^k}
\end{array}
\]
We have that \begin{align}\frac{\pd H}{\pd\xi_k}=\frac{1}{2}g^{ij}\delta^k_i\xi_j+\frac{1}{2}g^{ij}\xi_i\delta^k_j=g^{kj}\xi_j\end{align}and\begin{align}-\frac{\pd H}{\pd x^k}=-\frac{1}{2}\frac{\pd g^{ij}}{\pd x^k}\xi_i\xi_j\end{align} 

To make the notation clearer, we will denote the time derivative using dot notation. If $\gamma(t)$ satisfies Hamilton's equations then by equation $(5.2)$ we have that $\xi_k=g_{ak}\dot x^a$. Plugging this into the second line of Hamilton's equations and using $(5.3)$ we get that \begin{align}\dot{\xi_k}=\frac{\pd g_{ak}}{\pd x^q}\dot x^a\dot x^q+g_{ak}\ddot x^a=-\frac{1}{2}\frac{\pd g^{ij}}{\pd x^k}g_{ia}g_{jp}\dot x^a\dot x^p\end{align}
We can simplify this expression using the following claim.

\begin{claim}
We have that $-\frac{\pd g^{ij}}{\pd x^k}g_{ia}g_{jp}=\frac{\pd g_{ap}}{\pd x^k}$
\end{claim}

\begin{proof}
We know that $g_{ap}g^{pj}=\delta^j_a$. Differentiating this with respect to $x^k$ we get that \[\frac{\pd g_{ap}}{\pd x^k}g^{pj}+g_{ap}\frac{\pd g^{pj}}{\pd x^k}=0.\] Multiplying both sides by $g_{ji}$ and summing over $j$ gives \[\frac{\pd g_{ap}}{\pd x^k}\delta^p_i=-\frac{\pd g^{pj}}{\pd x^k}g_{ap}g_{ij}\]
\end{proof}

\end{example}

Using this claim, equation $(5.4)$ becomes \[\frac{\pd g_{ak}}{\pd x^q}\dot x^a\dot x^q+g_{ak}\ddot x^a=\frac{1}{2}\frac{\pd g_{ap}}{\pd x^k}\dot x^a\dot x^p.\]Rearranging, we get that if $\gamma(t)$ satisfies Hamilton's equations then \[\ddot x^b=g^{kb}\left(\frac{1}{2}\frac{\pd g_{ap}}{\pd x^k}\dot x^a\dot x^p-\frac{\pd g_{ak}}{\pd x^q}\dot x^a\dot x^q\right)=-\frac{1}{2}g^{kb}\left(\frac{\pd g_{ak}}{\pd x^p}\dot x^a\dot x^p+\frac{\pd g_{pk}}{\pd x^a}\dot x^a\dot x^p-\frac{\pd g_{ap}}{\pd x^k}\dot x^a\dot x^p\right)\]This is precisely the geodesic equation. Hence, a curve satisfying the Hamilton equations is a geodesic. Conversely, let $\gamma(t)=(x^1(t),\dots,x^n(t))$ be a geodesic in $M$. Then if we set $\xi_k=g_{ak}\dot x^a$, applying the above argument to $\psi(t)=(x^1(t),\dots,x^n(t),\xi_1(t),\dots,\xi_n(t))$ shows that $\psi$ is an integral curve of $V_H$.

\begin{remark} As in the previous example, consider a Lagrangian of the form $L=K$. Then, as in the proof of Proposition 5.15,  we have for fixed $x\in M$ \[(\Phi_{L}(x,v))_i=g_{ji}v^j\]That is, for $W_x=W^i\left.\frac{\pd}{\pd x^i}\right|_x\in T_xM$ we have that \[\Phi_{L_x}(W_x)=g_{ji}W^idx^j.\]Hence in this case the Legendre transform is just the musical isomorphism \[\Phi_{L_x}:T_xM\to T_x^\ast M \ \ \ \ \ \ \ W\mapsto g(W,\cdot ).\]

In section 2.7 we showed that the symplectomorphism generated by the Riemann distance function was the geodesic flow \[\varphi:TM\to TM \ \ \ \ \ \ \ V\mapsto (\gamma_V(1),\gamma_V^\prime(1)).\]But in order to find this symplectomorphism we identified $T^\ast M$ with $TM$ via the musical isomorphism. In fact, now we can see that all we were doing in that section was solving the Hamilton equations. We were trying to find $V$ and $W$ such that $g(V,\cdot)=d_a L$ and $g(W,\cdot)=-d_bL$. We took geodesics (motions in $(M,L)$) and mapped them under the Legendre transform to motions in $(T^\ast M,\omega,H).$ It follows that an equivalent way to define the geodesic flow is as follows.

\end{remark}

\begin{definition}
Consider the smooth function\[H:T^\ast M\to\R \ \ \ \ \ \ \ (x,\xi)\mapsto \frac{1}{2}g^{ij}(x)\xi_i\xi_j\]and its Hamiltonian vector field\[V_H:=\frac{\pd H}{\pd \xi_i}\frac{\pd}{\pd x^i}-\frac{\pd H}{\pd x^i}\frac{\pd}{\pd\xi_i}.\]The flow generated by $V_H$ is called the \bf{geodesic flow}\rm.
\end{definition}
\vspace{0.1cm}

\begin{example}\bf{(The Simple Pendulum Under the Legendre Transform)} \rm
\vspace{0.1cm}

Recall that in Example 3.25 we found the Euler-Lagrange equations for the simple pendulum. We also saw in Example 4.10 how it was described in the Hamiltonian formulation. Noticing that the pendulum was constrained to $S^1$, we worked in a coordinate chart $(U,\theta)$ of $S^1$ and considered $(TS^1,\theta,\widetilde\theta)$. With the metric on $S^1$ induced from $\R^2$, i.e. $g=g_{22}=l^2(d\theta)^2$, we saw that the kinetic energy was $K=\frac{1}{2}ml^2(\widetilde\theta)^2$ and the potential energy was $U=mgl(1-\cos\theta)$. Recall that in the Hamiltonian setting, we didn't change the potential, but we set the kinetic energy to be $K=\frac{\xi^2}{2ml^2}$. To see why we did this, we apply the Legendre transform to the Lagrangian set up. Let $(\Phi_L(TU),\theta,\xi)$ be a chart in the induced Hamiltonian system so that $\xi=\frac{\pd L}{\pd\widetilde\theta}=\frac{\pd}{\pd\widetilde\theta}(\frac{1}{2}ml^2\widetilde\theta^2)=ml^2\widetilde\theta.$ It follows that $\widetilde\theta=\frac{\xi}{ml^2}$. Thus, \[K=\frac{1}{2}ml^2\frac{\xi^2}{m^2l^4}=\frac{\xi^2}{2ml^2}.\]Since in the induced Hamiltonian system the first coordinate is the same as in the Lagrangian setting, the potential energy remains unchanged.
\end{example}

\newpage

\section{Noether's Theorem}
Noether's theorem provides a relationship between symmetries and constants of motion. Before going into the details we first need the formal definitions. In this section we will always assume that our manifolds are geodesically complete. 

\subsection{Noether's Theorem in Lagrangian Mechanics}
First recall the different notions we have of `lifting' maps. Let $M$ be a manifold and $f:M\to M$ a diffeomorphism. In section $2.5$ we defined the lift of $f$ to the cotangent bundle to be the map \[f_\sharp:T^\ast M\to T^\ast M \ \ \ \ \ \ \ (x,\xi)\mapsto (f(x),(f^\ast)^{-1}(\xi)).\]Note that we can also lift $f$ to a map on the tangent bundle by taking the differential of $f$. To avoid confusion, we will denote the lift of $f$ to $TM$ by \[\widetilde f:TM\to TM \ \ \ \ \ \ \ (p,V)\mapsto (f(p),f_{\ast,p}(V)).\]Lastly, we defined in section 3.2 the lift of a curve $\gamma:\R\to M$ to the tangent bundle by\[\widetilde\gamma:\R\to TM \ \ \ \ \ \ \ t\mapsto (\gamma(t),\gamma^\prime(t)).\]
\begin{definition}
In a Lagrangian system $(M,L)$ a \bf{continuous symmetry} \rm is a one parameter family of diffeomorphisms $\left\{\theta_s:M\to M \ ; s\in\R\right\}$ such that for each $s\in\R$ we have $(\widetilde\theta_s)^\ast L=L$. That is, the family of maps is a continuous symmetry if for all $s\in \R$ we have that $L\circ\widetilde\theta_s=L$. 
\end{definition}

\del
{
\begin{remark}
A slightly weaker notion of continuous symmetry, but one which will be sufficient for the examples below, is by requiring that for each motion $\gamma:\R\to M$ we have that $L\left(\widetilde{(\theta_s\circ\gamma)}(t)\right)=L(\widetilde\gamma(t))$. That is, the one parameter family is a continuous symmetry if  for all $s\in\R$, \ $\theta_s\circ\gamma:\R\to M$ is also a motion in $(M,L)$. Notice that this is a weaker notion of symmetry as given in Definition 6.1 since $\widetilde{(\theta_s\circ\gamma)}(t)=\theta_{s,\ast}(\widetilde\gamma(t))=\widetilde\theta_{s,\gamma(t)}(\widetilde\gamma(t))$. 

\end{remark}
}

A continuous symmetry can be thought of as a symmetry of motion. The standard notion of symmetry is invariance under some sort of mapping, i.e. an object is called symmetric if there is a map that preserves it. But here the objects being acted on are motions and so this definition is referring to the preservation of solutions to the Euler-Lagrange equations. Some obvious examples that we observe in homogeneous space are the invariance of the laws of motion under space and time translations. Noether's theorem says that both these families of continuous symmetries (space and time translations) have corresponding conserved quantities. We will see below that they are conservation of momentum and conservation of energy respectively.


\begin{definition}
In a Lagrangian system $(M,L)$ a \bf{conserved quantity} \rm(or \bf{constant of the motion}\rm) is a smooth function $G\in C^\infty(TM)$ with the property that for any motion $\gamma:\R\to M$ in $(M,L)$, the total time derivative of $G$ vanishes on the image of $\gamma$. That is for all $t\in\R$ \[\dt G(\gamma(t))=0\]All of the conservation laws in physics correspond to a conserved quantity. 
\end{definition}

\begin{theorem} \bf{(N$\ddot {\text{o}}$ether)} \rm Let $(M,L)$ be a Lagrangian system and $\gamma:\R\to M$ a motion. Let $(x^1,\dots,x^n)$ be local coordinates for $M$ and $(x^1,\dots,x^n,v^1,\dots,v^n)$ the induced local coordinates on $TM$. For any continuous symmetry $\left\{\theta_s: M\to M \ , \ s\in\R\right\}$ in $(M,L)$ there exists a conserved quantity. The conserved quantity is given by the formula \[\left(\frac{\pd}{\pd v^i}L\right)\cdot\left(\frac{d}{ds}\circ\theta_{s,\ast}^i\right)\in C^\infty(TM).\]
\end{theorem}

\begin{proof}By hypothesis, for any $s\in\R$, $L(\widetilde{(\theta_s\circ\gamma)}(t))=L(\widetilde\gamma(t))$. That is, \begin{align*}
0&=\frac{\pd}{\pd s}L(\widetilde{(\theta_s\circ\gamma)}(t))\\
&=\frac{\pd}{\pd s}L\left((\theta_s\circ\gamma)(t),(\theta_s\circ\gamma)^\prime(t)\right)\\
&=\pdx L\left(\widetilde{(\theta_s\circ\gamma)}(t)\right)\left(\frac{d}{ds}(\theta_s\circ\gamma)^i(t)\right)+\frac{\pd}{\pd v^i} L\left(\widetilde{(\theta_s\circ\gamma)}(t)\right)\left(\frac{d}{ds}\left(\dt(\theta_s\circ\gamma)^i(t)\right)\right)\\
\end{align*}
By hypothesis we have $\widetilde{\theta_s\circ\gamma}$ satisfies the Euler-Lagrange equations:\[ \pdx L\left(\widetilde{(\theta_s\circ\gamma)}(t)\right)=\dt \left(\frac{\pd}{\pd v^i} L\left(\widetilde{(\theta_s\circ\gamma)}(t)\right)\right)\]Plugging in the left hand side of this equation into the above gives 
\begin{align*}
0&=\dt \left(\frac{\pd}{\pd v^i} L\left(\widetilde{(\theta_s\circ\gamma)}(t)\right)\right)\left(\frac{d}{ds}(\theta_s\circ\gamma)^i(t)\right)+\frac{\pd}{\pd v^i} L\left(\widetilde{(\theta_s\circ\gamma)}(t)\right)\left(\frac{d}{ds}\left(\dt(\theta_s\circ\gamma)^i(t)\right)\right)\\
&=\dt\left[\left(\frac{\pd}{\pd v^i}L\left(\widetilde{(\theta_s\circ\gamma)}(t)\right)\right)\left(\frac{d}{ds}(\theta_{s,\ast}(\widetilde\gamma(t)))^i\right)\right]
\end{align*}That is, $\left(\frac{\pd}{\pd v^i}L\right)\cdot\left(\frac{d}{ds}\circ\theta_{s,\ast}^i\right)$ is a conserved quantity. 
\end{proof}

\begin{remark}
By noticing that $\frac{\pd L}{\pd v^i}$ is nothing but the Legendre transform of $L$, we can give a coordinate free description of the resulting conserved quantity. That is, setting \[G:=\left(\frac{\pd L}{\pd v^i}\right)\left(\frac{d}{ds}\circ(\theta_{s,\ast})^i\right)\] we have that for arbitrary $(p,v)\in TM$ \[G(p,v)=\left(\Phi_{L_p}(v)\right)\cdot\left(\theta_\ast^{(p)}\left(\left.\frac{d}{ds}\right|_{s=0}\right)\right).\]

\end{remark}
\begin{example} \bf{(SO(3) gives Continuous Symmetries under a Central Force)} \rm
\vspace{0.1cm}

Let $F$ be a central force acting on a particle in $\R^3$ with metric $g$ and natural Lagrangian $L=K-U$. Let $\gamma$ be the motion of the particle. A \emph{one-parameter subgroup} of $SO(3)$ is a collection of maps $\{\theta_s\in SO(3)\}$, where $s\in\R$, with the property that $\theta_{s+t}=\theta_s\circ\theta_t$. By definition, for every $x\in\R^3$ we have that $\|\theta_s(x)\|=\|x\|$ and hence for each $1\leq i\leq 3$ we have that $U(\theta_s(\gamma(t)))=U(\gamma(t))$. Since each $\theta_s$ is a linear operator on $\R^3$ each map has a matrix representation, say $[\theta_s]$. Note that 
\begin{align*}
v^i(\widetilde{(\theta_s\circ\gamma)}(t))v^j(\widetilde{(\theta_s\circ\gamma)}(t))&=\dt((\theta_s\circ\gamma(t))^i)\cdot\dt((\theta_s\circ\gamma(t))^i)\\
&=\left([(\gamma^i)^\prime(t)]^T[\theta_s]^T\right)\cdot\left([\theta_s][(\gamma^i)^\prime(t)]\right)&\text{since $\theta_s$ is linear}\\
&=v^i(\widetilde\gamma(t))v^j(\widetilde\gamma(t))&\text{since $[\theta_s]$ is orthogonal}
\end{align*}
It follows that $K(\widetilde{(\theta_s\circ\gamma)}(t))=K(\widetilde\gamma(t))$. Hence $L(\widetilde{(\theta_s\circ\gamma)}(t))=L(\gamma(t))$ showing that each one parameter subgroup of $SO(3)$ is a continuous symmetry on natural Lagrangian systems under a central force. This statement easily generalizes to $\R^n$. It also generalizes to a system of $k$ particles in $\R^n$ by considering the manifold $\R^{nk}$ so that the motion of the $k$ particles is described by one curve.
\end{example}

\begin{remark} In a natural Lagrangian system for which the potential energy is zero, every element of the Gallilean group SGal$(3)$ corresponds to a continuous symmetry. In the cases where the Euler-Lagrange equations reduce to Newton's second law, saying that the elements of SGal$(3)$ are continuous symmetries is equivalent to the statement that Newton's laws are invariant under the action of elements of SGal$(3)$.
\end{remark}

\begin{example} \bf{(Rotational Invariance Gives Conservation of Angular Momentum)}\rm
\vspace{0.1cm}

Example 6.5 showed that, under a central force, each one parameter subgroup of $SO(3)$ is a continuous symmetry. Hence by Noether's theorem, each has a corresponding conserved quantity.  For example consider rotation about the $x^3$-axis. That is, consider the one parameter family given by \[\theta_s=\left[\begin{array}{c c c} \cos s &-\sin s &0\\ \sin s&\cos s &0\\ 0&0& 1\end{array}\right].\]Noether's theorem shows that for any $t\in\R$ the corresponding conserved quantity is \[\left(\Phi_{L_{\gamma(t)}}(\gamma^\prime(t))\right)\cdot\left(\theta_\ast^{(\gamma(t))}\left(\left.\frac{d}{ds}\right|_{s=0}\right)\right).\]Computing, we get that the conserved quantity is
\begin{align*}&\left[m(\gamma^1)^\prime(t),m(\gamma^2)^\prime(t),m(\gamma^3)^\prime(t)\right]\cdot\left(\left[\begin{array}{c c c} 0 &-1 &0\\ 1&0 &0\\ 0&0& 0\\ \end{array}\right]\left[\begin{array}{c}(\gamma^1)^\prime(t)\\ (\gamma^2)^\prime(t)\\ (\gamma^3)^\prime(t)\end{array}\right]\right) \\&\! \\&=m(\gamma^2)^\prime(\gamma^1)^\prime-m(\gamma^1)^\prime(\gamma^2)^\prime\end{align*}This last value is precisely the angular momentum in the $x^3$ direction. That is, rotational symmetry about the $x^3$-axis has angular momentum in the $x^3$ direction as its corresponding conserved quantity.
\end{example}

\begin{example} \bf{(Translational Invariance Gives Conservation of Momentum)}\rm
\vspace{0.1cm}

Consider a closed system in $(\R^3, K-U)$ subject to a conservative force whose potential is independent of the $x^1$-coordinate. Notice that translation in the $x^1$-direction is given by the one parameter family $\{\theta_s:\R\to\R^3, (x^1,x^2,x^3)\mapsto (x^1+s,x^2,x^3)\}$. For a motion $\gamma$, it's clear that $\frac{d}{dt}(\theta_s\circ\gamma(t))=\gamma^\prime(t)$ and so \[\widetilde{(\theta_s\circ\gamma)}(t)=(\gamma^1(t)+s,\gamma^2(t),\gamma^3(t),(\gamma^1)^\prime(t),(\gamma^2)^\prime(t),(\gamma^3)^\prime(t)).\]By our hypothesis we have that $U(\theta_s(\gamma(t)))=U(\gamma(t))$. Hence $L(\widetilde{(\theta_s\circ\gamma)}(t))=L(\widetilde\gamma(t))$. By Noether's theorem we have that
\[\left(\frac{\pd L}{\pd v^i}(\widetilde{(\theta_s\circ\gamma)}(t))\right)\cdot\left(\left.\frac{d}{ds}\right|_{s=0}(\theta_{s,\ast}(\widetilde\gamma(t)))^i\right)\]  is a conserved quantity. But this is equal to
\[\left[m(\gamma^1)^\prime(t),m(\gamma^2)^\prime(t),m(\gamma^3)^\prime(t)\right]\cdot \left[\begin{array}{c}1\\0\\0\end{array}\right]=m(\gamma^1)^\prime(t).\]That is, momentum in the $x^1$-direction is the resulting conserved quantity of translation in the $x^1$-direction. This example easily generalizes to $\R^n$. Note also that in a system with $k$ particles in $\R^n$ interacting through conservative forces, we can replace our base space with $\R^{nk}$ and study the motion of $1$ particle.

\end{example}

\begin{example} \bf{(Time Invariances Gives Conservation of Energy)} \rm
\vspace{0.1cm}

We showed in proposition 5.17 that if the Lagrangian was time independent then the total energy was conserved. However, it is not obvious how to view time translation as a continuous symmetry. To see how this can be done, suppose first that in a Lagrangian system $(M,L)$ the Lagrangian is time dependent. That is, suppose $L\in C^\infty(T(M\times \R))$. Consider the Lagrangian system $(M\times \R, L_1)$ where $L_1\in C^\infty (T(M\times\R))$ is defined as follows. If $(x^1,\dots,x^n,v^1,\dots,v^n)$ are the induced coordinates on $TM$ and $(t,u)$ are the induced coordinates on $T\R$ then define $L_1$ by \[L_1:T(M\times\R) \to\R \ \ \ \ \ \ \ (t,x,u,v)\mapsto L(t,x,\frac{v}{u})u.\]Notice that if we have a motion $\gamma:\R\to M$, then the curve \[\hat\gamma:=\alpha:\R\to \R\times M \ \ \ \ \ \ \ \varepsilon\mapsto (\varepsilon,\gamma(\varepsilon))\]is a motion in $(\R\times M,L_1)$. This is because $\widetilde\alpha(\varepsilon)=(\varepsilon,\gamma(\varepsilon),1,\gamma^\prime(\varepsilon))$ and\[L_1(\varepsilon,\gamma(\varepsilon),1,\gamma^\prime(\varepsilon))=L(\varepsilon,\gamma(\varepsilon),\gamma^\prime(\varepsilon)).\] Hence it is justified to call $\gamma:\R\to M$ a motion in $(M,L)$ if $\hat\gamma=\alpha$ is a motion in $(\R\times M, L_1)$. Suppose that $\{\theta_s:\R\times M\to\R\times M\}$ is a continuous symmetry in $(\R\times M,L_1)$. Then by Noether's theorem we have that the following is a conserved quantity in $(\R\times M,L_1)$; \[\frac{\pd L_1}{\pd u}(t,x,u,v)\cdot \frac{d}{ds}\theta_s^0+\frac{\pd L_1}{\pd v^i}L(t,x,u,v)\cdot\frac{d}{ds}\theta_s^i\]which is \[L(t,x,\frac{v}{u})\frac{d}{ds}\theta_s^0-\frac{v^i}{u}\frac{\pd L}{\pd v^i}(t,x,\frac{v}{u})\frac{d}{ds}\theta_s^0+\frac{\pd L}{\pd v^i}(t,x,\frac{v}{u})\frac{d}{ds}\theta_s^i.\] If the Lagrangian is time dependent we define conserved quantities as follows. Given a motion $\gamma$ in $(M,L)$ we get the motion $\alpha=\hat\gamma$ in $(\R\times M,L_1)$ defined above. We observed that a continuous symmetry / conserved quantity in $(\R\times M,L_1)$ is also a continuous symmetries / conserved quantity in $(M,L)$ since $L_1=L$ when we set $u=1$. By Noether's theorem, a continuous symmetry $\{\theta_s\}$ in $(\R \times M,L_1)$ gives a conserved quantity in $(\R\times M,L_1)$ which in turn gives a conserved quantity, $G$, in $(M,L)$ by setting $u=1$. That is we have that $\{\theta_s\}$ is a continuous symmetry in $(M,L)$ with corresponding conserved quantity $G$. We can now view time translation as a continuous symmetry and compute its corresponding conserved quantity. It's clear that the continuous symmetry of $(\R\times M, L_1)$ given by $\{\theta_s:\R\times M\to \R\times M \ ,\ (t,x,u,v)\mapsto (t+s,x,u,v)\}$ is representing time translation. Now suppose that $\gamma:\R\to M$ is a motion in $(M,L)$ and $L$ is time independent. Consider the induced motion $\hat\gamma=\alpha:\R\to \R\times M$. We have that 
\begin{align*}
L_1(\widetilde{(\theta_s\circ\alpha)}(t)&=L_1(\varepsilon+s,\gamma(\varepsilon),1,\gamma^\prime(\varepsilon))\\
&=L(\varepsilon+s,\gamma(\varepsilon),\gamma^\prime(\varepsilon)))\\
&=L(\varepsilon,\gamma(\varepsilon),\gamma^\prime(\varepsilon))&\text{since $L$ is time independent}\\
&=L_1(\varepsilon,\gamma(\varepsilon),1,\gamma^\prime(\varepsilon))\\
&=L_1(\widetilde\alpha)
\end{align*}
But we have that $\frac{d}{ds}\theta_s^i=0$ for all $i\geq 1$. As well, $\frac{d}{ds}\theta_s^0=1$ and so the conserved quantity is just $L(t,x,\frac{v}{u})-\frac{v^i}{u}\frac{\pd L}{\pd v^i}(t,x,\frac{v}{u})$. By setting $u=1$, it follows that the conserved quantity in $(M,L)$ is $L(t,x,v)-v^i\frac{\pd L}{\pd v^i}(t,x,v)=-L^\ast$. But in a natural Lagrangian system, Proposition 5.15 showed that $L^\ast$ is the total energy.
\end{example}

\subsection{Noether's Theorem in Hamiltonian Mechanics}


Fix a Hamiltonian system $(X,\omega,H)$. We first define the notions of continuous symmetry and conserved quantity in the symplectic setting.

\begin{definition} A \bf{continuous symmetry} \rm is a vector field $W\in\Gamma(TX)$ such that \[\mathcal L_W\omega=0 \ \ \text{ and } \ \ \mathcal L_WH=0\] 
\end{definition}

\begin{definition} A \bf{conserved quantity (}\rm or \bf{constant of motion)} \rm is a function $f\in C^\infty(X)$ that Poisson commutes with $H$. By the antisymmetry of the Poisson bracket, this means that \[ \{f,H\}=0=\{H,f\}.\]
\end{definition}
Since $\{H,f\}=\omega(V_H,V_f)=(V_H\hk \omega)(V_f)=(dH)(V_f)=V_fH=\mathcal L_{V_f}H$, we have that $f$ Poisson commutes with $H$ if and only if \[ \mathcal L_{V_f}H=V_fH=0=V_Hf=\mathcal L_{V_H}f.\]

\begin{theorem} \bf{(Noether)} \rm If $W\in\Gamma(TX)$ is a continuous symmetry, then $W$ is locally Hamiltonian and its Hamiltonian function is a constant of the motion. Conversely, given a constant of the motion $f\in C^\infty(X)$, its Hamiltonian vector field, $V_f$, is a continuous symmetry.
\end{theorem}

\begin{proof} Let $W$ be a continuous symmetry. By hypothesis, \[0=\mathcal L_W\omega=d(W\hk \omega)+W\hk d\omega=d(W\hk\omega)\]Hence, by Poincare's lemma, around every point there exists a neighbourhood $U$ and a function $f\in C^\infty(U)$ such that $W\hk\omega=df$. That is, locally $W=V_f$ so that $W$ is locally Hamiltonian. By hypothesis, $\mathcal L_WH=0=\mathcal L_{V_f}H$ so that $\{f,H\}=0=\{H,f\}$. Conversely, let $f\in C^\infty(X)$ be a conserved quantity so that $\{f,H\}=0=\{H,f\}$. Consider the corresponding vector field $V_f$. It was just shown that $\mathcal L_{V_f}H=0$ while \[\mathcal L_{V_f}\omega=d(V_f\hk\omega)+V_f\hk d\omega=d(V_f\hk\omega)=d(df)=0\]
\end{proof}

\begin{example}\bf{(Symmetries on the $2$-Torus)} \rm
\vspace{0.1cm}

Consider the $2$-torus $\TT^2$ with local coordinate chart $(U,\theta,\varphi)$. Here $U=A\times B$, where $(A,\theta)$ and $(B,\varphi)$ are two different local coordinate charts on $S^1$. Consider the Hamiltonian system $(U,\omega=d\theta\wedge d\varphi, H)$ where $H\in C^\infty (TU)$ is defined by $H(\theta,\varphi):=\theta$. It follows $dH=d\theta$ and so the Hamiltonian vector field is $V_H=-\frac{\pd}{\pd\varphi}$. Consider the vector field $W=\frac{\pd}{\pd\varphi}\in\Gamma(TU)$. The flow of $W$ is \[\theta_s:U\to U \ , \ (\theta,\varphi)\mapsto (\theta,\varphi+s).\] By Cartan's magic formula, we have that \[\mathcal L_{\frac{\pd}{\pd\varphi}}\omega=d(d\theta)+\frac{\pd}{\pd\varphi}\hk d\omega=0\] and \[\mathcal L_WH=\frac{\pd\theta}{\pd\varphi}=0.\] That is, $W$ is a continuous symmetry in this Hamiltonian system. Since $W\hk \ \omega=d\theta$ it follows that on $U$ we have $W=V_f$ where $f:U\to \R \ , \ (\theta,\varphi)\mapsto \theta$. That is, $f=H$ and hence $\{f,H\}=0$. That is, $f$ is the corresponding conserved quantity. Notice that $f=H$ is not a global function on $\TT^2$ and so we can only say that $W$ is locally Hamiltonian.
\end{example}

Noether's theorem in the Hamiltonian setting says that each continuous symmetry on a symplectic manifold is locally Hamiltonian. That is, if $W\in\Gamma(TX)$ is a continuous symmetry, then for each $p\in X$ there exists an open set $U_p$ and $f_p\in C^\infty(U_p)$ (the conserved quantity) such that $W=V_{f_p}$ on $U_p$. However, when the symplectic manifold is a cotangent bundle, the following proposition gives a condition on when the conserved quantity is global.
 
 \begin{proposition}
 Given a manifold $M$, we know that the cotangent bundle $(T^\ast M,\omega=-d\alpha)$ is a symplectic manifold, where $\omega$ is the tautological $2$-form. Fix a Hamiltonian function $H\in C^\infty(T^\ast M)$. Let $W\in\Gamma(T(T^\ast M))$ be a continuous symmetry in the Hamiltonian system $(T^\ast M,\omega,H)$ which preserves that tautological $1$-form $\alpha$. Then the corresponding conserved quantity (Hamiltonian function) is $\alpha(W)$. Note that this conserved quantity is globally defined. In other words, if $W\in\Gamma(T(T^\ast M))$ is a continuous symmetry which preserves $\alpha$, then \[W=V_{\alpha(W)}.\]
 \end{proposition}
 
 \begin{proof}
 Suppose that $W\in\Gamma(T(T^\ast M))$ is a continuous symmetry such that $\mathcal L_W\alpha=0$. By Cartan's magic formula this means that \[d(W\hk\alpha)=-W\hk d\alpha=W\hk\omega.\]
 \end{proof}
 
 \del
 {
  That is, $W=V_\alpha(W)$. Conversely, suppose that $f\in C^\infty(T^\ast M)$ is a conserved quantity of the form $f=\alpha(Y)$ for some $Y\in\Gamma(T(T^\ast M))$. Then the Hamiltonian vector field corresponding to $f$ is \[V_f=V_{\alpha(Y)}.\]By  Noether's theorem we have that $V_{\alpha(Y)}$ is a continuous symmetry. Furthermore, \[\mathcal L_{V_{\alpha(Y)}}\alpha=d(V_{\alpha(Y)}\hk\alpha)-(V_{\alpha(Y)}\hk \omega)=d(V_{\alpha(Y)}\hk\alpha)-d(\alpha(Y))=0.\]
 
}

\subsection{Noether's Theorem Under the Legendre Transform}

This subsection gives another example of how the Legendre transform translates statements between Lagrangian and Hamiltonian mechanics. That is, we show how the Legendre transform translates the statements of Noether's theorem.

\begin{theorem}Let $(M,L)$ be a Lagrangian system where $L\in C^\infty(TM)$ is strongly convex. Suppose that $\{\theta_s:M\to M, s\in\R\}$ is a continuous symmetry so that the corresponding conserved quantity is $F:=\left(\frac{\pd}{\pd v^i}L\right)\left(\frac{d}{ds}\circ(\theta_{s,\ast}^i)\right)$. This continuous symmetry generates a vector field $W\in\Gamma(TM)$. The claim is that the vector field $W_\sharp\in\Gamma(T(T^\ast M))$, as defined in Lemma 2.24, is a continuous symmetry in the induced Hamiltonian system $(T^\ast M,\omega, H=L^\ast)$ and the corresponding conserved quantity is $F\circ\Phi_L^{-1}=F\circ\Phi_{H}$. 
\vspace{0.1cm}

Conversely, given an arbitrary manifold $M$, consider the Hamiltonian system $(T^\ast M, \omega=-d\alpha, H)$ for some arbitrary strongly convex Hamiltonian $H\in C^\infty (T^\ast M)$. Suppose that $W\in\Gamma(T(T^\ast M))$ is a continuous symmetry which preserves $\alpha$. Let $G$ denote the corresponding conserved quantity. If $\{\psi_s:T^\ast M\to T^\ast M, s\in\R\}$ is the flow of $W$ then by Theorem 2.25 there exists a family of diffeomorphisms $\{\theta_s:M\to M\}$ such that $\theta_s^\sharp=\psi_s$. The claim is that $\{\theta_s:M\to M, s\in\R\}$ is a continuous symmetry in $(M,L:=H^\ast)$ and that the corresponding conserved quantity is $G\circ\Phi_H^{-1}=G\circ\Phi_L$.
\end{theorem}

\begin{proof} Given a continuous symmetry $\{\theta_s:M\to M \ , \ s\in\R\}$ in $(M,L)$, let $F=\left(\frac{\pd L}{\pd v^i}\right)\left(\frac{d}{ds}\circ\theta_{s,\ast}^i\right)$ be the corresponding conserved quantity. Let $W$ be the infinitesimal generator of $\{\theta_s\}$ and $W_\sharp$ denote its lift. We are trying to show that $W_\sharp$ is a continuous symmetry in $(T^\ast M,\omega,H)$. By Proposition 2.22 it follows \[\mathcal L_{W_\sharp}\alpha=0 \ \ \text{ and } \ \ \mathcal L_{W_\sharp}\omega=0\]It remains to show that $\mathcal L_{W_\sharp}H=0$. By Theorem 6.14 we know that $W_\sharp=V_{\alpha(W_\sharp)}$. However, $W_\sharp\hk\alpha=F\circ\Phi_L^{-1}=F\circ\Phi_H$ since for arbitrary $(p,\xi_p)\in T^\ast M$ by definition 
\begin{align*}
F\circ\Phi_L^{-1}(p,\xi_p)&=\left(\Phi_L(\Phi_L^{-1}(\xi_p))\right)\left(\theta^{(p)}_\ast\left(\left.\frac{d}{ds}\right|_{s=0}\right)\right)\\
&=\xi_p\left(W_p\right)\\
&=\xi_p\left(\pi_\ast((W_\sharp)_{(p,\xi_p)})\right)\\
&=\alpha(W_\sharp)(p,\xi_p)
\end{align*}
Thus, showing that $\mathcal L_{W_\sharp}H=0$ is equivalent to showing, by Proposition 4.12, that for any integral curve $\Psi$ of $V_H$ we have $\dt (F\circ\Phi_L^{-1})(\Psi(t))=0$. But by Theorem 5.17, any integral curve of $V_H$ is of the form $(\Phi_L\circ\widetilde\alpha)(t)$ for some motion $\alpha(t)$ in $(M,L)$. But then \[\dt F\circ\Phi_L^{-1}(\Phi_L(\alpha(t)))=\dt F(\alpha(t))=0\]since $F$ is a conserved quantity. All of this shows that $W_\sharp$ is a conserved quantity in $(T^\ast M,\omega, L^\ast)$ with globally defined conserved quantity $F\circ \Phi_{H}$.\\



Conversely, consider a Hamiltonian system of the form $(T^\ast M, \omega =-d\alpha, H)$ for some strongly convex $H\in C^\infty(T^\ast M)$. Let $W\in\Gamma(T(T^\ast M))$ be a continuous symmetry whose flow preserves $\alpha$. Then by Theorem 6.14 we have that the corresponding conserved quantity is $G=W\hk\alpha$. 
Furthermore, by Theorem 2.25 , if $\{\Psi_s:T^\ast M\to T^\ast M, \ s\in\R\}$ is the flow of $W$ then there exists a family $\{\theta_s:M\to M,\ s\in\R\}$ such that $\Psi_s=\theta_s^\sharp$. That is, each $\Psi_s$ is defined by \[\Psi_s:T^\ast M\to T^\ast M \ \ \ \ \ \ \ (p,\xi)\mapsto (\theta_s(p),(\theta_s^\ast)^{-1}(\xi)).\]\del{By hypothesis, we have that for all $s\in\R$ and $(x,\xi)\in T^\ast M$ \[H(x,\xi)=(\Psi_s^\ast H)(x,\xi)=H(\theta_s(x),(\theta_s^\ast)^{-1}(\xi)).\] But since this holds for all $s\in\R$ it follows that $\theta_s^\ast H=H$.} We want to show that the family $\{\theta_s:M\to M,\ s\in\R\}$ is a continuous symmetry in the induced Lagrangian setting $(M,L=H^\ast)$. Consider the vector field \[Y=\xi_i\frac{\pd}{\pd\xi_i}\in\Gamma(T(T^\ast M)),\] which is negative the symplectic dual of $\alpha$. By Claim 2.26 we have that \[\Psi_{s,\ast}Y=Y,\] while by hypothesis \[\Psi_s^\ast H=H.\]By definition
\begin{align*}L(x,v)&=H^\ast(x,v)\\
&=\xi_i\frac{\pd H}{\pd\xi_i}-H(x,\xi)\\
&=Y(H)-H(x,\xi).
\end{align*}Hence
\begin{align*}
(\Psi_s^\ast L)(x,v)&=\Psi_s^\ast(Y(H))-\Psi_s^\ast(H(x,\xi))\\
&=\Psi_s^\ast(Y(H))-H(x,\xi)
\end{align*}Thus it suffices to show that $\Psi_s^\ast(Y(H))=Y(H)$. Indeed, for arbitrary $p\in M$ we have that
\begin{align*}
(\Psi_s^\ast(Y(H)))_p&=((Y(H))\circ \Psi_s)(p)\\
&=Y_{\Psi_s(p)}H\\
&=(dH)_{\Psi_s(p)}(Y_{\Psi_s(p)})\\
&=(dH)_{\Psi_s(p)}((\Psi_{s,\ast}Y)_p)&\text{since $\Psi_{s,\ast}Y=Y$}\\
&=(\Psi_s^\ast(dH))_p(Y_p)\\
&=(d(\Psi_s^\ast H))_p(Y_p)\\
&=(dH)_p(Y_p)&\text{since $\Psi_s^\ast H=H$}\\
&=Y_p(H)
\end{align*}

\del
{ 
where this Lagrangian system is equipped with the induced coordinates $x^1,\dots,x^n,v^1,\dots,v^n$ where $v^i=\frac{\pd H}{\pd\xi_i}$. Notice that since $\theta_s:M\to M$ is a diffeomorphism, if we set $\widetilde x=\theta_s(x)$ and $\widetilde \xi=(\theta_{s,x}^\ast)^{-1}(\xi)$ we have that $(\widetilde x,\widetilde \xi)$ is another coordinate system on $T^\ast M$ and the chain rule shows that  \[\widetilde\xi_j=\frac{\pd x^i}{\pd \widetilde x^j}\xi_i\]and \[\frac{\pd \xi_j}{\pd\widetilde\xi_i}=\frac{\pd\widetilde x^i}{\pd x^j}.\] It follows that \[\xi_j\frac{\pd H}{\pd\xi_j}=\widetilde\xi_i\frac{\pd H}{\pd\widetilde\xi_i}.\]

}

But by hypothesis, $\Psi_s$ is the inverse of the pullback of $\theta_s$.  Thus we have shown that $(\widetilde\theta_s)^\ast L=L$ for all $s\in\R$. That is, $\{\theta_s:M\to M, \ s\in\R\}$ is a continuous symmetry in the Lagrangian system $(M,L)$. Noether's theorem shows that the corresponding conserved quantity is $F=\left(\frac{\pd L}{\pd v^i}\right)\left(\frac{d}{ds}\circ\theta_{s,\ast}^i\right)$. As in the proof of the converse, we have that $F\circ\Phi_H=\alpha(X)$. That is \[F=\alpha(X)\circ \Phi_L.\]

\del{

For $\{\theta_s\}$ to be a continuous symmetry in $(M,H^\ast)$ it needs to be shown that for a motion $\gamma:\R\to M$ \[L(\theta_{s,\ast}(\gamma(t),\gamma^\prime(t)))=L((\gamma(t),\gamma^\prime(t)))\]But since $\gamma$ is a motion the theorem shows that $\Phi_L\circ\widetilde\gamma$ is an integral curve for $X_H$. But by hypothesis $\mathcal L_XH=0$ so that each $\theta_s^\sharp=\Psi_s$ preserves integral curves. So in particular\[\theta_s^\sharp\left(\Phi_{L_{\gamma(t)}}(\gamma(t),\gamma^\prime(t))\right)=\Phi_{L_{\gamma(t)}}(\gamma(t),\gamma^\prime(t))\]We then get the following chain of implications 

\begin{align*}
\theta_s^\sharp\left(\Phi_{L_{\gamma(t)}}(\gamma(t),\gamma^\prime(t))\right)=\Phi_{L_{\gamma(t)}}(\gamma(t),\gamma^\prime(t))&\Longrightarrow \pi\circ\theta_s^\sharp\left(\Phi_{L_{\gamma(t)}}(\gamma(t),\gamma^\prime(t))\right)=\pi\circ\Phi_{L_{\gamma(t)}}(\gamma(t),\gamma^\prime(t))\\
&\Longrightarrow \pi\circ\theta_s^\sharp\left(\Phi_{L_{\gamma(t)}}(\gamma(t),\gamma^\prime(t))\right)=\gamma(t)\\
&\Longrightarrow \theta_s\circ\pi\left(\Phi_{L_{\gamma(t)}}(\gamma(t),\gamma^\prime(t))\right)=\gamma(t)\\
&\Longrightarrow\theta_s(\gamma(t))=\gamma(t)\\
&\Longrightarrow\theta_{s,\ast}(\gamma^\prime(t))=\gamma^\prime(t)
\end{align*}It thus follows that $\{\theta_s:M\to M,s\in\R\}$ is a continuous symmetry in $(M,H^\ast)$. Noether's theorem says that the corresponding conserved quantity is $F:=\left(\frac{\pd L}{\pd y^i}\right)\left(\frac{d}{ds}\theta_s^i\right)$. But from the proof of the converse, we know that $F\circ\Phi_L^{-1}=G$ and so $F=G\circ\Phi_L=G\circ\Phi_{H^\ast}=G\circ\Phi_H^{-1}$. That is, the corresponding conserved quantity in $(M,H^\ast)$ is $G\circ\Phi_H^{-1}$.}

\end{proof}

\begin{remark}
In summary we have shown that, given a continuous symmetry in $(M,L)$ with corresponding conserved quantity $F$, there is a corresponding continuous symmetry in the Hamiltonian system $(T^\ast M,\omega,H=L^\ast)$ which preserves $\alpha$ and has conserved quantity $F\circ\Phi_H$. Conversely, given a continuous symmetry in a Hamiltonian system of the form $(T^\ast M,\omega,H)$ which preserves $\alpha$ and has corresponding conserved quantity $G$, there is a a corresponding continuous symmetry in $(M,L=H^\ast)$ with conserved quantity $G\circ \Phi_L$. In subsection 6.5 we will consider what happens when we relax the definitions of continuous symmetry and conserved quantities.
\end{remark}




\del{

\begin{remark} This correspondence of Noether's theorem in Hamiltonian and Lagrangian mechanics is indeed one-to-one. That is \[\{\theta_s:M\to M,s\in\R\}\mapsto X\mapsto X_\sharp\mapsto\{\theta_s^\sharp:T^\ast M\to T^\ast M\}\mapsto\{\theta_s:M\to M,s\in\R\}\]and\[F=\left(\frac{\pd L}{\pd y^i}\right)\left(\frac{d}{ds}\theta_s^i\right)\mapsto F\circ\Phi_L^{-1}\mapsto F\circ\Phi_L^{-1}\circ\Phi_H^{-1}=F\circ\Phi_L^{-1}\circ\Phi_{H^\ast}=F\circ\Phi_L^{-1}\circ\Phi_L=F\]\\

\noindent On the other hand \[\Gamma(T(T^\ast M))\owns Y\mapsto \{\Psi_s:T^\ast M\to T^\ast M\}\mapsto \{\theta_s:M\to M ; \text{ such that } \theta_s^\sharp=\Psi_s\}\mapsto\{\theta_s^\sharp:T^\ast M\to T^\ast M\}\mapsto Y\]and\[G:=Y\hk\alpha\mapsto G\circ\Phi_H^{-1}\mapsto G\circ\Phi_H^{-1}\circ\Phi_L^{-1}=G\circ\Phi_H^{-1}\circ\Phi_{L^\ast}=G\circ\Phi_H^{-1}\circ\Phi_H=G\]
\end{remark}


\begin{example}\bf{(Rotational Invariance and Conservation of Angular Momentum)}\rm
\vspace{0.1cm}

Should be easy, TOO LAZY ASK CAM!!

\end{example}
}

\begin{example}\bf{(Translational Invariance and Conservation of Momentum)}\rm
\vspace{0.1cm}

In example 6.8 we considered the natural Lagrangian system $(\R^3,K-U)$, where $U$ was assumed to be independent of $x^1$, with the continuous symmetry \[\{\theta_s:\R^3\to\R^3 \ , \ (x^1,x^2,x^3)\mapsto (x^1+s,x^2,x^3)\}.\]We saw that as a consequence of this symmetry we got conservation of momentum. We can also see this by using Theorem 6.15, and converting to the Hamiltonian setting. Indeed our continuous symmetry generates the vector field $\frac{\pd}{\pd x^1}\in\Gamma(T\R^3)$. Since the Jacobian of $\theta_s$ is the identity matrix we have that $\left(\frac{\pd}{\pd x^1}\right)_\sharp$ is just  $\frac{\pd}{\pd x^1}\in\Gamma(T(T^\ast\R^3))$. By Theorem 6.15,, it follows that $\left(\frac{\pd}{\pd x^1}\right)_\sharp$ is a continuous symmetry with conserved quantity $\alpha\left(\left(\frac{\pd}{\pd x^1}\right)_\sharp\right)=\xi_1$. But $\xi_1$ is just $mv^1$, which when applied to a motion is the momentum in the $x^1$-direction. 
 \end{example}

\subsection{The Converse of Noether's Theorem in the Lagrangian Setting}

In subsection 6.1 we showed that given a Lagrangian system $(M,L)$ and a continuous symmetry $\{\theta_s:M\to M, \ s\in\R\}$, Noether's theorem gave the corresponding conserved quantity \[\left(\frac{\pd L}{\pd v^i}\right)\cdot\left(\frac{d}{ds}\circ\theta_{s,\ast}^i\right).\]With the results from the previous subsection we can now show that given a conserved quantity of this form, the corresponding family $\{\theta_s:M\to M, \ s\in\R\}$ is a continuous symmetry in $(M,L)$. Indeed, suppose that we have a family of diffeomorphisms $\{\theta_s:M\to M, \ s\in\R\}$ such that \[\dt\left(\left(\frac{\pd L}{\pd v^i}\right)\cdot\left(\frac{d}{ds}\circ\theta_{s,\ast}^i\right)\right)=0.\] This family of diffeomorphisms generates a vector field $Y\in\Gamma(TM)$. In Theorem 6.15, we showed that the vector field $Y_\sharp\in\Gamma(T(T^\ast M))$ is a continuous symmetry in $(T^\ast M,\omega, H=L^\ast)$ which preserves $\alpha$ and has conserved quantity $\alpha(Y_\sharp)$. However, by Lemma 2.24 the flow of $Y_\sharp$ is $\theta_{s,\sharp}$. It follows from the proof of Theorem 6.15 that $\{\theta_s:M\to M, \ s\in\R\}$ is a continuous symmetry in $(M,H^\ast)=(M,L)$. Hence, we have the following diagram, where within each brace we are considering the induced Hamiltonian or Lagrangian system:

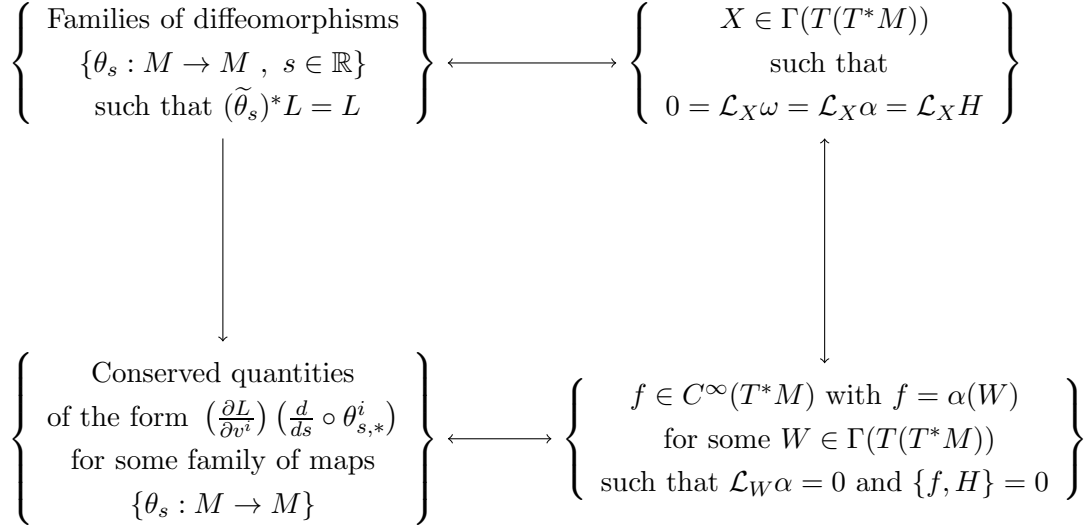
\[
\begin{tikzpicture}
\node (x1) at (0,0) 
{$\left\{\begin{array}{c}\text{Conserved quantities}\\ \text{of the form }\left(\frac{\pd L}{\pd v^i}\right)\left(\frac{d}{ds}\circ\theta_{s,\ast}^i\right)\\ \text{for some family of maps} \\\{\theta_s:M\to M\}\end{array}\right\}$};
\node (m1) at (0,5) 
 {$\left\{\begin{array}{c}\text{Families of diffeomorphisms}\\ \{\theta_s:M\to M \ , \ s\in\R\} \\ \text{ such that }(\widetilde\theta_s)^\ast L=L\end{array}\right\}$};
\node (x2) at (8,0) 
{$\left\{\begin{array}{c}f\in C^\infty(T^\ast M)\text{ with }f=\alpha(W)\\ \text{ for some } W\in\Gamma(T(T^\ast M))\\ \text{such that } \mathcal L_W\alpha=0\text{ and }\{f,H\}=0\end{array}\right\}$};
\node (m2) at (8,5) 
{$\left\{ \begin{array}{c}X\in\Gamma(T(T^\ast M))\\ \text{ such that}\\ 0=\mathcal L_X\omega=\mathcal L_X\alpha=\mathcal L_XH \end{array}\right\}$};
\draw[->] (m1) to node[left] {}  (x1);
\draw[<->] (m1) to node[above] {}  (m2);
\draw[<->] (x1) to node[below] {}  (x2);
\draw[<->] (m2) to node[right] {}  (x2);
\end{tikzpicture}
\]

This diagram proves the converse of Noether's theorem in the Lagrangian setting when the continuous symmetries are restricted to be of the above form.
Using this diagram we can exhibit a conserved quantity in a Lagrangian system which does not arise via Noether's theorem from a continuous symmetry.

\begin{example}\bf{(The Laplace-Runge-Lenz Vector)}\rm
\vspace{0.1cm}

Consider a particle of mass $m$ moving under a central force field in $\R^3$. Let $\vec r$ denote the position vector of this particle. For simplicity, endow $\R^3$ with the standard metric and let $(x^1,x^2,x^3)$ denote the standard coordinates on $\R^3$ and $(x^1,x^2,x^3,v^1,v^2,v^3)$ and $(x^1,x^2,x^3,\xi_1,\xi_2,\xi_3)$ the induced coordinates on $T\R^3$ and $T^\ast\R^3$ respectively. Let $L=K-U$ be the natural Lagrangian. As is the case for the gravitational and electrostatic forces, we assume our potential energy is of the form $U=-\frac{k}{\vec r}$ where $k$ is some constant. As in the proof of Proposition 5.14 we have that locally $\xi_i=mv^i$, the momentum in the $i$th direction of the particle. That is, $\xi_i=p_i$ as functions on $\R^6=T^\ast\R^3=T\R^3$. In these coordinates we have that the tautological one form is $\alpha=p_idx^i$. The \bf{Laplace-Runge-Lenz} \rm vector is defined to be \[\vec A=\vec p\times \vec L-\frac{mk\vec r}{\|\vec r\|}\]where $\vec p$ is the particle's momentum and $\vec L$ is the particle's angular momentum. By Proposition $3.11$ we have that the $\dt \vec L=0$. Identifying the vector field $\frac{\pd}{\pd r}$ with $\frac{\vec r}{\|\vec r\|}$ we have that \[ma=m\ddot{\vec r}=F=-\nabla U=-\frac{\pd U}{\pd r}\frac{\pd}{\pd r}=-\frac{k}{\|r\|^3}\vec r.\]It follows that

\begin{align*}
\frac{d\vec A}{dt}&=\dot{\vec p}\times \vec L+\vec p\times \dot{\vec L}-\dt\left(\frac{mk}{\|\vec r\|}\vec r\right)\\
&=\dot{\vec p}\times \vec L-\dt\left(\frac{mk}{\|\vec r\|}\vec r\right)\\
&=-\frac{mk}{\|\vec r\|^3}\left(\vec r\times (\vec r\times \dot{\vec r})\right)-\dt\left(\frac{mk}{\|\vec r\|}\vec r\right)\\
&=-\frac{mk}{\|\vec r\|^3}\left((\vec r\cdot\dot{\vec r})\vec r-\|\vec r\|^2\dot{\vec r}\right) -\dt\left(\frac{mk}{\|\vec r\|}\vec r\right)\\
&=-\frac{mk}{\|\vec r\|^3}\left(\frac{1}{2}\dt\left(\vec r\cdot\vec r\right)\vec r-\|\vec r\|^2\dot{\vec r}\right) -\dt\left(\frac{mk}{\|\vec r\|}\vec r\right)\\
&=-\frac{mk}{\|\vec r\|^3}\left(\frac{1}{2}\dt\left(\|\vec r\|^2\right)\vec r-\|\vec r\|^2\dot{\vec r}\right) -\dt\left(\frac{mk}{\|\vec r\|}\vec r\right)\\
&=-\frac{mk}{\|\vec r\|^3}\left(\|\vec r\|\left(\dt\|\vec r\|\right)\vec r-\|\vec r\|^2\dot{\vec r}\right) -\dt\left(\frac{mk}{\|\vec r\|}\vec r\right)\\
&=mk\left(\frac{\dot{\vec r}}{\|\vec r\|}-\frac{\left(\dt\|\vec r\|\right)}{\|\vec r\|^2}\vec r\right)-\dt\left(\frac{mk}{\|\vec r\|}\vec r\right)\\
&=\dt\left(\frac{mk}{\|\vec r\|}\vec r\right)-\dt\left(\frac{mk}{\|\vec r\|}\vec r\right)\\
&=0
\end{align*}
Hence the first component of $\vec A$, which is \[A^1=(\vec p\times \vec L)^1-\frac{mk x^1}{\|\vec r\|}=p^2L^3-p^3L^2-\frac{mkx^1}{\|r\|}\]is a conserved quantity. That is, $A^1\in C^\infty(\R^6)$ is such that $\dt A^1=0$ on motions in the Hamiltonian system. It follows, from Theorem 6.15 that $A^1\circ \Phi_L$ is a conserved quantity in the Lagrangian system $(\R^3,K-U)$. In order to show that this conserved quantity does not arise from a continuous symmetry, by the above diagram we need to show that the induced continuous symmetry in the Hamiltonian system $(T^\ast\R^3, K+U)$ corresponding to the conserved quantity $A^1\circ\Phi_L\circ\Phi_H=A^1$ has flow which is not the lift of curves on the base manifold.
\vspace{0.1cm}

By definition, our Hamiltonian system is $(T^\ast \R^3, dx^i\wedge dp_i, L^\ast)$, where by the argument in the proof of Proposition 5.15 we have that \[K+U=H:=L^\ast=\frac{\|p\|^2}{2m}-\frac{k}{\|r\|}.\]Note that by definition 
\begin{align*}
X_{A^1}&=\frac{\pd A^1}{\pd p_i}\frac{\pd}{\pd x^i}-\frac{\pd A^1}{\pd x^i}\frac{\pd}{\pd p_i}\\
&=L^3\frac{\pd}{\pd x^2}-L^2\frac{\pd}{\pd x^3}-\frac{mk}{r}\frac{\pd}{\pd p_1}+\frac{mkx^1}{\|\vec r\|^2}\frac{\pd r}{\pd x^i}\frac{\pd}{\pd p_i}\\
&=L^3\frac{\pd}{\pd x^2}-L^2\frac{\pd}{\pd x^3}+\frac{mk}{\|\vec r\|^3}(-(x^2)^2-(x^3)^2)\frac{\pd}{\pd p_1}+\frac{mk}{\|\vec r\|^3}x^1x^2\frac{\pd}{\pd p_2}+\frac{mk}{\|\vec r\|^3}x^1x^3\frac{\pd}{\pd p_3}
\end{align*}
A straightforward calculation shows that $X_{A^1}H=0$ so that $A^1$ is a conserved quantity in the induced Hamiltonian setting.
\vspace{0.1cm}

\del{
A straightforward computation shows that \[X_{A^1}H=0.\]That is, $\{A^1,H\}=0$ so that $A^1$ is a conserved quantity in the Hamiltonian system $(T^\ast M,\omega,H)$.}

To find the flow, or integral curves $\gamma(t)=(x(t),\xi(t))$, of the continuous symmetry $X_{A^1}$ we need to solve the Hamilton equations:
\begin{align*}
(x^1)^\prime &= 0\\
(x^2)^\prime &= L^3\\
(x^3)^\prime&=-L^2\\
(p_1)^\prime&=-\frac{mk}{\|\vec r\|^3}(-(x^2)^2-(x^3)^2)\\
(p_2)^\prime&=\frac{mk}{\|\vec r\|^3} x^1 x^2\\
(p_3)^\prime&=\frac{mk}{\|\vec r\|^3}x^1x^3
\end{align*}

All we are trying to show is that the flow of $X_{A^1}$ is not the lift of a one parameter family of diffeomorphisms on $M$. But by definition, given a continuous symmetry $\theta_s:M\to M$, the lift $\theta_{s,\sharp}$ is equal to $(\theta_s^\ast)^{-1}$, which is a linear function on each fibre of the cotangent bundle . In particular then, if the flow of $X_{A^1}$ came from lifting curves on $M$ it would be that each $(p_k)^\prime$, for $k=1,2,3,$ were a linear function of $p_1,p_2$ and $p_3$. However, we can see immediately from the form of the above ODE's that this is not the case. Thus the flow of $X_{A^1}$ is not the lift of a continuous symmetry on $\R^3$. We thus have shown that $A^1\circ\Phi_L$ is a conserved quantity in the Lagrangian system $(\R^3,K-U)$ which does not come from a continuous symmetry.

\begin{remark}
An equivalent way to see this would be to show that $\mathcal L_{X_{A^1}}\alpha\not=0$. This can be done explicitly, but the calculation is quite lengthy. Note that although $\mathcal L_{X_{A^1}}\alpha\not=0$, we showed that $X_{A^1}$ is a continuous symmetry in $(T^\ast\R^3,\omega,K+U)$ so that $\mathcal L_{X_{A^1}}\omega=0$. That is, not all continuous symmetries in Hamiltonian systems of the form $(T^\ast M,\omega,H)$ need to preserve the tautological $1$-form.
\end{remark}

\del
{
To see this, we do not need to solve the ODE completely. Indeed, a computation shows that if $a=x^1(0)\ ,\  b=x^2(0)-L^3$ and $c=x^3(0)+L^2$ we get that \[p_1(t)= \frac{\log\left(\sqrt{a^2+b^2}\sqrt{t^2(a^2+b^2)+c^2}+a^2t+b^2t\right)}{\sqrt{a^2+b^2}}-\frac{t}{\sqrt{t^2(a^2+b^2)+c^2}}+p_1(0).\]It follows that if $\Psi_t$ is the flow of $X_{A^1}$ it cannot be that $\Psi_t=\theta_{t,\sharp}$ for some $\theta_t:M\to M$. Indeed, if it were a lift, then $(p_1)^\prime(t)$ would be a linear function of $p_1(t) \ , \ p_2(t)$ and $p_3(t)$ since the pullback of a diffeomorphism is linear. However, it is clear that this is not the case. It follows from the above diagram that $A^1$ cannot come from a continuous symmetry. Indeed, if the family $\{\theta_s:M\to M\}$ gave this corresponding symmetry, then letting $X\in\Gamma(TM)$ denote the vector field with flow $\theta_s$, it would be that $X_\sharp=X_{A^1}$ contradicting what we just argued.

}

\subsection{Relaxing the Definitions of Symmetries and Conserved Quantities}
\vspace{0.1cm}

With the Laplace-Runge-Lenz vector in mind, it is an interesting question to consider what would happen if we relaxed the definitions of continuous symmetry and conserved quantity. In particular, is there a way to make the following diagram traceable both clockwise and counterclockwise?
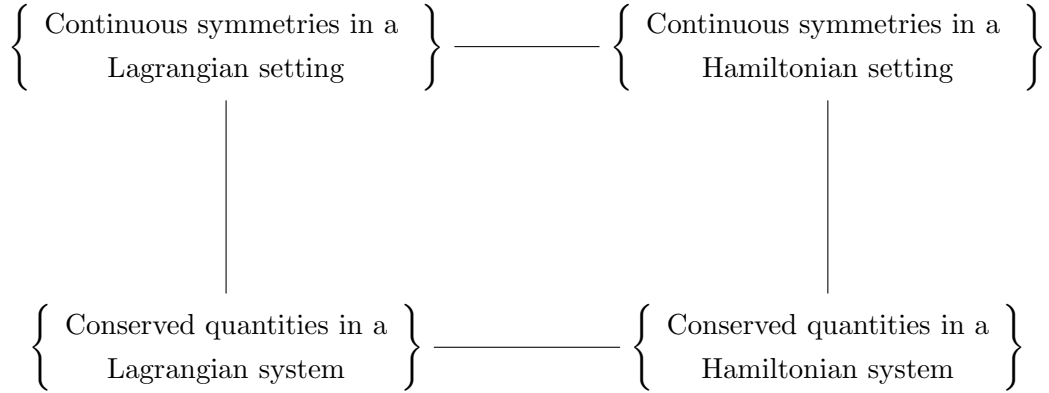
\[
\begin{tikzpicture}
\node (x1) at (0,0) 
{$\left\{\begin{array}{c}\text{Conserved quantities in a}\\ \text{Lagrangian system}\end{array}\right\}$};
\node (m1) at (0,4) 
 {$\left\{\begin{array}{c}\text{Continuous symmetries in a }\\ \text{ Lagrangian setting }\end{array}\right\}$};
\node (x2) at (8,0) 
{$\left\{\begin{array}{c}\text{Conserved quantities in a}\\ \text{Hamiltonian system}\end{array}\right\}$};
\node (m2) at (8,4) 
{$\left\{ \begin{array}{c}\text{Continuous symmetries in a }\\ \text{ Hamiltonian setting }\end{array}\right\}$};
\draw[-] (m1) to node[left] {}  (x1);
\draw[-] (m1) to node[above] {}  (m2);
\draw[-] (x1) to node[below] {}  (x2);
\draw[-] (m2) to node[right] {}  (x2);
\end{tikzpicture}
\]
We can see right away that with our definitions this is impossible. Indeed, any continuous symmetry in a Lagrangian system $(M,L)$ gives a continuous symmetry in $(T^\ast M,\omega,L)$ which preserves $\alpha$. However, we showed above that the Laplace-Runge-Lenz vector is a continuous symmetry which does not preserve $\alpha$. \\

Recall that the original definition of a continuous symmetry in the Hamiltonian setting does not include the requirement of preserving the tautological $1$-form. However, even without this requirement we saw that, in the Hamiltonian setting, continuous symmetries and locally defined conserved quantities are in one-to-one correspondence. If we return to the Lagrangian setting with the original definitions of symmetries and conserved quantity, we see that the two notions are not in one-to-one correspondence. Indeed, we showed in the previous subsection that the Laplace-Runge-Lenz vector is a conserved quantity, not of the form $\left(\frac{\pd L}{\pd v^i}\right)\left(\frac{d}{ds}\theta_{s,\ast}^i\right)$, which does not come from a continuous symmetry.\\

However, notice that the one-to-one correspondence of conserved quantities in the Lagrangian and Hamiltonian setting still holds. This is because if $F\in C^\infty(TM)$ is constant on the motions in $(M,L)$ then we have, by Theorem 5.19, that \[\dt (F\circ\Phi_H)(\Psi(t))=0\] for all motions $\Psi:\R\to T^\ast M$. That is, by Proposition 4.12, $F\circ\Phi_H$ is a conserved quantity in $(T^\ast M,\omega,H=L^\ast)$. Conversely, if $G\in C^\infty(T^\ast M)$ is such that $\{G,H\}=0$ then $\dt\left(G(\Psi(t))\right)=0$ for all integral curves $\Psi$ of $X_H$, then by Theorem 5.20 it follows that \[\dt\left(G\circ\Phi_L(\widetilde\gamma(t))\right)=0\] for all motions $\gamma:\R\to M$ in $(M,L)$. That is we have the one-to-one correspondence
\[
\begin{tikzpicture}
\node (x1) at (0,0)
{$\left\{\begin{array}{c}\text{Conserved quantities in a }\\ \text{ Lagrangian setting }\end{array}\right\}$};
\node(m2) at (9,0)
{$\left\{ \begin{array}{c}\text{Conserved quantities in a }\\ \text{ Hamiltonian setting }\end{array}\right\}.$};
\draw[<->] (x1) to node[left]{} (m2);
\end{tikzpicture}
\]

However, in order to reconcile the one-to-one correspondence between symmetries one needs to study the notion of `generalized symmetries'. A thorough treatment of this topic can be found in chapter 5 of \cite{Olver}. Roughly speaking,  these are a one-parameter family of `Lagrangian preserving' maps in a Lagrangian system $(M,L)$ which do not necessarily arise from lifting curves on $M$. That is, a generalized symmetry can be thought of as a one-parameter family $\{\Upsilon_s:TM\to TM \ , \ s\in\R\}$ such that $\Upsilon_s^\ast L=L$. Notice that with this new definition, Noether's theorem still holds as stated in the Lagrangian setting. The proof of Theorem 6.3 did not use the fact that we were lifting each $\theta_s$ to the tangent bundle. If one replaces $\widetilde\theta_s$ with $\Upsilon_s$, the proof of Theorem 6.3 is unchanged. It is likely that with the notion of generalized symmetries, one can show that all arrows in the following diagram go both ways. We have filled in the correspondences discussed in this paper.

\[
\begin{tikzpicture}
\node (x1) at (0,0) 
{$\left\{\begin{array}{c}\text{Conserved quantities}\\ F\in C^\infty (TM)\\  \text{ such that } \dt F=0\end{array}\right\}$};
\node (m1) at (0,5) 
 {$\left\{\begin{array}{c}\text{Families of diffeomorphisms}\\ \{\Upsilon_s:TM\to TM \ , \ s\in\R\} \\ \text{ such that }(\Upsilon_s)^\ast L=L\end{array}\right\}$};
\node (x2) at (8,0) 
{$\left\{\begin{array}{c}f\in C^\infty(T^\ast M)\\ \text{ such that }\\ \text{locally } \{f,H\}=0 \end{array}\right\}$};
\node (m2) at (8,5) 
{$\left\{ \begin{array}{c}X\in\Gamma(T(T^\ast M))\\ \text{ such that}\\ \mathcal L_X\omega=0=\mathcal L_XH \end{array}\right\}$};
\draw[->] (m1) to node[left] {}  (x1);
\draw[-] (m1) to node[above] {}  (m2);
\draw[<->] (x1) to node[below] {}  (x2);
\draw[<->] (m2) to node[right] {}  (x2);
\end{tikzpicture}
\]

\end{example}

\del{

The converse of Noether's theorem in the Lagrangian setting is not always true. That is, for a conserved quantity $F\in C^\infty (TM)$ there does not always exist a continuous symmetry $\{\theta_s:M\to M, s\in\R\}$.\\

FOR EXAMPLE!!\\

Let $F\in C^\infty(TM)$ be a conserved quantity and consider $F\circ \Phi_L^{-1}=F\circ\Phi_{L^\ast}$. Consider the Hamiltonian system $(T^\ast M,\omega, H:=L^\ast)$. By the theorem, for any integral curve $\Psi$ of $X_H$, we have that \[\dt F\circ\Phi_{L^\ast}(\Psi(t))=0\]But this means that $\left\{F\circ\Phi_{L^\ast},\Phi_{L^\ast}\right\}=0$ showing that $F\circ\Phi_{L^\ast}$ is a conserved quantity. Let $X\in\Gamma(T(T^\ast M))$ denote the induced vector field. Since $F\circ\Phi_{L^\ast}=X\hk\alpha$...\\

 DOES IT??  NO, THIS ONLY HOLDS IF $\alpha$ PRESERVED. CAN'T CONCLUDE  FROM THIS THAT FLOW OF $X$ IS THE LIFT OF CURVES, AND THAT THESE CURVES ARE A SYMMETRY.

}


\newpage

\begin{thebibliography}{99}
\vspace{0.2cm}

\bibitem{Da Silva} Ana Cannas da Silva.
{\em Lectures on Symplectic Geometry, } Springer-Verlag, 2001.

\bibitem{Arnold} V.I. Arnold.
{\em Mathematical Methods of Classical Mechanics, } 2nd edition, Springer-Verlag, 1989

\bibitem{Goldstein}Herbert Goldstein. 
{\em Classical Mechanics, } 2nd edition, Addison-Wesley Publishing Company, 1981.

\bibitem{Lee} John M. Lee.
{\em Introduction to Smooth Manifolds, } 2nd edition, Springer, 2000.

\bibitem{Olver} Peter J. Olver.
{\em Applications of Lie Groups to Differential Equations, }2nd edition, Springer-Verlag, 1986.

\bibitem{Taylor} John R. Taylor.
{\em Classical Mechanics,} University Science Books, 2005.

\bibitem{Tu} Loring W. Tu.
{\em An Introduction to Manifolds, } 2nd edition, Springer, 2011.

\bibitem{Abraham} Ralph Abraham. 
{\em Foundations of Mechanics, } W.A. Benjamin, 1967.

\bibitem{Olver}Peter J. Olver.
{\em Applications of Lie Groups to Differential Equations, }2nd edition, Springer-Verlag, 1986.

\bibitem{Lee2}John M. Lee.
{\em Riemannian Manifolds: An Introduction to Curvature, } Springer-Verlag, 1997.

\bibitem{Tong} David Tong.
{\em Lectures on Classical Dynamics, } http://www.damtp.cam.ac.uk/user/tong/dynamics.htm, University of Cambridge, 2012.

\bibitem{Butterfield} Jeremy Butterfield.
{\em On Symmetry and Conserved Quantities in Classical Mechanics, }University of Western Ontario Series in Philosophy of Science, 2006

\bibitem{Bell}Jordan Bell.
{\em The Legendre Transform, }http://individual.utoronto.ca/jordanbell/notes/legendre.pdf, University of Toronto, 2014

\bibitem{Linear Algebra} Stephen H. Friedberg, Arnold J. Insel, Lawrence E. Spence. 
{\em Linear Algebra, } 4th edition, Prentice Hall, 2003.



\end{thebibliography}
\end{document}